 \documentclass[journal,twoside]{IEEEtran}
\usepackage{multirow}
\usepackage{lipsum}
\usepackage{amsfonts}
\usepackage{array}
\usepackage{amsmath,cases,amssymb}
\usepackage{amsmath}
\usepackage{amsthm}
\usepackage{graphicx}
\usepackage{cite}
\usepackage{subfigure}
\usepackage{epsfig}
\usepackage{url}
\usepackage{algorithm}
\usepackage{algorithmic}
\usepackage{epstopdf}
\usepackage{balance}
\usepackage{color}
\usepackage{makecell}
\usepackage{sansmath}

\DeclareGraphicsExtensions{.eps,.pdf}

\newcommand{\maxi}{{\mathrm{maximize}}}

\newcommand{\non}{\nonumber}

\newcommand{\bw}{\mathbf{w}}

\newcommand{\bH}{\mathbf{H}}

\newcommand{\bI}{\mathbf{I}}
\newcommand{\bh}{\mathbf{h}}

\newcommand{\bW}{\mathbf{W}}

\newcommand{\Qcal}{\mathcal{Q}}
\newcommand{\Kcal}{\mathcal{K}}

\newcommand{\UEk}{\mathtt{UE}_k}

\makeatletter
\newcommand{\doublewidetilde}[1]{{%
		\mathpalette\double@widetilde{#1}}}
\newcommand{\double@widetilde}[2]{%
	\sbox\z@{$\m@th#1\widetilde{#2}$}%
	\ht\z@=.5\ht\z@
	\widetilde{\box\z@}}
\makeatother

\newcommand*{\hili}{\color{black}}

\newtheorem{theorem}{Theorem}
\newtheorem{lemma}{Lemma}
\newtheorem{corollary}{Corollary}

\newtheorem{remark}{Remark}
\newtheorem{assumption}{Assumption}
\begin{document}
	
	%\title{Demand-Based Problem Of Multi-Beam Satellite Communications Systems: Optimizations And Deep Learning Solutions}
	\title{\huge Robust Congestion Control for Demand-Based Optimization in Precoded Multi-Beam High Throughput Satellite Communications}
	
	\author{Van-Phuc~Bui, Trinh~Van~Chien, \textit{Member}, \textit{IEEE}, Eva~Lagunas, \textit{Senior Member}, \textit{IEEE}, Jo\"el~Grotz, Symeon~Chatzinotas, \textit{Senior Member}, \textit{IEEE}, and Bj\"orn~Ottersten, \textit{Fellow}, \textit{IEEE} %\vspace{-0.8cm}
		\thanks{Manuscript received xxx; revised xxx and xxx; accepted
		xxx. Date of publication xxx; date of
			current version xxx. This   work   was    supported   by  the   Luxembourg   National Research  Fund  (FNR)  under  the  project INtegrated Satellite - TeRrestrial Systems for Ubiquitous Beyond 5G CommunicaTions (INSTRUCT-FNR/IPBG19/14016225/INSTRUCT), partially supported by  the Luxembourg National Research Fund (FNR) project
			titled Dynamic Beam Forming and In-band Signalling for Next Generation
			Satellite Systems (DISBuS-FNR/BRIDGES19/IS/13778945/DISBuS),  and partially supported by the Luxembourg National Research Fund (FNR) under the project FlexSAT (C19/IS/13696663). Please note that the views of the authors of this paper do not necessarily reflect the views of ESA and/or SES. The parts of this paper have been accepted to present at the IEEE ASMS/SPSC~2022 \cite{Phuc2022}.  The associate editor coordinating the
			review of this article and approving it for publication was C. Jiang.
			\textit{(Corresponding author: Trinh Van Chien.)}}
		\thanks{V.-P. Bui was with  the Interdisciplinary Centre for Security, Reliability and Trust (SnT), University of Luxembourg, L-1855 Luxembourg, Luxembourg (email: phuc.bui@uni.lu).}
	\thanks{T. V. Chien is with the School of Information and Communication Technology (SoICT), Hanoi University of Science and Technology (HUST), 100000 Hanoi, Vietnam (email: chientv@soict.hust.edu.vn).}
	\thanks{E. Lagunas, S. Chatzinotas, and B. Ottersten are with the Interdisciplinary Centre for Security, Reliability and Trust (SnT), University of Luxembourg, L-1855 Luxembourg, Luxembourg (email: \{ eva.lagunas,symeon.chatzinotas, bjorn.ottersten\}@uni.lu).}
	\thanks{Jo\"el Grotz is with the  SES, Chateau de Betzdorf, Betzdorf 6815, Luxembourg (email: Joel.Grotz@ses.com).}
		
	}
	\maketitle
%	\vspace{-1.5cm}
	\begin{abstract}
		% Satellite communications with multiple spot beams is a promising technology for application to future networks since it enables to offer a high total sum rate to many users in a broadband coverage area. 
		\fontsize{10}{10}{High-throughput satellite communication systems are growing in strategic importance thanks to their role in delivering broadband services to mobile platforms and  residences and/or businesses in rural and remote regions globally. Although precoding has emerged as a prominent technique to meet ever-increasing user demands, there is a lack of studies dealing with congestion control. This paper enhances the performance of multi-beam high throughput geostationary satellite systems under congestion, where the users' quality of service (QoS) demands cannot be fully satisfied with limited resources. In particular, we propose congestion control strategies, relying on simple power control schemes. We formulate a multi-objective optimization framework balancing the system sum-rate and the number of users satisfying their QoS requirements. Next, we propose two novel approaches that effectively handle the proposed multi-objective optimization problem. The former is a model-based approach that relies on the weighted sum method to enrich the number of satisfied users by solving a series of the sum-rate optimization problems in an iterative manner. The latter is a data-driven approach that offers a low-cost solution by utilizing supervised learning and exploiting the optimization structures as continuous mappings. The proposed general framework is evaluated for different linear precoding techniques, for which the low computational complexity algorithms are designed. Numerical results manifest that our proposed framework effectively handles the congestion issue and brings superior improvements of rate satisfaction to many users than previous works. Furthermore, the proposed algorithms show low run-time and make them realistic for practical systems.} 
		% Nonetheless, the congestion issue appears when at least one user is served with a lower data throughput than requested. To cope with this issue,  / to avoid infeasibility appearing while solving the optimization problem
	\end{abstract}
%\vspace{-0.2cm}
	
	%\vspace{-0.25cm}
	%	\vspace{-0.4cm}
	%%%%%%%%%%%%%%%%%%%%%%%%%%%%%%%%%%%%%%%%%%%%%%%%
	\section{Introduction}\label{sec:intro}
	%%%%%%%%%%%%%%%%%%%%%%%%%%%%%%%%%%%%%%%%%%%%%%%%
	Multi-beam high throughput satellite (MB-HTS) systems have been acknowledged as an efficient solution providing ubiquitous high-speed broadband services to users in a large coverage area, especially for inaccessible or insufficiently covered places by current terrestrial networks \cite{Kodheli:21:tut}. Current broadband satellite communication systems make use of a multi-beam footprint, which boosts the frequency reuse improving spectral efficiency as well as system capacity \cite{Joroughi:17:twc,Lei:11,Choi:twc:05}. Due to low-cost and low-interference designs, an MB-HTS system may allocate limited radio resources uniformly across beams with the merits of simple procedures and inexpensive operating expenditure \cite{sat_equal}. Notwithstanding, the uniform resource allocation combined with the limited available spectrum may be inefficient in facing the rapid growth of traffic demands \cite{LagunasASMS2020, trinh2021user,Krivochiza:21:access}. In this context, full frequency reuse across satellite beams has stood up as a promising alternative boosting spectral efficiency and system capacity \cite{Vazquez:16, Shankar:21:comletter}. 
	
	There is a vast literature related to precoded MB-HTS, many of them including Quality of Service (QoS) constraints in terms of minimum Signal-to-Noise Ratio or minimum throughput per user \cite{ZhangAsilomar2020,QiWCL2020}. However, the uneven QoS requests pose a constant challenge to such works particularly for high QoS scenarios and limited satellite resources.
	% sophisticated scenarios with heterogeneous demands leading to the unbalancing loads among the  beams \cite{CADSAT, flepredem}. With the rapid growth of traffic demands, next satellite generations should consider digital payload incorporation with the software-defined resource allocation to allow more flexibility and network efficiency. Since a satellite resource is subject to severe physical, technological, and regulatory constraints, it must be efficiently exploited to fully utilize the potential benefits of the multi-beam operation and to meet the high  reliability,  and  low  latency requirements as the main theme of the future networks \cite{SES,trinh2021user, Kisseleff:20:comlet}. 
	% For high data throughput applications with different quality of services (QoSs), full frequency reuse among the beams is attractive because the achievable spectral efficiency per user is significantly improved with a large bandwidth by virtue of the Shannon capacity \cite{Vazquez:16, Shankar:21:comletter}. Once many users simultaneously access the network,  a proper precoding technique is in need for MB-HTS systems to mitigate mutual interference and enhance the signal strength over a long propagation distance, which is ideally designed based on instantaneous channel state information for a particular purpose \cite{Krivochiza:21:access}. 
	To maintain the individual QoS requirement of each user, the authors in \cite{Zheng:twc:12} formulated and solved a precoding design in a multi-beam satellite system by the use of an alternating optimization algorithm. Despite the data throughput improvement over the proposed iterative procedure, the solution in \cite{Zheng:twc:12} is not scalable since the max-min fairness optimization framework is not able to guarantee an acceptable QoS level for a large-scale system with many users. A precoding design targeting the system energy efficiency maximization is presented in \cite{Qi:comletter:20} under practical total power constraint and QoS requirements. Nevertheless, this framework requires time-consuming optimization, which greatly limits its applicability to real-world systems. Linear precoding \cite{bjornson2014}, e.g., zero-forcing (ZF) or regularized zero-forcing (RZF), has demonstrated good performance with low complexity in MB-HTS systems \cite{Christopoulos2015,Qi:comletter:18,trinh2021user}. However, the aforementioned works relied on non-empty feasible regions to make sure that the proposed optimization can reach a solution. For a complex system with significant number of users with divergent QoS requirements, there is an overwhelming probability that at least one user is in an extreme adverse channel condition or the requested QoS is too high under the limited radio resources. The existed solutions will, therefore, be unattainable due to congestion  resulting in an infeasible problem. %To the best of the authors' knowledge, 
	No known works have studied how to detect unsatisfied users and operate MB-HTS systems with a linear precoding technique under harsh optimization conditions, where the congestion appears. In this paper, we address this gap by formulating a multi-objective optimization framework balancing the system sum-rate and the number of users satisfying their QoS requirements. To solve this, we pursue two methodologies: (i) model-based approach, and (ii) data-driven approach. \textcolor{black}{While model-based methods are known to provide accurate solutions, data-driven approaches have shown to speed up the convergence towards close-to-optimal solutions \cite{Chien19:TCOM, sun2018learning} that are motivated by advances in machine learning as presented subsequently.}
	
	Machine learning has demonstrated its potential in constructing data-driven algorithms for engineering problems in signal processing and resource allocation via the use of neural networks \cite{Chien19:TCOM, goodfellow2016deep}. Rather than requesting humans to identify, formulate, and solve a system-level model as in traditional-based optimization theories, neural networks make efforts in wireless communications to learn the essential features of a data set, then use such information for predicting and decision making. One critical role is to design low complexity neural networks in which machine learning is applied for approximating high-cost optimization algorithms. In contrast to the maturity of machine learning developed for  terrestrial networks, learning-based approaches applied to satellite communications and performance evaluations are in their infancy \cite{vazquezMagazine2021}. To name a few,  the inherent NP-hard issues of different beam hopping optimization problems  were effectively handled with high accuracy in \cite{Lei:Access:20}. Moreover, channel allocation strategies under the viewpoints of mixed-integer programming were studied in \cite{Hucomletter18}, where authors exploited reinforcement learning to minimize the service blocking probability and enhance the data throughput. Regarding the power allocations,  the authors in \cite{Chen:cof:20} optimized the transmit power coefficients subject to the traffic demands for a multi-beam satellite network without considering precoding. Furthermore, the work in \cite{Luis:conf:19} proposed a deep learning model for power allocation with a simplified rate expression. We emphasize that these related works only studied single-objective optimization problems without raising concerns on the congestion controls that cannot be avoided in practical systems. For future MB-HTS systems, the applications of machine learning  for multi-objective signal processing optimization are promising to balance conflicting metrics and to ensure the individual QoS requirements with a tolerable computational complexity towards online resource allocation.
	
	{\hili The congestion problem was investigated and handled in \cite{1618923, 4567596, 1459062, 7497508} and references therein in the terrestrial networks. In particular, the authors in \cite{1618923} considered a primal-dual decomposition to determine and withdraw users interfering the most with other users until the remaining spectral efficiency demands can be satisfied. However, no power constraints were considered in \cite{1618923}. By using a limited power budget, a game-theoretic formulation of the power control issue was developed in \cite{4567596} to guarantee users' information rates. Also, a power allocation policy to decrease the requested throughput of users with poor channel conditions was proposed in \cite{1459062}. Besides, in \cite{7497508}, the congestion issue was handled by maximizing the minimum spectral efficiency of the users and neglecting the users' demands, which is a distinct issue that could result in none of the QoS requirements being met. Different solutions to handle the total energy minimization optimization problem under congestion was introduced in \cite{9531522}. Nonetheless, all these related works considered the congestion control by formulating single objective optimization problems and using a traditional model-based optimization theory to obtain the solution.
	To the best of the authors' knowledge, the transmit power allocation and the QoS satisfactions for the multi-objective optimization to tackle the joint maximization of both the sum rate and demand-based constraints subject to the limited power budget has never been considered before. This paper considers MB-HTS systems under multiple-access scenarios where many users with individual data throughout requirements share the same time and frequency resource. Congestion may appear for different reasons. For example, congestion may occur when one of the users has a sudden peak of demand (i.e. high QoS constraint), when their channel condition is not good, and/or when he is receiving too strong interference. We handle the congestion issue that appears when solving the sum data throughput maximization due to the practical aspects such as the weak channel conditions and limited power budget at the satellite.} Thanks to the European Space Agency (ESA) \cite{ESA}, the proposed algorithms are tested with a practical beam pattern. Our main contributions are summarized as follows:
	\begin{itemize}
		\item We formulate a new multi-objective optimization problem for the MB-HTS systems to maximize the number of users served satisfying their QoS requirements and the sum rate of the entire network. {Even though the problem is a non-smooth nonlinear program, it effectively handles the congestion issue by splitting the scheduled user set into the satisfied and unsatisfied user sets and combining both of them into the multi-objective optimization framework.} 
		\item We propose a general model-based solution that exploits the weighted sum method to transfer the original multi-objective problem to a single-objective maximization with a balance between the utility metrics.  Conditioned by the total transmit power limit, a heuristic algorithm iteratively solves the single-objective problem by prioritizing the number of satisfied users. This proposed algorithm then allocates the remaining power to maximize the sum rates. The generality of the model-based approach lets room for network operators to design a sum rate maximization solver.
		\item Next, we propose a general data-driven methodology where a neural network is used to predict the transmit power coefficients and satisfied-user set solutions with low computational complexity. It is achieved by exploiting supervised learning and based on the solution from the model-based approach. From a series of continuous mappings, the neural network only requires the channel gains as input. The generality of the data-driven approach is a consequence of the model-based approach and the network can opt for an arbitrary type of neural network architectures.
		\item By the convenience of the semi-closed form solution to the power allocation from the water-filling method, we typically design the low-cost algorithms for the MB-HTS systems by adapting the general model-based approach. The channel orthogonality can effectively contribute to reducing the computational complexity, even though the water filling method needs to be applied in an iterative manner. The power solutions can be effectively used for training fully connected neural networks.
		\item By using a practical satellite beam pattern provided by ESA, the performance of the proposed algorithms is evaluated by extensive numerical results. The solution is compared with the benchmarks \cite{lu2019robust,trinh2021user,Krivochiza:21:access} in the literature in terms of both sum rate and users' QoS satisfaction. Meanwhile, the neural network achieves the solution with high prediction accuracy in a few milliseconds.
	\end{itemize}
	
	\textit{Notation}: The upper and lower bold letters are used to denote the matrix and vectors, respectively. The notation $\mathcal{CN}(\cdot, \cdot)$ denotes the circularly symmetric Gaussian distribution and $\mathbb{E}\{ \cdot\}$ is the expectation operator. The notation $\| \cdot \|$ is the Euclidean norm and $| \mathcal{K} |$ is the cardinality of the set $\mathcal{K}$. The superscripts $(\cdot)^H$  and $(\cdot)^T$ are the Hermitian transpose and regular transpose, respectively. The element-wise inequality is denoted as $\succeq $. A unit vector of length $K$ is denoted as $\mathbf{1}_K$. The trace of a matrix is denoted as $\mathrm{tr}(\cdot)$. The complex, real, non-negative real,  extended non-negative real field is $\mathbb{C}$, $\mathbb{R}$, $\mathbb{R}_{+}$, and $\mathbb{R}_{++} = \mathbb{R}_{+} \cup \varnothing$, respectively. To the end, the imaginary unit of a complex number is $j$ with $\sqrt{j}=-1$.
	
	The rest of this paper is organized as follows: Section~\ref{sec:sys_model} presents in detail the satellite system model and formulates a category of multi-objective optimization problems jointly optimizing the sum rate and individual QoS per user. Section~\ref{Sec:ModelDL} describes the model-based and data-driven approaches to solve the above optimization problem in polynomial time.  The practical applications of our framework are demonstrated by the state-of-the-art practical communication satellite systems with a linear precoding technique and the water-filling method. % It should be noted that the results are directly applicable generically to different concepts of MB-HTS multi-beam satellite systems. 
	Section~\ref{Sec:Results} gives extensive numerical results% to demonstrate the effectiveness of the proposed solutions
	, while the main conclusions are finally drawn in Section~\ref{sec:conclusion}.
	% \vspace{-0.1cm}
	
	%\vspace{-0.25cm}
	\begin{figure*}[t]
		\begin{minipage}{0.5\textwidth}
			\centering
			\includegraphics[width=1\textwidth]{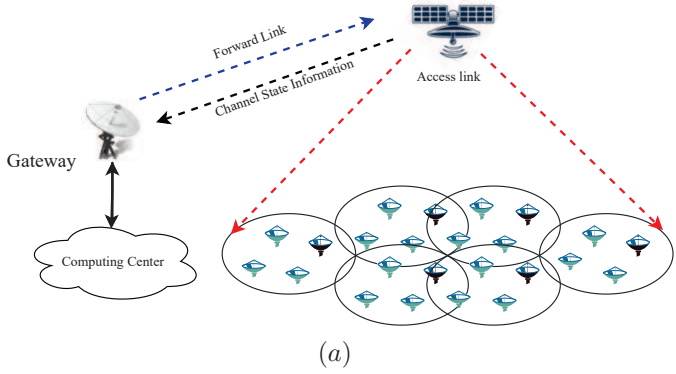}
			%\vspace{-0.5cm}
			\\
			$(a)$
			%\vspace{-0.3cm}
		\end{minipage}
		\hfill
		\begin{minipage}{0.45\textwidth}
			\centering
			\includegraphics[width=0.9\textwidth]{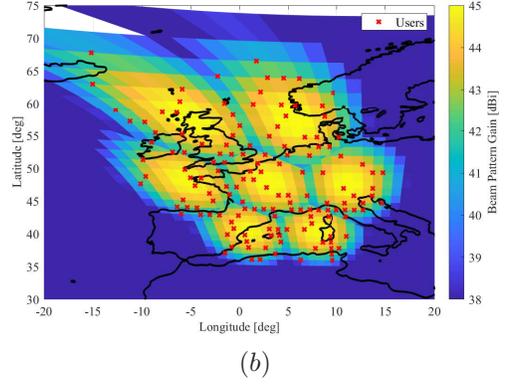}
			%\vspace{-0.15cm}
			\\
			$(b)$
			%\vspace{-0.25cm}
			%\caption{SRM: WF fig1 prob 7beam ZF}
		\end{minipage}
		\caption{The precoded multi-beam multi-user satellite system model: $(a)$ Schematic diagram of our considered system model with one single scheduled  user per beam; and $(b)$ The considered overlapping beam pattern.} 
		\label{fig:systemmodel}
		%\vspace{-10pt}
	\end{figure*}
	%%%%%%%%%%%%%%%%%%%%%%%%%%%%%%%%%%%%%%%%%%%%%%%
%	\vspace{-0.25cm}
	\section{System Model and Problem Statement}\label{sec:sys_model}
%	\vspace{-0.1cm}
	%%%%%%%%%%%%%%%%%%%%%%%%%%%%%%%%%%%%%%%%%%%%%%%
	In this section, we first introduce the MB-HTS system architecture, where the full available bandwidth is simultaneously used by all beams and, within each beam, the multiple users are multiplexed in a Time Division Multiple (TDM) manner in the forward link on a DVB-S2X carrier from the Gateway to the user beams. Meanwhile, Time Division Multiple Access (TDMA) is used on the return link.
	%multiple users are served following a Time Division Multiple Access (TDMA). 
	Next, motivated by the shortcomings of previous works in handling the demand-based constraints, a new multi-objective optimization framework is proposed.
	\subsection{System Model \& Channel Capacity}
%	\vspace{-0.1cm}
	
	\textcolor{black}{We consider the forward link of a broadband MB-HTS system that aggressively reuses the user link frequency to simultaneously serve multiple users sharing the same time and frequency plane as schematically shown in Fig.~\ref{fig:systemmodel}(a), with the overlapping beam pattern depicted in Fig.~\ref{fig:systemmodel}(b).}\footnote{\textcolor{black}{The capacity of multi-beam GEO systems allow multiple users to simultaneously access the network. The considered multiple-access scenarios bring superior improvements of the sum rate by serving more users and exploiting a proper precoding technique to mitigate mutual interference. However, the congestion will be problematic if, for example, each user is associated with its individual QoS demand and a limited power budget at the satellite. The present paper will address this raising issue by using both the model-based and data-driven approaches.}} Assuming $N$ overlapping beams, a maximum of $N$ users in the coverage area can be scheduled and served in each scheduling instance by the satellite. We assume that the actual scheduled  users per scheduling instance is $K$, as illustrated by the black-colored users in Fig.~\ref{fig:systemmodel}(a). In this paper, the system operates in a unicast mode, i.e., $K \leq N$.	We denote $\UEk$ the scheduled user~$k$ with  $k \in \Kcal \triangleq \{1, 2, \dots, K\}$  and $|\Kcal| = K$. {\color{black} Let us define $\bh_k\in\mathbb{C}^{N}$ the channel vector between the satellite and $\mathtt{UE}_k$, then the channel matrix ${\mathbf{H}}$ is defined as ${\mathbf{H}} = [{\bh}_1,{\bh}_2,\dots,{\bh}_K] \in \mathbb{C}^{N\times K}$. In particular, the channel is modeled in LOS link \cite{channel_model, 8746876}, and collects the channel state information (CSI) and phase rotations from the over-air propagation in the forward link, which is split into the two components as $	{\mathbf{H}} = \bar{\mathbf{H}}\mathbf{\Phi}$, where $\bar{\mathbf{H}}\in\mathbb{R}_+^{N\times K}$ indicates the practical features involving the satellite antenna radiation pattern, thermal noise, received antenna gain, and path loss. }
%	Let us define $\bh_k\in\mathbb{C}^{N}$ the channel vector between the satellite and $\UEk$, then the channel matrix ${\bH}$, defined as ${\bH} = [{\bh}_1,{\bh}_2,\dots,{\bh}_K] \in \mathbb{C}^{N\times K}$, collects the channel state information (CSI) and phase rotations from the over-air propagation in the forward link, which is split into the two components as $	{\bH} = \bar{\bH}\mathbf{\Phi}$, 
%		%		\begin{equation} \label{eq:ChannelMat}
%		%			{\bH} = \bar{\bH}\mathbf{\Phi},
%		%		\end{equation}
%		where $\bar{\bH}\in\mathbb{R}_+^{N\times K}$ indicates the practical features involving the satellite antenna radiation pattern, thermal noise, received antenna gain, and path loss. 
	The $(n,k)$-th element of $\bar{\bH}$ is concretely computed as $[\bar{\bH} ]_{nk} = ({\lambda\sqrt{G_R G_{nk}}})/({4\pi d_k\sqrt{K_BTB}})$, 
		%	\begin{equation}
		%		[\bar{\bH} ]_{nk} = \frac{\lambda\sqrt{G_R G_{nk}}}{4\pi d_k\sqrt{K_BTB}},
		%	\end{equation}
	where $\lambda$ is the wavelength of a plane wave; $d_k$ is the distance from $\UEk$ to the satellite; $G_R$ and $G_{nk}$ are the receiver antenna gain and the gain from the $n$-th satellite feed towards $\UEk$, $\forall n = 1, \ldots N$; $K_B$ is the Boltzmann constant; $T$ is the receiver noise temperature. The diagonal matrix $\mathbf{\Phi}\in\mathbb{C}^{K\times K}$ indicates the signal phase rotations owing to  different propagation paths, whose the $(k,l)$-th component is given as $[\mathbf{\Phi}]_{kl}= e^{j\phi_k}$ if $k = \ell$,
		%\begin{equation}
		%	[\mathbf{\Phi}]_{kl}=
		%	\left\{\begin{matrix}
		%		e^{j\phi_k} & \mbox{If } k = \ell, \\ 
		%		0 &  \mbox{Otherwise},
		%	\end{matrix}\right.
		%\end{equation}
	where  $\phi_k$ is a residual random phase component introduced by the satellite payload \cite{ESA}. Otherwise, $[\mathbf{\Phi}]_{kl}= 0$. %The channel model in \eqref{eq:ChannelMat} is hereafter exploited to define the precoding vectors and then to solve the considered optimization problem.}
	
	Let us define $s_k$ the data symbol that the system transmits to $\UEk$ with $\mathbb{E} \{ |s_k|^2 \} = 1$ and its allocated transmit power $p_k \in \mathbb{R}_{+}$. A predetermined precoding technique is implemented at the gateway to eliminate mutual interference among users and boost the system performance. Denoting $\mathbf{w}_k \in \mathbb{C}^N$ as the normalized precoding vector for $\UEk$ with $\| \bw_k\| = 1$, then the transmitted signal to all the $K$ scheduled users, denoted by $\mathbf{x} \in \mathbb{C}^N$, is $\mathbf{x} = \sum\nolimits_{k\in\Kcal} \sqrt{p_k} \mathbf{w}_k s_k$.
	%	\begin{equation} \label{eq:transsig}
	%		\mathbf{x} = \sum\nolimits_{k\in\Kcal} \sqrt{p_k} \mathbf{w}_k s_k.
	%	\end{equation}
	For practical satellite systems, the following system transmit power constraint must be satisfied:
%	%\vspace{-0.2cm}
	\begin{align}
%		\begin{split}
			&\mathbb{E} \{ \| \mathbf{x} \|^2 \}  \leq P_{\max} \non \\
			\Rightarrow  &\sum\nolimits_{k\in\Kcal} p_k \| \mathbf{w}_k \|^2 \mathbb{E} \{ |s_k|^2 \}  \stackrel{(a)}{=} \sum\nolimits_{k\in\Kcal} p_k \leq P_{\max},
%		\end{split}
%		%\vspace{-0.2cm}
	\end{align}
	where $(a)$ is obtained assuming that the data symbols are  mutually independent and the precoding vectors are normalized. Moreover, $P_{\max}$ is the maximum power that the satellite can allocate to the data transmission. By exploiting the transmitted signal notation $\mathbf{x}$, the received signal at $\UEk$, denoted by $y_k \in \mathbb{C}$, is a projection of the transmitted signal onto its propagation channel as
%	%\vspace{-0.2cm}
	\begin{align} \label{eq:y}
		y_k &= \mathbf{h}_k^H \mathbf{x} + n_k, \\
		&=  \sqrt{p_k} \bh_k^H\bw_k s_k + \sum\nolimits_{\ell\in\Kcal\backslash \{k\}} \sqrt{p_\ell} \bh_k^H\bw_\ell s_\ell + n_k, \forall k\in\Kcal, \non
%		%\vspace{-0.2cm}
	\end{align}
	where $n_k$ denotes the additive noise at the receiver with $n_k \sim \mathcal{CN}(0, \sigma^2)$. In the last equality of \eqref{eq:y},  the first part contains the desired signal for $\UEk$, while the remaining parts are mutual interference and noise. Assuming the availability of perfect channel state information (CSI) available at the gateway side,\footnote{\textcolor{black}{This paper  assumes perfect CSI with the purpose of validating our robust congestion control as an initial framework focused on static users. The impact of imperfect CSI besides channel aging problems and many issues are left for future work.}} the channel capacity of $\UEk$ is computed as follows
%	%\vspace{-0.2cm}
	\begin{equation}\label{eq:rate}
		R_k ( \{ p_{k'} \} ) = B \log_2 \left(1+\gamma_k( \{ p_{k'} \}) \right),\mbox{ [Mbps]},\mbox{ } \forall k\in\Kcal,
%		%\vspace{-0.2cm}
	\end{equation}
	where $\{ p_{k'} \} = \{ p_1, \ldots, p_K \}$ is the set of all the transmit power coefficients, and $B$~[MHz] is the overall bandwidth used for the user link. The signal-to-interference-and-noise ratio (SINR), %denoted by
	$\gamma_k(\{ {p}_{k'} \})$, is %given by
%	%\vspace{-0.2cm}
	\begin{equation} \label{eq:SINR}
		\gamma_k( \{ p_{k'} \}) = \frac{p_k|\bh_k^H\bw_k|^2}{\sum_{\ell\in\Kcal\backslash \{k\}}p_\ell |\bh_k^H\bw_\ell|^2 +\sigma^2}, \  \forall k\in\Kcal.
%		%\vspace{-0.2cm}
	\end{equation}
	We emphasize that  the SINR expression~\eqref{eq:SINR} can be applied to an arbitrary channel model and precoding technique. In this paper, we exploit \eqref{eq:SINR} to formulate and solve the demand-based optimization problems with the practical constraints that arise in the future satellite communications.
	%	\vspace{-0.25cm}

	%%%%%%%%%%%%%%%%%%%%%%%%%%%
%	\vspace{-0.55cm}
	\subsection{Single-Objective Optimization With QoS Constraints}
%	%\vspace{-0.2cm}
	%%%%%%%%%%%%%%%%%%%%%%%%%%%
	
	For MB-HTS systems, conventional power allocation problems focus on maximizing a utility function while maintaining the QoS requirements of the scheduled users under a limited power budget. By taking the sum-rate as an objective function example, a popular optimization formulation \cite{Bjon13:tcit, Liu13:TSP, Chiang08} is
%	%\vspace{-0.2cm}
	\begin{subequations} \label{RelatedWork}
		\begin{alignat}{2}
			&\underset{\{ p_{k'} \in \mathbb{R}_{+} \}}\maxi &\quad & f_0(\{ p_{k'} \}) \triangleq  \sum\nolimits_{k\in\Kcal}R_k( \{ p_{k'} \} ) \label{RelatedWorka}\\
			&\mbox{subject to} && R_k(\{ p_{k'} \})  \geq \xi_k, \mbox{ } \forall k\in\Kcal, \label{RelatedWorkb}\\
			&&& \sum\nolimits_{k\in\Kcal} p_{k} \leq P_{\max}, \label{RelatedWorkc}
		\end{alignat}
%		%\vspace{-0.2cm}
	\end{subequations}
	where $\xi_k$ [Mbps] corresponds the QoS requested by $\UEk$. In \eqref{RelatedWork}, the objective function $f_0(\{ p_{k'} \})$ can be an arbitrary utility function in satellite communications \cite{gao2021sum,lei2011secure}. Even though all the constraints are affine, solving problem~\eqref{RelatedWork} is still challenging when the objective function is non-convex. However, the feasible domain  is a convex set, thus if  $f_0 \left(\{ p_{k'}\} \right)$ is continuous and bounded from below, the global optimum to problem~\eqref{RelatedWork} always exists by means of the Weierstrass' theorem \cite{van2018joint}.

	\textcolor{black}{Problem~\eqref{RelatedWork} optimizes the transmit powers to simultaneously satisfy the QoS requirements of all the $K$ scheduled  users conditioned on the power limitation. Indeed, if the system is able to provide the QoS requirements simultaneously to all the users, problem~\eqref{RelatedWork} has a non-empty feasible set and it can be solved to obtain the global optimal solution. However, for many unfortunate users' locations and channel conditions, as well as for systems with strict power limitations, the system cannot provide the QoS requirements to every scheduled user that results in the congestion issue, where at least one user is served less data throughput than requested. This is because, in many user locations, one or more scheduled  users are located in places where the propagation channels are inferior with the dramatically small channel gains. Furthermore, the interference-limited scenario considered herein may further enlarge the infeasible cases. The congestion issue makes it challenging for the satellite to meet the requested demands simultaneously. In other words, this leads problem~\eqref{RelatedWork}  to be infeasible with high probability due to an empty feasible domain, i.e. problem~\eqref{RelatedWork} has no solution.}
	\begin{figure*}[t]
		\begin{minipage}{0.325\textwidth}
			%		\centering
			\includegraphics[width=0.9\textwidth]{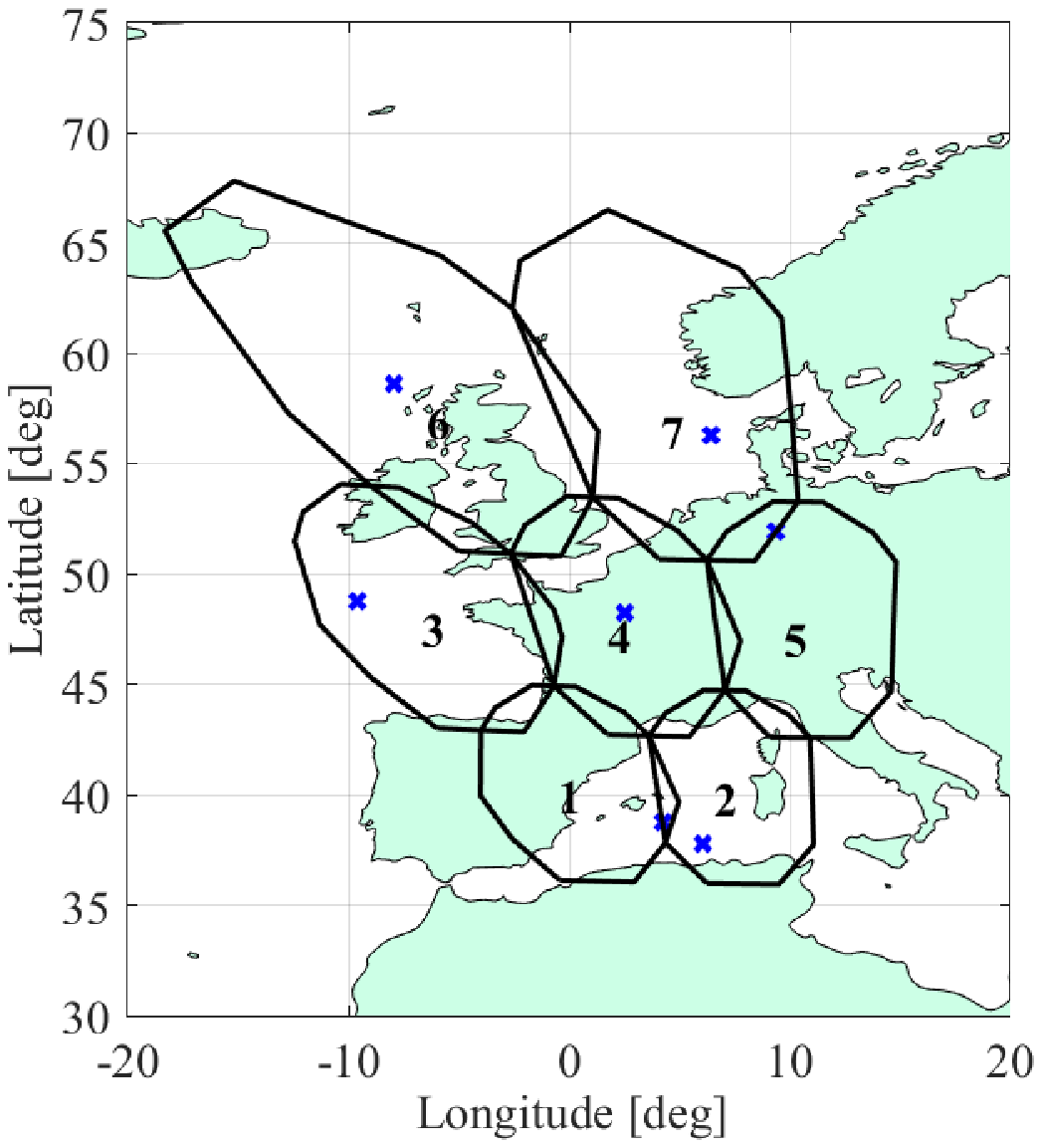} \\
			\centering $(a)$
			%\vspace{-1pt}
		\end{minipage}
		\begin{minipage}{0.325\textwidth}
			%		\centering
			\includegraphics[width=0.9\textwidth]{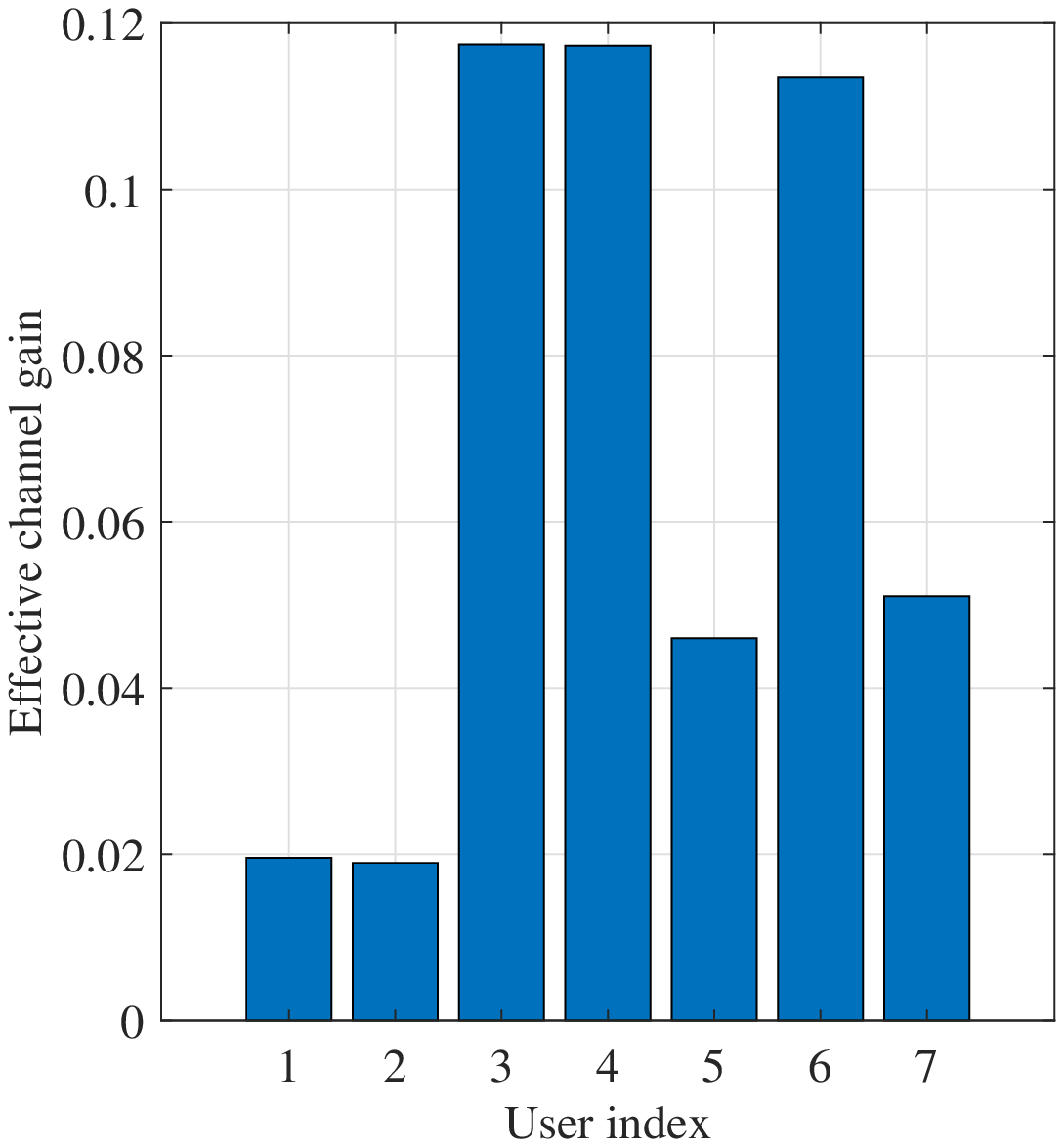}\\
			\centering $(b)$
			%\vspace{-1pt}
		\end{minipage}
		\begin{minipage}{0.325\textwidth}
			%		\centering
			\includegraphics[width=0.9\textwidth]{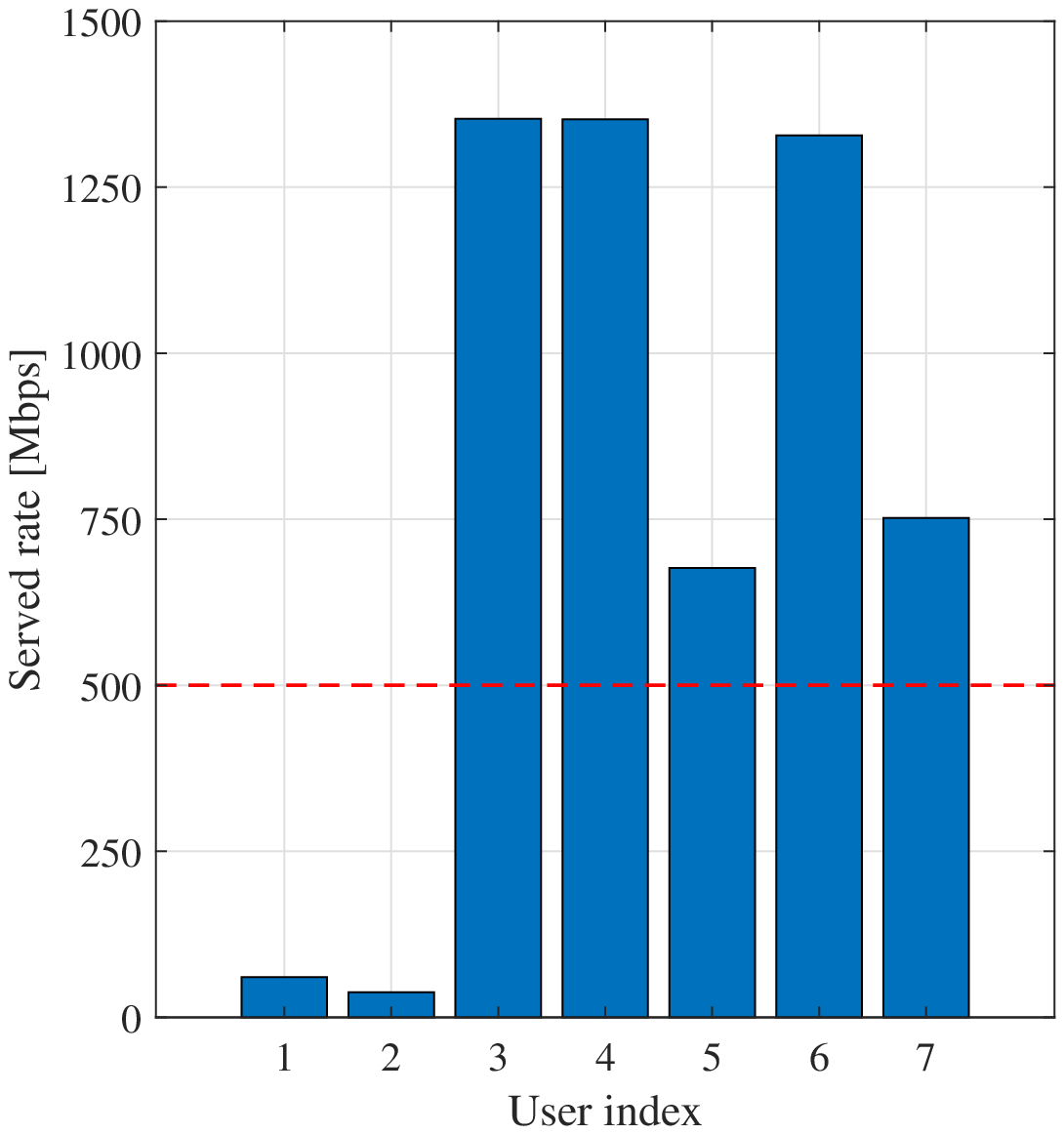} \\
			\centering $(c)$
			%\vspace{-1pt}
		\end{minipage}
		%\vspace{-0.2cm}
		\caption{A scheduling instance where each beam serves one scheduled  user: $(a)$  the user locations; $(b)$  the effective channel gains, i.e., defined as $|\mathbf{h}_k^H \mathbf{w}_k|^2, \forall k \in \Kcal$; and $(c)$  the served rate [Mbps] by utilizing the ZF precoding technique}
		\label{Fig:ChannelvsRate}
%		%\vspace{-30pt}
	\end{figure*}
	%yielding the implementation of power allocation solution with the demand-based constraints impossible if the interior-point methods are simply deployed \cite{van2021uplink} and much more efforts are in need. 
	For tractability, we can formulate an optimization without the demand-based constraints as follows
	%\begin{equation} \label{RelatedWorkSumRate}
	%	\begin{aligned}
	%		&\underset{\{ p_{k'} \in \mathbb{R}_{+} \}}\maxi &\quad & f_0(\{ p_{k'} \})\\
	%		&\mbox{subject to} && \sum_{k\in\Kcal} p_{k} \leq P_{\max}, 
	%	\end{aligned}
	%\end{equation}
%	%\vspace{-0.2cm}
	\begin{subequations} \label{RelatedWorkSumRate}
		\begin{alignat}{2}
			&\underset{\{ p_{k'} \in \mathbb{R}_{+} \}}\maxi &\quad & \sum\nolimits_{k\in\Kcal}R_k( \{ p_{k'} \} )\\
			&\mbox{subject to} && \sum\nolimits_{k\in\Kcal} p_{k} \leq P_{\max}, 
		\end{alignat}
% 		%\vspace{-0.2cm}
	\end{subequations}
	which was considered in \cite{aravanis2015power} and references therein.
	% \footnote{For the network with a hybrid number of users where or not they can be either served by their QoS requirements with no explicit guarantees, the Jain's fairness index is a good metric to measure how the offered data throughput matches the demands at the user levels \cite{jain1984quantitative}. By computing the satisfaction demand of each user, i.e., denoted by $o_{k}$ as a ratio between the offered data throughput and the QoS requirement of user~$k, \forall k$, then the Jain's index is computed as $J = (\sum_{k \in \Kcal} o_k)^2 / (K\sum_{k \in \Kcal} o_k^2)$, which varies from $1/K$ to $1$.}
	Fig.~\ref{Fig:ChannelvsRate}(a) shows an example of $N=7$ beams with $K=7$ scheduled users.  Fig.~\ref{Fig:ChannelvsRate}(c) plots the achievable rates for each of the scheduled users by considering problem~\eqref{RelatedWorkSumRate} as a consequence of their effective channel gains, which are depicted in Fig.~\ref{Fig:ChannelvsRate}(b) for completeness. The detailed parameter settings are given in Section~\ref{Sec:Results}. For this particular realization of user locations, there are two users with unfortunate effective channel conditions, which combined with the limited power budget will make it challenging for the satellite to ensure the users to be simultaneously serve with the same individual QoS requirement (say $500$ [Mbps]). However, the remaining scheduled users would still get their requested QoS or even better data throughput if some demand-based constraints would have been relaxed (such that the ones of user~$1$ and $2$).  It is because those users are located at the extreme locations as the boundary of the beams. Not shown here, but the harsh situation also comes from the fact that the QoS requirements are too high and the system cannot meet their services even consuming the entire power budget. 
	Motivated by the results in Fig.~\ref{Fig:ChannelvsRate}, a practical solution for power allocation is developed in this paper where QoS requirement satisfaction for the majority of the users is sought. For those users who cannot satisfy the QoS constraints, it may be sufficient to relax their QoS constraints or skip them for these particular scheduling instances. For such, we propose to convert \eqref{RelatedWork} from an infeasible problem to a feasible one. However, the identification of the users who are not able to reach their QoS requirements is not trivial. This paper investigates a class of  power allocation problems whose objective function includes both the sum-rate and the total number of satisfied users, which can effectively cope with such infeasible instances due to the network dimension whenever the congestion issue appears.\footnote{\textcolor{black}{The congestion is a complex issue in satellite communications. One potential solution for this issue is based on  the user scheduling over the time and frequency plane. However, for a given set of scheduled users, the congestion may still appear when allocating the limited power budget to maximize the total sum rate of the entire network  and satisfy the individual QoS demands.  Since user scheduling may help mitigating partially the congestion, the combination of the proposed
			power and congestion control approach with more advanced user scheduling is left for future
			work.}}
	
%	\vspace{-0.55cm}
	%%%%%%%%%%%%%%%%%%%%%%%%%%%%%%%%%%%%%%%
	\subsection{Proposed Multi-Objective Optimization} \label{sec:MultiObject}
%	%\vspace{-0.2cm}
	%%%%%%%%%%%%%%%%%%%%%%%%%%%%%%%%%%%%%%%
	To deal with congestion scenarios, we propose to split the $K$ scheduled users into two sets: $\Qcal$ with  $\Qcal \subseteq \Kcal$ being the satisfied-user set that contains users served by the system with data throughput equal or greater than their QoS requirements. The remaining users belong to the unsatisfied-user set $\Kcal \setminus \Qcal$. Our goal is to maximize the cardinality of the satisfied-user set $\Qcal$ and also to seek for the maximal value of the sum-rate metric $\sum_{k\in\Kcal}R_k( \{ p_{k'} \} )$. The ultimate goal is introduced as %by the objective vector $\mathbf{g}\left(\{ p_{k'}\}, \Qcal \right)$ as
%	%\vspace{-0.2cm}
	\begin{equation} \label{eq:ObjFunc2}
		\mathbf{g}\left(\{ p_{k'}\}, \Qcal \right) = \left[ \sum\nolimits_{k\in\Kcal}R_k( \{ p_{k'} \} ),  |\Qcal| \right]^T,
%		%\vspace{-0.2cm}
	\end{equation}
	which should be categorized as a multi-objective function, where the two performance metrics are optimized in a single framework. Motivated by the use of \eqref{eq:ObjFunc2}, we study a joint design of the power allocation and the satisfied-user selection to optimize the multi-objective function $\mathbf{g}\left(\{ p_{k'}\}, \Qcal \right)$ 
%	%\vspace{-0.2cm}
	\begin{subequations} \label{probGlobal}
%		\vspace{-5pt}
		\begin{alignat}{2}
			&\underset{ \{ p_{k'} \in \mathbb{R}_{+} \}, \Qcal}{\mathrm{maximize}}&\quad & \mathbf{g}\left(\{ p_{k'}\}, \Qcal \right) \label{probGlobala}\\
			&\mbox{subject to} && R_k( \{ p_{k'} \}) \geq \xi_k,  \forall k\in\Qcal, \label{probGlobalb}\\
			&&& \sum\nolimits_{k \in\Kcal} p_k \leq P_{\max},\label{probGlobalc}\\
			&&& \Qcal  \subseteq \Kcal. \label{probGlobald}
%			%\vspace{-0.2cm}
		\end{alignat}
%		%\vspace{-0.2cm}
	\end{subequations}
	\textcolor{black}{The key distinction from previous works in the literature is that problem~\eqref{probGlobal} is always feasible since the satisfied-user set $\Qcal$ can span from an empty set, i.e., no user satisfies its QoS demand; to the scheduled-user set $\mathcal{K}$, i.e., all the $K$ scheduled users satisfy their QoS requirement. The proposed formulation is very convenient in practice as problem \eqref{probGlobal} can provide a power allocation solution in any channel conditions whilst still ensuring the system's performance in some extended aspect. Expressly, the objective function \eqref{probGlobala} indicates that we find an optimal set of the transmit power coefficients that simultaneously maximizes the utility function $f_0(\{ p_{k'}\})$ and the satisfied-user set $\Qcal$.} We stress that thanks to the constraint \eqref{probGlobalb},  problem~\eqref{probGlobal} only guarantees the individual QoS requirements of the satisfied-user set $\Qcal$. Different from a single objective function in \eqref{RelatedWork}, the decision space of problem~\eqref{probGlobal} is defined by
%	%\vspace{-0.2cm}
	\begin{equation}
		\begin{split}
			\mathcal{D} = \Big\{ \{ p_{k'}\}, \Qcal  \big|  R_k( \{ p_{k'} \}) &\geq \xi_k, \forall k\in\Qcal, \\
			P_{\max} &\geq \sum\nolimits_{k \in\Kcal} p_k , \Qcal \subseteq \Kcal \Big\},
		\end{split}
%		%\vspace{-0.25cm}
	\end{equation}
	which is a non-convex set. The data of problem~\eqref{probGlobal} consists of the decision space $\mathcal{D}$, the objective function vector $\mathbf{g}\left(\{ p_{k'}\}, \Qcal \right)$, together with the objective space $\mathbb{R}_{++}^2$. In principle, $\mathbf{g}\left(\{ p_{k'}\}, \Qcal \right)$  is mapped from the objective space to an ordered space, say $(\mathbb{R}_{++}^2, \geq, \subseteq)$, in which the feasibility is testified along with iterations by the order relations $\geq$ and $\subseteq$. This mapping is referred to as the $\theta$ model that depicts a relation between the objective space and the order space, where the maximization in \eqref{probGlobal} is determined. Alternatively speaking, problem~\eqref{probGlobal} should be completely defined by the data $(\mathcal{D}, \mathbf{g}(\{ p_{k'} \}, \Qcal), R_{++}^2)$, the model map $\theta$, and the order space $\mathbb{R}_{++}^2$. We now characterize  an $\pmb{\epsilon}$-\textit{Pareto optimal} solution $\{ \{ p_{k'}^\ast \}, \Qcal^{\ast} \} \in \mathcal{D}$  to problem~\eqref{probGlobal}, if there exists no $\{ \{ p_{k'} \}, \Qcal \} \in \mathcal{D}$ such that
%	%\vspace{-0.2cm}
	\begin{equation} \label{eq:ParetoOptimal}
		\mathbf{g}\left(\{ p_{k'} \}, \Qcal \right) + \pmb{\epsilon} \succeq  \mathbf{g}\left(\{ p_{k'}^\ast \}, \Qcal^\ast \right)  , 
%		%\vspace{-0.2cm}
	\end{equation}
	where $\pmb{\epsilon} = [\epsilon_1, \epsilon_2]^T$ with $\epsilon_1, \epsilon_2 \in \mathbb{R}_+$ are the tolerance corresponding to the two objective functions. The property \eqref{eq:ParetoOptimal} implies no other solutions $\{ \{ p_{k'} \}, \Qcal \} \in \mathcal{D}$ fulfilled the coexisted conditions: $f_0\left(\{ p_{k'} \}, \Qcal \right) + \epsilon_1 \geq f_0 \left(\{ p_{k'}^\ast \}, \Qcal^\ast \right)$, and $|\Qcal| + \epsilon_2 \geq |\Qcal^\ast|$, 
	%	\begin{equation}
	%		f_0\left(\{ p_{k'} \}, \Qcal \right) + \epsilon_1 \geq f_0 \left(\{ p_{k'}^\ast \}, \Qcal^\ast \right)  \mbox{, and } |\Qcal| + \epsilon_2 \geq |\Qcal^\ast|,
	%	\end{equation}
	which unveils a balance between the two objective functions at the optimum. We observe that if $\epsilon_1 = \epsilon_2 = 0$, the above definition reduces to an $\pmb{\epsilon}$-Pareto optimal solution%, i.e., also called as \textit{efficient}, \textit{non-dominated}, or \textit{non-inferior} solution 
	, which can be only improved by upgrading one objective function and scarifying the other. Thus, an $\pmb{\epsilon}$-\textit{properly Pareto optimal} solution is introduced as an $\pmb{\epsilon}$-Pareto optimal solution with a bound trade-off between the two objectives defined in \eqref{eq:ObjFunc2}. An $\pmb{\epsilon}$-\textit{Pareto dominant vector} is derived as the objective function vector $\mathbf{g}(\{ p_{k'} \}, \Qcal)$ at the corresponding $\pmb{\epsilon}$-properly Pareto optimal solution. We notice that the $\pmb{\epsilon}$-\textit{Pareto frontier} collects all the properly $\pmb{\epsilon}$-Pareto optimal vectors.
%	%\vspace{-0.2cm}
	\begin{remark}
		Problem~\eqref{probGlobal} jointly optimizes the sum rate and the total number of satisfied users subject to the limited transmit power constraint under the viewpoints of multi-objective optimization. The proposed problem \eqref{probGlobal} is a generalized version of previous works on a single-objective function with/without demand-based constraints as \cite{gao2021sum,aravanis2015power} and references therein. Problem~\eqref{probGlobal} can effectively handle the congestion issue appearing when some users do not meet their QoS requirements. This practical matter in multiple access communications originates from the limited power budget, the channel conditions, and the individual QoS requirements. An extension to a multiple-objective optimization framework with more than two objective functions or with different metrics should be interesting for a future work.
	\end{remark}
%	%\vspace{-0.2cm}
	{ By exploiting either  the scalarization  or  nonscalarization approach to handle the multiple objective functions}, we may attain an $\pmb{\epsilon}$-properly Pareto optimal solution to  problem~\eqref{probGlobal}, { following by the $\pmb{\epsilon}$-Pareto frontier}. If the nonscalarization approach is employed, there is no prior information about the objective functions available in advance. For this direction, natural inspired algorithms that simultaneously optimize all the objective functions are often exploited to attain the $\pmb{\epsilon}$-Pareto frontier \cite{le2020pareto}. The nonscalarization approach requires significantly high computational complexity since the Pareto frontier is obtained by directly solving the multiple-objective optimization problem.  Once the scalarization approach is utilized by exploiting the preferential information from the decision maker about the objective functions, we can transfer the multi-objective optimization problem \eqref{probGlobal} to a single-objective optimization problem. The scalarization approach obtains the $\pmb{\epsilon}$-Pareto frontier by iteratively solving some single objective optimizations, each concentrating on a given set of priorities between the objective functions. Consequently, the scalarization approach usually offers the solution to  problem~\eqref{probGlobal} with lower computational complexity than the nonscalarization approach \cite{ehrgott2005multicriteria}.

%	%\vspace{-0.5cm}
	%%%%%%%%%%%%%%%%%%%%%%%%%%%%%%%%%%%%%%%%%%%%%%%%%%
	\section{Model-based and Data-driven Approaches} \label{Sec:ModelDL}
%	%\vspace{-0.2cm}
	%%%%%%%%%%%%%%%%%%%%%%%%%%%%%%%%%%%%%%%%%%%%%%%%%%
	This section presents the model-based approach to obtain an $\pmb{\epsilon}$-properly Pareto optimal solution to problem~\eqref{probGlobal} in polynomial time by exploiting the scalarization approach. The obtained solution is then utilized in Section \ref{Sec:DataDriven} for training a neural network that can predict a solution to  problem~\eqref{probGlobal} with extremely low computational complexity and tolerable accuracy.
	%	\vspace{-0.25cm}
%	\vspace{-0.55cm}
	%%%%%%%%%%%%%%%%%%%%%%%%%%%%%%%%%%%%%%%%%%%%%%%%%%
	\subsection{Model-based Approach}
%	%\vspace{-0.2cm}
	%%%%%%%%%%%%%%%%%%%%%%%%%%%%%%%%%%%%%%%%%%%%%%%%%%
	In this section, we proceed with  problem~\eqref{probGlobal} by exploiting the weighted sum method \cite{ehrgott2005multicriteria}. Specifically, we define the weights $\mu_1 \geq 0$ and $\mu_2 \geq 0$ with $\mu_1 + \mu_2 = 1$ that respectively stand for the priority of the two objective functions in $\mathbf{g}(\{ p_{k'} \}, \Qcal)$. If $\{ \{ p_{k'}^{\ast} \}, \Qcal^\ast \}$ is an optimal solution to the single-objective optimization problem:
%	%\vspace{-0.2cm}
	\begin{subequations} \label{probGlobalv1}
		\begin{alignat}{2}
			&\underset{\{ p_{k'} \in \mathbb{R}_{+} \}, \Qcal}{\mathrm{maximize}} &\quad& \mu_1 \sum\nolimits_{k\in\Kcal}R_k( \{ p_{k'} \} ) +  \mu_2 |\Qcal| \\
			&\mbox{subject to} &\quad& R_k(\{ p_{k'} \} ) \geq \xi_k, \forall k\in\Qcal, \\
			&&& \sum\nolimits_{k\in\Kcal}p_k \leq P_{\max}, \label{eq:PowerConst} \\
			&&& \Qcal  \subseteq \Kcal,
		\end{alignat}
%		%\vspace{-0.2cm}
	\end{subequations}
	with an $\pmb{\epsilon}$-accuracy, then $\{ \{ p_{k'}^{\ast} \}, \Qcal^\ast \}$  is an $\pmb{\epsilon}$-properly  Pareto  optimal  solution  to  problem~\eqref{probGlobal}. We emphasize that in \eqref{probGlobalv1}, the weights $\mu_1$ and $\mu_2$ are flexibly designed by the decision maker. By adjusting these two values, an $\pmb{\epsilon}$-Pareto frontier to  problem~\eqref{probGlobal} is obtained. After that, the most desirable solution to the decision maker is chosen from the $\pmb{\epsilon}$-Pareto frontier.  Even though \eqref{probGlobalv1} is a single-objective problem, it is still non-convex due to a hybrid between the continuous and discrete feasible domains of the optimization variables. We, therefore, make an assumption follow the trends in QoS satisfaction in future satellite communications \cite{TedrosTWC}.
%	%\vspace{-0.2cm}
	\begin{assumption}\label{AssumpPrior}
		In order to offer the QoS requirements for a maximum number of users in the coverage area with a finite transmit power level, we focus on a scenario that the decision maker selects $\mu_1$ and $\mu_2$ to obtain the largest cardinality of the satisfied-user set $\Qcal$ before paying attention to maximize the sum-rate for a given assigned bandwidth. The channel conditions may lead to some scheduled  users not reaching their QoS requirements. One can improve the QoSs for those unsatisfied users by subtracting the power leftover, which is allocated to the satisfied users with a higher served rate than requested.
	\end{assumption}
%	%\vspace{-0.2cm}
	A priority on the QoSs of the scheduled users has been claimed by Assumption~\ref{AssumpPrior} and is effectively achieved by the satisfied-user set $\Qcal$. The limited power budget is therefore utilized in a strategy to maximize the rate demands for all the scheduled  users in the network instead of focusing on an individual entity. The remaining power, if possible, will be dedicated to maximizing the sum rate. { Motivated by the Perron-Frobenius theorem \cite{pillai2005perron,bambos2000channel},}  we observe the conditions required to all the scheduled  users with their rate satisfactions as shown in Theorem~\ref{Cardinality}.
%	%\vspace{-0.3cm}
	\begin{theorem} \label{Cardinality}
		If $\UEk$ requests a non-zero QoS, i.e., $\xi_k > 0,$ then all the $K$ scheduled  users can be served with at least their individual QoS requirements as the following conditions hold
%		%\vspace{-0.2cm}
		\begin{align}
			&\lambda(\mathbf{R} \mathbf{Q}) < 1, \label{eq:rho}\\
			& \mathbf{1}_K^T (\mathbf{I}_K - \mathbf{R}\mathbf{Q})^{-1} \pmb{\nu} \leq P_{\max}, \label{eq:PowerConstr}
%			%\vspace{-0.2cm}
		\end{align}
		where $\pmb{\nu} = [\nu_1, \ldots, \nu_K]^T \in \mathbb{R}_{+}^K$ with $\nu_k = \alpha_k \sigma^2 / ( (\alpha_k +1 ) |\mathbf{h}_k^2 \mathbf{w}_k|^2 )$ and $\alpha_k = 2^{\xi_k/B} -1, \forall k \in \Kcal$. The matrix $\mathbf{R} \in  \mathbb{R}^{K \times K}$ has the $(k,k')-$th element defined as $[\mathbf{R}]_{kk'} = \frac{\alpha_k}{(\alpha_k +1)|\mathbf{h}_k^H \mathbf{w}_k|^2}$ if $k=k'$. Otherwise, $[\mathbf{R}]_{kk'} = 0$. 
		%\begin{equation}
		%[\mathbf{R}]_{kk'} = \begin{cases} 
		%\frac{\alpha_k}{(\alpha_k +1)|\mathbf{h}_k^H \mathbf{w}_k|^2} & \mbox{If } k=k',\\
		%0 & \mbox{Otherwise},
		%\end{cases} ,
		%\end{equation}
		The $(k,k')$-th element of matrix $ \mathbf{Q} \in  \mathbb{R}^{K \times K}$ is 
		%\begin{equation}
		$[\mathbf{Q}]_{kk'} = | \mathbf{h}_k^H \mathbf{w}_{k'} |^2$.
		%\end{equation}
		In \eqref{eq:rho}, $\lambda(\mathbf{R} \mathbf{Q}) = \max \{ |\lambda_1|, \ldots, |\lambda_K| \}$ is the spectral radius of $\mathbf{R} \mathbf{Q}$, whose eigenvalues are denoted as $\lambda_1, \ldots, \lambda_K$.
	\end{theorem}
%	%\vspace{-0.5cm}
	\begin{proof}
	%	{The proof extends some previous works on radio links \cite{pillai2005perron,bambos2000channel} to the MB-HTS systems with an arbitrary precoding technique. We assume that the system can offer the requested QoS to all the $K$ scheduled  users in each scheduling instance with the power budget. The main proof is sketched in Appendix~\ref{Appendix:Cardinality}.}
	See  Appendix~\ref{Appendix:Cardinality}.
	\end{proof}
%	%\vspace{-0.2cm}
	Theorem~\ref{Cardinality} gives the necessary and sufficient conditions for the satellite to serve all the $K$ scheduled  users with the QoS requirements in an MB-HTS system, while still maximizing a utility function $f_0 (\{ p_{k'}\})$. {Unlike previous works, the conditions \eqref{eq:rho} and \eqref{eq:PowerConstr} explicitly represent the existed unique power solution for a precoded satellite system, which point out  the power allocation solution as a multi-variate function of many variables such as the propagation channels, the precoding vectors, the noise power, the QoS requirements, and the power budget. More precisely, the necessary condition in \eqref{eq:rho} ensures a unique power solution. The sufficient condition \eqref{eq:PowerConstr} ensures the satellite having enough power to provide the demand to each user. Though Theorem~\ref{Cardinality} assumes that $\Qcal = \Kcal$, it gives an efficient way to testify if all the $K$ scheduled  users can be served with their QoSs, and thus facilitates the reformulation of  problem~\eqref{probGlobalv1} in an efficient fashion by removing the optimization variable $\Qcal$. Conditioned on the  power budget of the satellite, the total transmit power needed to satisfy the QoS requirements can be bounded from below  as shown in Corollary~\ref{CorrNoSatisfied}.}
	\begin{corollary}\label{CorrNoSatisfied}
		For a given realization of users' locations and QoS requirements, the total transmit power is lower bounded by
%		%\vspace{-0.2cm}
		\begin{equation} \label{eq:LowerboundPower}
			\sum\nolimits_{k \in \Kcal} p_k \geq \mathbf{1}_K^T \pmb{\nu} / \| \mathbf{I}_K - \mathbf{R} \mathbf{Q} \|_2.
			%\vspace{-0.2cm}
		\end{equation}
	\end{corollary}
%%\vspace{-0.2cm}
	\begin{proof}
		From \eqref{eq:PowerSol} in Appendix~\ref{Cardinality}, the total transmit power that the $K$ scheduled  users need to satisfy the individual rate demand is reformulated as $\sum\nolimits_{k \in \Kcal} p_k \stackrel{(a)}{=} \mathrm{tr}( (\mathbf{I}_K - \mathbf{R}\mathbf{Q})^{-1} \pmb{\nu} \mathbf{1}_K^T ) \stackrel{(b)}{ \geq}  \mathrm{tr}( \pmb{\nu} \mathbf{1}_K^T) / \| \mathbf{I}_K - \mathbf{R} \mathbf{Q} \|_2 \stackrel{(c)}{=} \mathbf{1}_K^T \pmb{\nu} /\| \mathbf{I}_K - \mathbf{R} \mathbf{Q} \|_2$, 
		%		\begin{equation}
		%			\begin{split}
		%				\sum\nolimits_{k \in \Kcal} p_k &\stackrel{(a)}{=} \mathrm{tr}( (\mathbf{I}_K - \mathbf{R}\mathbf{Q})^{-1} \pmb{\nu} \mathbf{1}_K^T ) \stackrel{(b)}{ \geq}  \mathrm{tr}( \pmb{\nu} \mathbf{1}_K^T) / \| \mathbf{I}_K - \mathbf{R} \mathbf{Q} \|_2 \stackrel{(c)}{=} \mathbf{1}_K^T \pmb{\nu} /\| \mathbf{I}_K - \mathbf{R} \mathbf{Q} \|_2,
		%			\end{split}
		%		\end{equation}
		where $(a)$ and $(c)$ is obtained by utilizing the identity $\mathrm{tr}(\mathbf{X} \mathbf{Y}) = \mathrm{tr}(\mathbf{Y} \mathbf{X})$ with the two matched-size matrices $\mathbf{X}$ and $\mathbf{Y}$; $(b)$ is because $\mathbf{I}_K - \mathbf{R} \mathbf{Q}$ is a positive semidefinite matrix and then using \cite[Lemma B.8]{bjornson2017massive}. We  conclude the proof.
	\end{proof}
%\vspace{-0.2cm}
	The lower bound in \eqref{eq:LowerboundPower} is two-fold: First, the total transmit power is always positive if each scheduled  user requires a non-zero rate due to the mutual interference and the thermal noise. Second, it unveils the effectiveness of the precoding technique. A good selection should effectively mitigate the mutual interference among the scheduled users to attain the large spectral norm of matrix $\mathbf{I}_K - \mathbf{R} \mathbf{Q}$.

	%Next, we propose an algorithm for the overall demand-based space, even when some scheduled  users do not satisfy their QoS requirements. 
	
	Motivated by the aforementioned discussions, we next propose an algorithm to effectively address problem \eqref{probGlobal} and achieve a good local solution by solving the weighted sum optimization problem \eqref{probGlobalv1}. The satisfied-user set $\Qcal$ is initialized as an empty set due to no prior information. For given precoding vectors $\{\bw_{k'}\}$, the conditions \eqref{eq:rho} and \eqref{eq:PowerConstr} result in two possible cases:
	\begin{itemize}
		\item[$i)$] If those conditions hold, then all the $K$ scheduled users achieve (at least) their individual QoS requirements. Therefore, $R_k(\{p_{k'}\}) \geq  \xi_k, \forall k\in\Kcal$, and $\Qcal = \Kcal$. From Theorem~\ref{Cardinality} and Assumption~\ref{AssumpPrior}, problem~\eqref{probGlobalv1} is mathematically equivalent to \eqref{RelatedWork}. This case always offers a nonempty feasible set and corresponds to a system with no congestion.
		\item[$ii)$] As one of those conditions is not satisfied, at least one scheduled  user does not satisfy its QoS requirement (unsatisfied user), and therefore congestion appears. A special mechanism needs to handle this case if one considers the traditional sum-rate optimization \eqref{RelatedWork} due to an empty feasible set. However, it is not such the case for problem~\eqref{probGlobalv1}.
	\end{itemize}
	We stress that the first case maximizes the sum rate that satisfies the demand-based constraints of all the $K$ scheduled  users by a limited power budget. Since the feasible region must have an interior point, we can apply an interior-point method to obtain the solution to  problem~\eqref{RelatedWork}, e.g., \cite{qi2018precoding},  which may be implemented by a general-purpose toolbox such as CVX \cite{cvx2015}. However, it is a high computational complexity solution and does not work for the second case when at least one unsatisfied user gets a lower data throughput than the requirement. In this case, to solve problem~\eqref{probGlobalv1}, the priority is to maximize the number of satisfied users. %To do this, it is necessary to determine which user has poor channels or serving them causes substantial mutual interference to the others that we can provide them under their demands. 
	Mathematically, we optimize the cardinality of $\Qcal$ as follows
	%\vspace{-0.2cm}
	\begin{subequations} \label{alg_SRM:phase21}
		\begin{alignat}{2}
			&\underset{\{ p_{k'} \in \mathbb{R}_{+} \}}{\mathrm{maximize}} &\quad&   |\Qcal|\\
			&\mbox{subject to} &\quad& R_k(\{ p_{k'} \} ) \geq \xi_k, \forall k\in\Qcal, \\
			&&& \sum\nolimits_{k'\in\Kcal}p_{k'} \leq {P}_{\max},\\
			&&& \Qcal \subseteq \Kcal.
		\end{alignat}
%		%\vspace{-0.2cm}
	\end{subequations}
	Since problem~\eqref{alg_SRM:phase21} is a non-convex and non-smooth problem, it is not trivial to obtain the global optimum of the transmit powers. We now propose an iterative low-cost solution to get rid of this issue with a good local solution for problem \eqref{alg_SRM:phase21}. As an effective way to initialize the satisfied-user set $\Qcal$, we solve the sum-rate maximization problem without the demand constraints in \eqref{RelatedWorkSumRate} to obtain an initial set of the power allocation coefficients $\{ p_{k'}^{\ast, (0)}\}$. Next, we use these power allocation coefficients to define the initial satisfied-user set $\Qcal^{\ast, (0)}$ as shown below,
	%\vspace{-0.2cm}
	\begin{equation} \label{eq:Q0}
		\Qcal^{\ast, (0)} = \big\{ k \big| R_k \big( \{ p_{k'}^{\ast, (0)}\} \big) \geq \xi_k, k \in \Kcal \big\},
		%\vspace{-0.2cm}
	\end{equation}
	where $ R_k \{ p_{k'}^{\ast, (0)}\}$ is given in \eqref{eq:rate} but with $p_{k'} = p_{k'}^{\ast, (0)}, \forall k$. We numerically observe that the scheduled users that typically satisfy its QoS requirements are those with good effective channel gains and/or those suffering less mutual interference. Those scheduled  users contribute significantly to the objective function of  problem \eqref{RelatedWorkSumRate}. 
	
	% Among them, the satisfied users served by at least their demands after identifying from the solution to problem \eqref{RelatedWorkSumRate} are added to the satisfied-user set $\Qcal$. 
	
	In the following, we exploit the fact that we can move a portion of the power that is assigned to scheduled users that are getting more than what they actually requested to improve the conditions of less fortunate users. In more details, we design an iterative approach that enables to expand the set $\Qcal$ after each iteration. The main idea is that the satisfied users in $\Qcal$ are only served by the exact QoS requirements, all the remaining power budget of the satellite should be allocated to the other scheduled users to enhance their data throughput such that there is an opportunity to join the satisfied-user set $\Qcal$. To find new users to be added to $\Qcal$ at iteration~$n$, we focus on the following optimization problem: 
	%\vspace{-0.2cm}
	\begin{subequations} \label{alg_SRM:phase23}
		\begin{alignat}{2}
			&\underset{\{ p_{k'}^{(n)} \in \mathbb{R}_{+} \}}{\mathrm{maximize}} &\quad& \sum\nolimits_{k\in\Kcal}R_{k} \big(\{p_{k'}^{(n)}\} \big) \label{alg_SRM:phase23a}\\
			&\mbox{subject to} &\quad& R_k(\{ p_{k'}^{(n)} \} ) = \xi_k, \forall k \in\Qcal^{\ast, (n-1)}, \label{alg_SRM:phase23b}\\
			&&& \sum\nolimits_{k\in\Kcal}p_k^{(n)} \leq P_{\max}, \label{alg_SRM:phase23c}
		\end{alignat}
%		%\vspace{-0.2cm}
	\end{subequations}
	with the optimal power solution $\{ p_{k'}^{\ast,(n)} \}$. Different from  aforementioned problems, it is worth noting that the constraints \eqref{alg_SRM:phase23b} target the satellite to serve the satisfied users  in $\Qcal$ with only their QoS demands. With a finite power level $P_{\max}$, the remaining satellite energy should be allocated to the scheduled  users with bad channel conditions by expecting that they are potential candidates to join the satisfied-user set $\Qcal$. If there are scheduled  users served equal to or greater than their demands at iteration~$n$, they will be added to the satisfied-user set $\Qcal$ by
	%\vspace{-0.2cm}
	\begin{equation} \label{eq:updateQ}
		\Qcal^{\ast, (n)} =  \big\{ k \big| R_k \big(\{ p_{k'}^{\ast, (n)} \big) \geq \xi_k, k \in \Kcal \big\},
		%\vspace{-0.2cm}
	\end{equation}
	where $ R_k \{ p_{k'}^{\ast, (n)}\}$ is given in \eqref{eq:rate} but with $p_{k'} = p_{k'}^{\ast, (n)}, \forall k$. After that the iteration index is increased as $n=n+1$, which leads to an iterative approach. 
	Notice that it should maximize the number of scheduled  users that satisfy their requirements in each iteration with the objective to maximize the sum rate of all the $K$ scheduled  users. We emphasize that the second case is only executed after checking that conditions \eqref{eq:rho} and \eqref{eq:PowerConstr} are not satisfied, so the cardinality of the satisfied-user set is less than the number of scheduled  users along iterations, i.e., $|\Qcal^{\ast,(n)} | < K, \forall n$. 
	Our proposed approach is summarized in Algorithm \ref{alg:glob_alg} with its convergence given in Theorem~\ref{theorem:Coverge}.
		%\vspace{-0.2cm}
	\begin{theorem} \label{theorem:Coverge}
		If all the $K$ scheduled  users cannot be served with their QoS requirements under a given power budget $P_{\max}$ and the obtained optimized power coefficients at each iteration by solving \eqref{alg_SRM:phase23},  the following convergence properties hold and therefore Algorithm~\ref{alg:glob_alg} converges to a fixed point solution,
		%\vspace{-0.2cm}
		\begin{align}
			\ldots & \geq |\Qcal^{\ast,(n)}| \geq |\Qcal^{\ast,(n-1)}| \geq \ldots \geq |\Qcal^{\ast,(0)}|, \label{eq:QSeries}\\
			\ldots & \leq \sum\nolimits_{k \in \Kcal} R_k \big( \{ p_{k'}^{\ast,(n)} \} \big) \leq \sum\nolimits_{k \in \Kcal} R_k \big( \{ p_{k'}^{\ast,(n-1)} \} \big)\non \\
			 &\leq  \ldots \leq \sum\nolimits_{k \in \Kcal} R_k \big( \{ p_{k'}^{\ast,(0)} \} \big), \label{eq:RSeries}
			%\vspace{-0.2cm}
		\end{align}
	\end{theorem}
%\vspace{-0.2cm}
	\begin{proof}
	%	The proof is to confirm the monotonic property of the sum rate utility function and the cardinality of the satisfied-user set along iterations. The detailed proof is available in Appendix~\ref{appendix:Coverge}.
See	Appendix~\ref{appendix:Coverge}.
	\end{proof}
	%\vspace{-0.2cm}
	Theorem~\ref{theorem:Coverge} indicates an improvement of the satisfied-user set after each iteration by sacrificing an amount of the sum-data throughput that is aligned with the $\pmb{\epsilon}$-properly Pareto optimal solution in Section~\ref{sec:MultiObject}. When the congestion issue appears, the $K$ scheduled  users are split into two sets: the satisfied-user set $\Qcal$ containing the users served by the data throughput at least their demands, and the unsatisfied-user set $\Kcal \setminus \Qcal$ with the other users served by the throughput less than their demands.
	
	\begin{algorithm}[t]
		\begin{algorithmic}[1]%\fontsize{11.5}{11}\selectfont
			\protect\caption{An iterative algorithm to obtain a local solution to problem \eqref{probGlobal}}
			\label{alg:glob_alg}
			\global\long\def\algorithmicrequire{\textbf{INPUT:}}
			\REQUIRE Channel vectors $\{\bh_{k}\}$; Maximum power $P_{\max}$; QoS requirement set $\{\xi_k\}$.
			\STATE Compute the precoding vectors $\{\bw_{k'}\}$ based on the channel vectors $\{ \mathbf{h}_{k'} \}$.
			\STATE Compute the matrices $\mathbf{R}, \mathbf{Q},$ and the vector $\pmb{\nu}$. % based on the channel vectors $\{ \mathbf{h}_{k'} \}$, the precoding vectors  $\{\bw_{k'}\}$, and the rate demand set $\{\xi_k \}$.
			%		\STATE Validate the conditions \eqref{eq:rho}, \eqref{eq:PowerConstr}
			\IF {Conditions \eqref{eq:rho} and \eqref{eq:PowerConstr} are satisfied}
			\STATE Update $\Qcal^\ast = \Kcal$ and solve problem \eqref{RelatedWork} to obtain  $\{ p_{k'}^{\ast} \}$.
			\ELSE
			\STATE Solve problem \eqref{RelatedWorkSumRate}	 to obtain $\{ p_{k'}^{\ast,(0)} \}$ and update $\Qcal^{\ast,(0)}$ as in \eqref{eq:Q0}.
			\STATE Initialize the accuracy $\delta = |\Qcal^{\ast, (0)}|$ and set $n=0$.
			\WHILE {$\delta \neq 0$} 
			\STATE Set iteration index $n = n+1$.	
			\STATE Solve problem \eqref{alg_SRM:phase23} to obtain $\{ p_{k'}^{\ast, (n)} \}$ and then update $\Qcal^{\ast, (n)}$ as in \eqref{eq:updateQ}. 
			\STATE Update the accuracy $\delta = |\Qcal^{\ast, (n)}| - |\Qcal^{\ast, (n-1)}|$. 
			\ENDWHILE
			\ENDIF
			\global\long\def\algorithmicrequire{\textbf{OUTPUT:}}
			\REQUIRE  The satisfied-user set $\Qcal^{\ast} = \Qcal^{\ast, (n)}$ and the optimized power coefficients $ \{p_{k'}^\ast \} = \{p_{k'}^{\ast, (n)} \}$.\\
		\end{algorithmic}
		%%\vspace{-1cm} 
	\end{algorithm}

	%\vspace{-0.2cm}
	\begin{remark}
		Algorithm~\ref{alg:glob_alg} prioritizes on maintaining the QoS requirement for every user in multi-access scenarios. A finite power budget is strategically allocated to maximize the number of satisfied users before the sum-rate maximization is implemented. When the congestion appears, Algorithm~\ref{alg:glob_alg} still  provides service to  unsatisfied users for the fairness enhancement. Even though the proposed algorithm cannot guarantee a global optimum due to the inherent nonconvexity of problem~\eqref{probGlobalv1} as jointly optimizing the satisfied-user set $\Qcal$ and the power coefficients $p_k, \forall k$, it provides a good preliminary mechanism to investigate the demand-based optimization with realistic conditions where the satellite simultaneously serves many users with the same radio resources. 	
	\end{remark}
		\vspace{-0.3cm}

	%%%%%%%%%%%%%%%%%%%%%%%%%%%%%%%%%%%%%%%%
	\vspace{-0.55cm}
	\subsection{Data-Driven Approach} \label{Sec:DataDriven}
	%\vspace{-0.2cm}
	%%%%%%%%%%%%%%%%%%%%%%%%%%%%%%%%%%%%%%%%
	
	In spite of an effective solution to handle the multi-objective problem~\eqref{probGlobal} by solving an alternative version in \eqref{probGlobalv1}, Algorithm~\ref{alg:glob_alg} must update the power coefficients and the satisfied-user set after many iterations until reaching a fixed point solution. The matter  might be, therefore, still burdensome for certain practical scenarios.  In this subsection, we propose to use a neural network model that can learn  the features of Algorithm~\ref{alg:glob_alg}, and then predict the power coefficients for each realization of user locations in the satellite system with extremely low computational complexity. We assume that the power solution obtained by Algorithm~\ref{alg:glob_alg} is available for the following series of the continuous mappings:
	%\vspace{-0.2cm}
	\begin{align} 
		\mathbf{w}_{\ell} =& \tilde{\mathbf{f}}_\ell ( \{ \mathbf{h}_k \} ),\  \forall \ell \in \Kcal, \label{eq:well}\\
		\mu_{kl}  =& |\bh_k^H\bw_\ell|^2,\  \forall k,\ell \in \Kcal, \label{eq:M1}\\
		\alpha_{k}^{\ast}  =& \frac{p_k^\ast  \mu_{kk} }{\sum_{\ell\in\Kcal\backslash \{k\}}p_{\ell}^\ast  \mu_{kl}  +\sigma^2} ,\  k \in \Kcal, \label{eq:M2}\\
		p_k^{\ast}  =& f_k ( a_k^\ast, \{ \mu_{k\ell} \})  \label{eq:M3}\\
		= & \alpha_{k}^{\ast}  \frac{\sigma^2}{\mu_{kk}}  + \alpha_{k}^{\ast} 	\sum\nolimits_{\ell \in \mathcal{K} \setminus \{k\} } p_{\ell }^\ast \frac{\mu_{kl} }{\mu_{kk} }, k \in \Kcal,\non
		%\vspace{-0.2cm}
	\end{align}
	where $\tilde{\mathbf{f}}_\ell ( \{ \mathbf{h}_k \} ): \mathbb{C}^{M \times K} \rightarrow \mathbb{C}^{M}$ is a multivariate function utilized to construct a precoding vector for user~$\ell$ from the instantaneous channels. After \eqref{eq:well}, the set of the $K$ precoding vectors is constructed, which are the input to compute the channel gains in the mapping \eqref{eq:M1} if $k=\ell$. Otherwise, \eqref{eq:M1} is used to compute the strength of the mutual interference. The continuous mapping in \eqref{eq:M2} evaluates the SINR level for an arbitrarily scheduled  user. { The optimized satisfied-user set $\Qcal^\ast$, which is discrete on the definition, can be reformulated by the optimized power coefficients $\{ p_k^\ast \}$ via utilizing $\{ \alpha_k^\ast \}$ in \eqref{eq:M2}, which is continuous. It is of paramount importance to design a low-cost machine learning framework and guarantee the existence of a neural network with a finite number of neurons for our considered framework.} {Finally, the last mapping \eqref{eq:M3} points out a way to update the power coefficient of $\UEk$ in relation to the offered rate to this user and the power allocation to the other scheduled  users in a multi access scenario.}  Since a composition of the continuous mappings is also a continuous mapping \cite{hornik1989multilayer}, Lemma~\ref{lemmaExistNeural} hereby approves the existence of a unique mapping that characterizes all the above procedures.
%	%\vspace{-0.2cm}
	\begin{lemma} \label{lemmaExistNeural}
		The power coefficients obtained by Algorithm~\ref{alg:glob_alg}  are characterized by %the continuous mapping
		%\begin{equation} \label{eq:Mapping}
		$\{ p_k \} = \mathcal{F} ( \{ \mathbf{h}_k \} )$,
		%\end{equation}
		where $\mathcal{F} ( \{ \mathbf{h}_k \} )$ represents the series of the continuous mappings in \eqref{eq:well}--\eqref{eq:M3}. It implies that there exists at least a neural network to learn and predict $\mathcal{F} ( \{ \mathbf{h}_k \} )$.
		%\vspace{-0.2cm}
	\end{lemma}
	\begin{proof}
		%The proof is based on the fact that the satisfied-user set $\Qcal$ can be computed by the power coefficients $\{ p_{k'} \}$ at the convergence. The detailed proof is available in Appendix~\ref{appendixExistNeural}.
		See Appendix~\ref{appendixExistNeural}.
	\end{proof}
	%\vspace{-0.2cm}
	
	%	 As stated in the universal approximation theorem, although there are other deep neural networks that give the similar or better performance than the proposed DNN model, we selected a fully-connected DNN owing to its low computational complexity. 
	\begin{figure*}[t]
		\centering
		\includegraphics[width=1\textwidth]{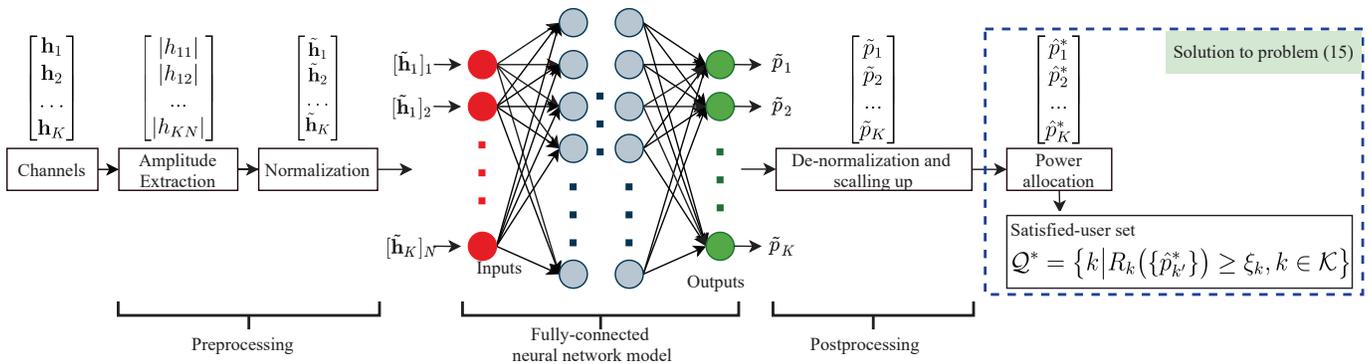}
		%\vspace{-0.5cm}
		\caption{The considered neural network architecture to learn and predict the solution to problem~\eqref{probGlobalv1}.}
		\label{fig:DNN_model}
		%\vspace{-30pt}
	\end{figure*}
	As the key point from Lemma~\ref{lemmaExistNeural}, a neural network only distills useful information from the instantaneous channels to learn the continuous mapping $\mathcal{F}(\{ \mathbf{h}_k \})$ and predict the power coefficients with low computational complexity since the satisfied-user set $\Qcal$ can be expressed as in \eqref{eq:M2}, by means of supervised learning. More precisely, different from previous works \cite{eisen2020optimal, xia2019deep}, this paper only makes use of the channel gains to learn a fully-connected neural network as the benefits of \eqref{eq:M1} conditioned by the precoding vectors as sketched in Fig.~\ref{fig:DNN_model}.\footnote{According to the universal approximation theorem \cite{hornik1989multilayer, goodfellow2016deep}, an adequately neural network can approximate a continuous mapping from a provided-input and designed-output data set. For a given accuracy, there may exist more than one neural network structures to learn  the series of continuous mappings in \eqref{eq:well}--\eqref{eq:M3}. The proof-of-concept idea in this paper is to demonstrate the effectiveness of neural networks in predicting the solution to a multi-objective optimization problem with low computational complexity.}
	
	\textit{Forward propagation}: We denote $\tilde{\bh}_{k} = [ |h_{k1}|, \ldots,  |h_{kN}|]^T \in \mathbb{R}_{+}^N$ the channel gain vector, with $h_{kn}$ denoting the $n$-th element. After that, each given realization of those channel gains are stacked into a vector as $\mathbf{x} =[\tilde{\bh}_{1}^T, \dots,\tilde{\bh}_{K}^T]^T \in\mathbb{R}_+^{KN}$. The law of conservation of energy indicates that the channel gain should be in a closed set, but their values might be extremely small due to deep fading. Subsequently, the channel gains are normalized to reducing fluctuations from the propagation environment  before utilizing them as the input to train the neural network. We numerically observe that this procedure will speed up the training phase and moderate the gradient vanishing problem. The normalized vector $\mathbf{x}_{\mathsf{in}} \in \mathbb{R}^{KN}$ is mathematically formulated from $\mathbf{x}$ as follows
	%\vspace{-0.2cm}
	\begin{equation} \label{eq:xinm}
		[\mathbf{x}_{\mathsf{in}}]_m = ([\mathbf{x}]_m -[\mathbf{x}_{\min}]_m)/( [\mathbf{x}_{\max}]_m - [\mathbf{x}_{\min}]_m),
		%\vspace{-0.2cm}
	\end{equation}
	where $[\mathbf{x}]_m$ is the $m$-th element of vector $\mathbf{x}$; $\mathbf{x}_{\min}, \mathbf{x}_{\max} \in \mathbb{R}_{+}^{KN}$ with the $m$-th element $[\mathbf{x}_{\min}]_m, [\mathbf{x}_{\max}]_m$ is respectively defined as $	[\mathbf{x}_{\min}]_m = \min \, \{ [\mathbf{x}]_m \} \mbox{ and } [\mathbf{x}_{\max}]_m = \max \, \{ [\mathbf{x}]_m \}$,
	%	\begin{equation}
	%		[\mathbf{x}_{\min}]_m = \min \, \{ [\mathbf{x}]_m \} \mbox{ and } [\mathbf{x}_{\max}]_m = \max \, \{ [\mathbf{x}]_m \},
	%	\end{equation}
	where $\{ [\mathbf{x}]_m \}$ contains all the realizations of $[\mathbf{x}]_m$ in the training data set. In the considered framework,  both the channel gains and the optimized power coefficients are normalized by applying the same methodology as in \eqref{eq:xinm}, and hence the data set is a compact set. %\textcolor{red}{Stop here....}
	%By denoting the normalization function as $\mathcal{H}^\mathtt{nml}_1(\cdot)$, the output normalized signal is given as
	%\begin{equation}
	%	\mathbf{x}^{(m)}_3=\mathcal{H}^\mathtt{nml}_1(\mathbf{x}^{(m)}_2).
	%\end{equation}
	The normalized data $\mathbf{x}_{\mathsf{in}}$ is now considered to be the input of the neural network for learning the set of weights and biases over some hidden layers. Activation functions are executed at neurons of each hidden layer to imitate  nonlinear properties in the data set. In detail, if $\mathbf{x}_{uv}$ denotes the input vector of the $u$-th neuron at the $v$-th hidden layer, then the corresponding output value is defined as $y_{uv}= f_{uv}(\mathbf{w}_{uv}^T \mathbf{x}_{uv} + b_{uv})$, where $\mathbf{w}_{uv}$ and $b_{uv}$ represent the weights and bias associated with this neuron; $f_{uv}(\cdot)$ is the activation function that imitates the nonlinear properties in a data set. After passing through the hidden layers, the output signal of the neural network block is denoted by $\tilde{\mathbf{p}} \in\mathbb{R}_+^K$. The forward propagation is deployed for both the training and testing phases. Furthermore, for the testing phase,
	the predicted data power vector $\hat{\mathbf{p}}  \in \mathbb{R}_+^K$ is obtained by denormalizing as
	%\vspace{-0.2cm}
	\begin{equation}
		[\hat{\mathbf{p}}]_k = [\tilde{\mathbf{p}}]_k \left( [\tilde{\mathbf{p}}_{\max}]_k - [\tilde{\mathbf{p}}_{\min}]_k \right)  + [\tilde{\mathbf{p}}_{\min}]_k,
		%\vspace{-0.2cm}
	\end{equation}
	where $[\cdot]_k$ is the $k$-th element of power vectors, while $[\tilde{\mathbf{p}}_{\max}]_k$ and $[\tilde{\mathbf{p}}_{\min}]_k$ are the maximum and minimum value of the power coefficient for $\UEk$ in the data set. Due to the local normalization that has generated a compact set for the power coefficient of each user, a neural network with a finite number of neurons may not guarantee the limited power budget constraint \eqref{eq:PowerConst}. To get rid of this issue, the following mapping is made as 
	%	\begin{equation}
	$[\hat{\mathbf{p}}^\ast]_k = P_{\max} [\hat{\mathbf{p}}]_k  \big/ \sum\nolimits_{k' \in \Kcal} [\hat{\mathbf{p}}]_{k'}$, 
	%	\end{equation}
	then $\sum\nolimits_{k \in \Kcal}  [\hat{\mathbf{p}}^\ast]_k = P_{\max}$ aligning with the full power consumption to maximize the sum rate \cite{Bjon13:tcit}.

	\textit{Back propagation:} It is only exploited in the training phase with the supports of the optimized power coefficients from Algorithm~\ref{alg:glob_alg}. The mean squared error (MSE) metric is adopted as the loss function for the training phase, which is defined as $\mathcal{L}^\text{MSE}( {\Theta}) = \mathbb{E} \{ \|\tilde{\mathbf{p}} - \tilde{\mathbf{p}}^{\ast}\|^2_2 \}$, 
	%	\begin{equation}\label{loss_function}
	%		\mathcal{L}^\text{MSE}( {\Theta}) = \mathbb{E} \{ \|\tilde{\mathbf{p}} - \tilde{\mathbf{p}}^{\ast}\|^2_2 \},
	%	\end{equation}
	where ${\Theta}$ is the set comprising all the weights and biases used in the neural network; $\tilde{\mathbf{p}}^\ast$ is the  vector with the optimized power coefficients obtained from Algorithm~\ref{alg:glob_alg} and after normalization. The loss function $\mathcal{L}^\text{MSE}( {\Theta})$ is expected over many realizations of different user locations and possible combinations over the $N$ overlapping beams. From a set of initial values, the weights and bias are iteratively updated by minimizing $	\mathcal{L}^\text{MSE}( {\Theta})$ with the backward propagation of the data set \cite{Chien19:TCOM}. Thanks to the benefits of supervised learning in training a neural network and to learn the multi-objective problem as analyzed in \eqref{eq:well}--\eqref{eq:M3}, Algorithm~\ref{alg:glob_alg} is utilized to generate the training data. The Adam optimization is used for backpropagation \cite{kingma2017adam}. %Notice that the offline training is studied in this paper for the sake of simplicity. 
	The momentum and babysitting the learning rate are exploited to reduce training time and get the best performance \cite{goodfellow2016deep}. 
	
	%\vspace{-0.5cm}
	%%%%%%%%%%%%%%%%%%%%%%%%%%%%%%%%%%%%%%
	\section{Satellite Communications with Linear Precoding and Water Filling} %{\color{red}[Merge with system model - new section II]}}
	%\vspace{-0.2cm}
	%%%%%%%%%%%%%%%%%%%%%%%%%%%%%%%%%%%%%%
	This section presents an application of our framework with a concrete linear precoding technique. Thanks to the semi-closed form power solution, a fine-tuning should be made to integrate the water filling method into Algorithm~\ref{alg:glob_alg} on a case-by-case basis.
	\vspace{-0.55cm}
	\subsection{Demand-based Optimization with Zero Forcing Precoding}
	%\vspace{-0.2cm}
	We now apply the ZF precoding technique to our framework, which  effectively cancels out  all mutual interference \cite{bjornson2014}.\footnote{{In this paper, the scheduled  users are selected to ensure that the channel matrix is not ill-conditioned for effectively cancelling out mutual interference once the ZF precoding technique is utilized.}} Precisely, for a given channel matrix $\bH$, the precoding matrix $\bW^{\mathrm{zf}} \in \mathbb{C}^{N \times K}$ is formulated as $\bW^{\mathrm{zf}} = \bH (\bH^H\bH )^{-1}$, 
	%	\begin{equation}\label{eq: ZF}
	%		\bW^{\mathrm{zf}} = \bH (\bH^H\bH )^{-1},
	%	\end{equation}
	and the precoding vector $\bw_k^\mathrm{zf}$ defined for $\UEk$ is calculated by $\bw_k^\mathrm{zf} = \bar{\bw}_k^\mathrm{zf} / \|\bar{\bw}_k^\mathrm{zf} \|$,
	%	\begin{equation}\label{eq: precoding_vector_ZF}
	%		\bw_k^\mathrm{zf} = \bar{\bw}_k^\mathrm{zf} / \|\bar{\bw}_k^\mathrm{zf} \|,
	%	\end{equation}
	where $\bar{\bw}_k^\mathrm{zf}$ is the $k$-th column of the matrix $\bW^{\mathrm{zf}}$. The channel capacity of $\UEk$ is reformulated from \eqref{eq:rate} to an equivalent form as
	%\vspace{-0.2cm}
	\begin{equation}\label{eq:ratezf}
		R_k^\mathrm{zf} ( p_{k} ) = B \log_2 \left(1+ \frac{p_k}{\|\bar{\bw}_k^\mathrm{zf} \|^2 \sigma^2} \right), \mbox{ [Mbps]}, \ \forall k\in\Kcal,
		%\vspace{-0.2cm}
	\end{equation}
	which demonstrates that all mutual interference from the other users to $\UEk$ is completely eliminated and the channel capacity is only the function of its own power coefficient.  
	\begin{algorithm}[t]
		\begin{algorithmic}[1] %\fontsize{11.5}{11}\selectfont
			\protect\caption{An algorithm  to obtain a local solution to problem \eqref{probGlobal} with the ZF precoding technique}
			\label{alg: prob_MMR_ZF}
			\global\long\def\algorithmicrequire{\textbf{INPUT:}}
			\REQUIRE Channel vectors $\{\bh_{k}\}$; Maximum power $P_{\max}$;  QoS requirement set $\{\xi_k\}$.
			%			\global\long\def\algorithmicrequire{\textbf{\textit{\underline{Phase 1}:}}}
			%			\REQUIRE
			\STATE Compute the precoding vectors $\{\bar{\bw}_{k}^{\mathrm{zf}} \}$ as $\bw_k^\mathrm{zf} = \bar{\bw}_k^\mathrm{zf} / \|\bar{\bw}_k^\mathrm{zf} \|$.
			%		\STATE Validate the conditions \eqref{eq:rho}, \eqref{eq:PowerConstr}
			\STATE Compute the minimum power levels $\{p_{\min,k}^{\ast} \}$ as in \eqref{eq:pkzf}.
			\IF {Condition \eqref{eq:PowerConstr} is satisfied}
			\STATE Solve problem \eqref{prob: ZF1} to obtain  $\{\tilde{p}_{k}^\ast \}$ by utilizing \eqref{eq:ratezf_wf}.
			\STATE Update  $p_k^\ast = \tilde{p}_{k}^{\ast} + p_{\min,k}^\ast, \forall k\in\Kcal$ and $\Qcal^\ast = \Kcal$.
			\ELSE
			\STATE Solve problem \eqref{prob: ZF2} with the order in \eqref{eq:AscendOrder} to obtain  $\Qcal^{\ast}$ as in \eqref{eq:Q0set} and $p_k^\ast = p_{\min,k}^\ast, \forall k \in \Qcal^\ast$. 
			\STATE Solve problem \eqref{prob: ZF3} to obtain  $p_{k}^\ast, \forall k \in \Kcal\setminus \Qcal^{\ast}$ as in \eqref{eq:OptPowerZF1}.
			\ENDIF
			%		\STATE Update solution $\Qcal^\ast = \Qcal$, $\{p_{k'}^\ast\}_{k'\in\Kcal}=\{p_{k'}\}_{k'\in\Kcal}$
			\global\long\def\algorithmicrequire{\textbf{OUTPUT:}}
			\REQUIRE  The satisfied-user set $\Qcal^{\ast}$ and the optimized power coefficients $ \{p_{k'}^\ast \} $.\\
		\end{algorithmic} 
	\end{algorithm}
	We now apply the {classical} water filling technique to tackle the joint power allocation and demand-based control as presented in Algorithm \ref{alg: prob_MMR_ZF}. Specifically, we first compute the precoding vectors $\{\bw^\mathtt{zf}_{k}\}$. From the channel capacity \eqref{eq:ratezf}, the minimum required power $p_{\min,k}^\ast$ allocates to $\UEk$ with its demand is %computed as follows
		%\vspace{-0.2cm}
	\begin{equation}\label{eq:pkzf}
		R_k^\mathrm{zf} ( p_{k} ) = \xi_k \Leftrightarrow p_{\min,k}^\ast = \alpha_k \|\bar{\bw}_k^\mathrm{zf} \|^2 \sigma^2,\quad \forall k\in\Kcal.
			%\vspace{-0.2cm}
	\end{equation}
	Thanks to the closed-form expression in \eqref{eq:pkzf}, after obtaining $\{p_{\min,k}^\ast \}$, we only need to testify the condition \eqref{eq:PowerConstr} to identify if the system can offer the QoS requirements to all the $K$ scheduled  users. Inspirited by  Algorithm~\ref{alg:glob_alg}, qualifying \eqref{eq:PowerConstr} by using $\sum_{k \in \Kcal} p_{\min,k}^\ast$  leads to the two possible cases with separated consequences. In the former case, where $\sum_{k \in \Kcal} p_{\min,k}^\ast \leq P_{\max}$, problem~\eqref{RelatedWork} should be solved to the optimal solution by the interior-point methods and a successive convex approximation in polynomial time \cite{Boyd2004a}. However, to avoid a high cost of computing the first and second derivatives required by the interior-point methods, we propose a low computational complexity algorithm that can apply for practical satellite communications. Motivated by the fact that a certain amount of the power budget will  be dedicated to guaranteeing all the scheduled  users' demands while the remaining power should spend on maximizing the sum rate, the following optimization problem is considered as
	%\vspace{-0.2cm}
	\begin{subequations} \label{prob: ZF1}
		\begin{alignat}{2}
			&\underset{\{\tilde{p}_{k'} \in \Kcal \}}{\mathrm{maximize} }&& \sum\nolimits_{k\in\Kcal}R_k^{\mathrm{zf}} (\tilde{p}_{k}) \label{prob: ZF1a}\\
			&\mbox{subject to}& & \sum\nolimits_{k\in\Kcal} \tilde{p}_k \leq P_{\max}-\sum\nolimits_{k \in\Kcal }p_{\min,k}^\ast \label{prob: ZF1b}.
		\end{alignat}
%		%\vspace{-0.2cm}
	\end{subequations}
	The constraint \eqref{prob: ZF1b} implies that the satellite only utilizes the remaining power after consuming  a portion of the power budget to ensure the $K$ scheduled  users served by their QoS requirements. From the water filling, the optimal solution to $\tilde{p}_{k}$ is computed in a semi-closed form as follows
	%\vspace{-0.2cm}
	\begin{equation}\label{eq:ratezf_wf}
		\tilde{p}_k^\ast = \max \left(0, \frac{1}{\lambda^{\ast} \ln2} - ||\bar{\bw}^\mathrm{zf}_k||^2\sigma^2 \right),\ \forall k\in\Kcal,
		%\vspace{-0.2cm}
	\end{equation}
	where $\lambda$ is the optimal solution to the Lagrange multiplier associated with the power constraint \eqref{prob: ZF1b}.
	The transmit power solution $\{p_k^\ast\}$ to problem \eqref{probGlobalv1} is attained by combining the solution $\{\tilde{p}_{k}^\ast \}$ to problem~\eqref{prob: ZF1} and the required powers $\{\hat{p}_{k'}\}$ as $p_k^\ast = \tilde{p}_k^\ast + p_{\min,k}, \forall k$. For the latter, if the condition \eqref{eq:PowerConstr} is not satisfied, i.e., $\sum_{k \in \Kcal} p_{\min,k}^\ast > P_{\max},$ we construct a heuristic mechanism to conquer problem \eqref{probGlobalv1} with the interference cancellation property of the ZF precoding technique. Accordingly, the satisfied-user set $\Qcal^{\ast}$ can be attained by solving the problem
	%\vspace{-0.2cm}
	\begin{subequations} \label{prob: ZF2}
		\begin{alignat}{2}
			&\underset{\Qcal}{\mathrm{maximize}} &\quad & |\Qcal| \label{prob: ZF2a}\\
			&\mbox{subject to}&& \sum\nolimits_{k \in\Qcal} p_{\min,k}^\ast \leq P_{\max}. \label{prob: ZF2b}
		\end{alignat}
%		%\vspace{-0.2cm}
	\end{subequations}
	From the benefits of the ZF precoding technique in mitigating mutual interference, an scheduled  user with better the spectral norm of the precoding vector than the other, i.e., computing as $\bar{\mathbf{w}}_k^{\mathrm{zf}}, \forall k,$ will consume less power, and therefore having constructive a contribution to the power resource as demonstrated in \eqref{eq:pkzf}. Hence, one can attain the solution to  problem \eqref{prob: ZF2}  by, first, sorting $\{p_{\min,k}^{\ast} \}$  in ascending order as
	%\vspace{-0.2cm}
	\begin{equation}\label{eq:AscendOrder}
		p_{\min,\pi_1}^{\ast} \leq p_{\min,\pi_2}^{\ast} \leq \ldots p_{\min,\pi_K}^{\ast},
		%\vspace{-0.2cm}
	\end{equation}
	where $\{\pi_1, \ldots, \pi_K \}$ is a permutation of $\{1, \ldots, K\}$. The satisfied-user set $\Qcal^{\ast}$ includes satisfied users, taken one by one, in the sorted-order list \eqref{eq:AscendOrder} such that
	%\vspace{-0.2cm}
	\begin{equation}\label{eq:Q0set}
		\begin{split}
			\Qcal^{\ast} = \bigg\{ k \Big|  &\sum\nolimits_{k=1}^{|\Qcal^{\ast}|}p_{\min,\pi_k}^{\ast} \leq P_{\max},\non\\ &\left.\sum\nolimits_{k=1}^{|\Qcal^{\ast}|+ 1} p_{\min,\pi_k}^{\ast} > P_{\max}, k \in \Kcal \right\}.
		\end{split}
		%\vspace{-0.2cm}
	\end{equation}
	The following power budget of the satellite after allocating to the satisfied users in $\Qcal^{\ast}$ with their QoS requirements
	%\begin{equation} \label{eq:Power1}
	$\widetilde{P}_{\max} = P_\text{max} - \sum\nolimits_{k=1}^{|\Qcal^{\ast}|}p_{\min,\pi_k}^\ast$
	%\end{equation}
	is dedicated to enhancing the data throughout for the remaining users. It results in $p_k^\ast = p_{\min,k}^\ast, \forall k \in \Qcal^{\ast}$. 
	The optimal power allocation to the unsatisfied users in  $\Kcal\backslash\Qcal^{\ast}$ is attained by performing the water filling method for the optimization problem as
	%\vspace{-0.2cm}
	\begin{subequations} \label{prob: ZF3}
		\begin{alignat}{2}
			\underset{\{p_{k'} \geq 0, k' \in \Kcal\backslash\Qcal^{\ast} \}}{\mathrm{maximize}}  & \sum\nolimits_{k\in\Kcal\backslash\Qcal^{\ast}}R_k^{\mathrm{zf}}(p_{k'}) \\
			\mbox{subject to} \quad& \sum\nolimits_{k \in\Kcal\backslash\Qcal^{\ast}}p_k \leq \widetilde{P}_{\max}. \label{ZF3C1}
		\end{alignat}
%		%\vspace{-0.2cm}
	\end{subequations}
	We emphasize that the water filling method can be applied to obtain the global solution to problem~\eqref{prob: ZF3}, for which the optimal power $p_{k}^{\ast}$ of $\UEk$ is computed in a semi closed form as
	%\vspace{-0.2cm}
	\begin{equation} \label{eq:OptPowerZF1}
		p_k^{\ast} = \max \left(0, \frac{1}{\tilde{\lambda}^{\ast} \ln2} - ||\bar{\bw}^\mathrm{zf}_k||^2\sigma^2 \right),\ \forall k\in\Kcal\setminus \Qcal^{\ast},
		%\vspace{-0.2cm}
	\end{equation}
	where $\tilde{\lambda}^{\ast}$ is the optimal solution to the Lagrange multiplier associated with the constraint \eqref{ZF3C1}. By completely mitigating mutual interference among the $K$ scheduled  users, Algorithm~\ref{alg: prob_MMR_ZF} has the main computational complexity on searching for the optimal Lagrangian multipliers $\lambda^\ast$ and $\tilde{\lambda}^{\ast}$.
	
	%	\vspace{-0.25cm}
	\vspace{-0.55cm}
	\subsection{Demand-based Optimization with Regularized Zero-Forcing Precoding}
	%\vspace{-0.2cm}
	%As demonstrated above, the water filling method has offered a low computational complexity for the joint sum rate and satisfied-user set by effectively canceling out mutual interference from the ZF precoding technique. 
	We now inherit the major benefits of the water filling method  to design a heuristic algorithm for the RZF technique. From the channel matrix $\bH$, the precoding matrx is formulated  as $\bW^{\mathrm{rzf}} = \bH ( \bH^H\bH + \frac{K\sigma^2}{P_{\max}} \bI_K )^{-1}$, 
	%\footnote{A system may use the maximum ration transmission (MRT) precoding technique, which is defined as $\bW^{\mathrm{mrt}} = \bH$. The MRT precoding vector dedicated to scheduled  user~$k$ is in parallel with its propagation channel, and therefore this technique focuses on maximizing a particular channel gain. It has low computational complexity without the burden of an matrix inverse, but offering significantly lower data throughput than the ZF and RZF precoding techniques, especially for MB-HTS systems \cite{zhang2019joint,Bjon13:tcit}.}
	%\begin{subequations}\label{}
	%	\begin{alignat}{2}
	%		&\bW^\mathrm{rzf} =  \bH( \bH^H\bH + \alpha \bI)^{-1},\\
	%		&\bW^\mathrm{mrt} = \bH,
	%	\end{alignat}
	%\end{subequations}
	%	\begin{equation}
	%		\bW^{\mathrm{rzf}} = \bH \left( \bH^H\bH + \frac{K\sigma^2}{P_{\max}} \bI_K \right)^{-1},
	%	\end{equation}
	where $\mathbf{I}_K$ is the identity matrix of size $K \times K$ and the RZF precoding vector defined for $\UEk$ is $\bw_k^{\mathrm{rzf}} = \bar{\bw}_k^{\mathrm{rzf}} / \|\bar{\bw}_k^{\mathrm{rzf}} \|$, 
	%	\begin{equation}\label{eq:PrecodingMRTRZF}
	%		\bw_k^{\mathrm{rzf}} = \bar{\bw}_k^{\mathrm{rzf}} / \|\bar{\bw}_k^{\mathrm{rzf}} \|,
	%	\end{equation}
	where $\bar{\bw}^{\mathrm{rzf}}$ is the $k$-th column of matrix $\bW^{\mathrm{rzf}}$. The RZF precoding technique does not entirely mitigate mutual interference with regard to its own benefits. Precisely,  it balances the transmit power and mutual interference up to a level \cite{bjornson2014}. Hence, the network should utilize \eqref{eq:rate} to evaluate the channel capacity. To exploit the water-filling method for the power control, with $\forall k\in\Kcal$, \eqref{eq:rate} is upper bounded by
	%\vspace{-0.2cm}
	\begin{align} \label{eq:Upperbound}
		R_k ( \{ p_{k'} \} ) &\leq B \log_2 \left(1+ \frac{p_k|\bh_k^H \bw_k^{\mathrm{rzf}}|^2}{\sigma^2} \right), \mbox{[Mbps]},\non\\
		 &\triangleq\widetilde{R}_k (  p_{k} ) 
		%\vspace{-0.2cm}
	\end{align}
	by neglecting mutual interference from the other scheduled  users. We stress that the upper bound on the channel capacity in \eqref{eq:Upperbound} aligns with the standard form that the water filling method can perform as shown in Algorithm~\ref{alg: MMR_RZF_MRT}. Because of the mutual interference, we should introduce a tolerable rate accuracy for $\UEk$, denoted by $\omega_k \geq 0$. Alternatively, the relaxed-QoS requirement of $\UEk$ should be $\xi_k + \omega_k$.  Similar to \eqref{eq:pkzf}, we thus compute the minimum required power $p_{\min,k}^\ast$ by using \eqref{eq:Upperbound} as
	%	\begin{equation}
	$p_{\min,k}^\ast = (2^{(\xi_k + \omega_k)/B}-1) \| \mathbf{w}_k^{\mathrm{rzf}} \|^2 \sigma^2, \forall k \in \Kcal$, 
	%	\end{equation}
	then if the conditions \eqref{eq:rho} and \eqref{eq:PowerConstr} hold, Algorithm~\ref{alg: MMR_RZF_MRT} solves the sum-rate optimization problem as 
	%\vspace{-0.2cm}
	\begin{subequations} \label{prob: RZF2}
		\begin{alignat}{2}
			&\underset{\{\tilde{p}_{k'} \in \Kcal \}}{\mathrm{maximize} }&& \sum\nolimits_{k\in\Kcal} \widetilde{R}_k (\tilde{p}_{k}) \label{prob: RZF2a}\\
			&\mbox{subject to}& & \sum\nolimits_{k\in\Kcal} \tilde{p}_k \leq P_{\max}-\sum\nolimits_{k \in\Kcal }p_{\min,k}^\ast \label{prob: RZF2b}.
		\end{alignat}
		%\vspace{-0.2cm}
	\end{subequations}
	Let us denote $\{ \tilde{p}_k^{\ast} \}$ the solution to problem~\eqref{prob: RZF2} that is concretely expressed in a semi-closed form as
	%\vspace{-0.2cm}
	\begin{equation}
		\tilde{p}_k^{\ast} = \max\left( 0, \frac{1}{\mu^{\ast} \ln2} - \frac{\sigma^2}{|\bh_k^H \bw_k^{\mathrm{rzf}}|^2}  \right),\ \forall k \in \Kcal,
		%\vspace{-0.2cm}
	\end{equation}
	where $\mu^\ast$ is the optimal Lagrange multiplier associated with the constraint \eqref{prob: RZF2b}, then we obtain the optimized power coefficient of $\UEk$ as $p_k^{\ast} = \tilde{p}_k^{\ast} + p_{\min,k}^{\ast}$ and the satisfied-user set $\Qcal^\ast = \{k|k\in\Kcal, R_k(\{p_{k'}^\ast \}) \geq \xi_k\}$ (Step~5 of Algorithm~\ref{alg: MMR_RZF_MRT}). 
	\begin{algorithm}[t]
		\begin{algorithmic}[1] %\fontsize{11.5}{11}\selectfont
			\protect\caption{An algorithm to obtain a local solution to problem \eqref{probGlobal} with the RZF precoding technique}
			\label{alg: MMR_RZF_MRT}
			\global\long\def\algorithmicrequire{\textbf{INPUT:}}
			\REQUIRE Channel vectors $\{\bh_{k}\}$; Maximum power $P_{\max}$;  QoS requirement set $\{\xi_k\}$; Tolerable rate accuracy set $\{ \omega_k \}$.
			\STATE Compute the precoding vectors $\{\bw_{k}^{\mathrm{rzf}}\}$ as $\bw_k^{\mathrm{rzf}} = \bar{\bw}_k^{\mathrm{rzf}} / \|\bar{\bw}_k^{\mathrm{rzf}} \|$. %in \eqref{eq:PrecodingMRTRZF}.
			\STATE Compute $\{{p}_{\min, k}^{\ast}| k\in\Kcal, R_k(\{p_{k'}\}) = \xi_k +\omega_k \}$; the matrices $\mathbf{R}, \mathbf{Q},$ and the vector $\pmb{\nu}$.% based on the channel vectors $\{ \mathbf{h}_{k'} \}$, the precoding vectors  $\{\bw_{k}\}$, and the relaxed-rate demand set $\{\xi_k+\omega_k\}$.
			\IF {Conditions \eqref{eq:rho} and \eqref{eq:PowerConstr} are satisfied}
			\STATE Solve problem \eqref{prob: RZF2} to obtain  $\{\tilde{p}_{k}^\ast\}$.
			\STATE Update $p_k^\ast = \tilde{p}^\ast_k + {p}^{\ast}_{\min,k}, \forall k\in\Kcal$ and  $\Qcal^\ast =\{k|k\in\Kcal, R_k(\{p_{k'}\}) \geq \xi_k\}$.
			\ELSE
			\STATE Solve \eqref{RelatedWorkSumRate} with the upper bounded channel capacity in \eqref{eq:Upperbound} to obtain $\{ p_k^{\ast, (0)}\}$ and $\widetilde{\Qcal}^{\ast,(0)} = \{k|k\in\Kcal, R_k(\{p_{k'}^{\ast,(0) }\}) \geq \xi_k + \omega_k \}$. 
			\STATE Initial the accuracy $\delta = |\widetilde{\Qcal}^{\ast,(0)}|$ and set $n=0$.
			\WHILE { $\delta \neq 0$}
			\STATE Set iteration index $n=n+1$.
			\STATE Compute $p_{k}^{\ast, (n-1)}$ for user~$k\in\widetilde{\Qcal}^{\ast,(n-1)}$ as in \eqref{eq:powerkn}.	
			\STATE Solve problem~\eqref{prob: RZF1} to obtain $\{ p_k^{\ast,(n)}\} , \forall k \in\Kcal\backslash\widetilde{\Qcal}^{\ast,(n-1)}$.
			\STATE Update $\widetilde{\Qcal}^{\ast,(n)} = \widetilde{\Qcal}^{\ast,(n-1)} \cup  \widetilde{\Qcal}_1^{\ast,(n)}$. %, where $\widetilde{\Qcal}_1^{\ast,(n)}$ is defined in \eqref{eq:Q1astn}.
			\STATE Update the accuracy $\delta = |\bar{\Qcal}^{\ast, (n)} | -  |\bar{\Qcal}^{\ast, (n-1)} |$.
			\ENDWHILE
			\STATE Update $\{p^*_{k}\} = \{p_{k}^{\ast,(n)}\}$, and $\Qcal^\ast = \{k| R_k(\{p^\ast_{k'}\}) \geq \xi_k, \forall k\in\bar{\Qcal}^{\ast,(n)}\}$.
			\ENDIF
			\global\long\def\algorithmicrequire{\textbf{OUTPUT:}}
			\REQUIRE  The satisfied-user set $\Qcal^{\ast}$ and the optimized power coefficients $ \{p_{k}^\ast \}$.\\
		\end{algorithmic}
	%%\vspace{-0.1cm}
	\end{algorithm}
	If the conditions \eqref{eq:rho} and \eqref{eq:PowerConstr} are not satisfied, then the congestion issue appears. Algorithm~\ref{alg: MMR_RZF_MRT} initially solves problem~\eqref{RelatedWorkSumRate} by applying the water filling method to obtain the optimized power coefficients $\{ p_k^{\ast, (0)}\}$ and the relaxed satisfied-user set $\widetilde{\Qcal}^{\ast,(0)} = \{k|k\in\Kcal, R_k(\{p_{k'}^{\ast,(0) }\}) \geq \xi_k + \omega_k \}$. At iteration~$n$, let us decompose $\widetilde{\Qcal}^{\ast,(n-1)} = \widetilde{\Qcal}^{\ast,(n-2)} \cup \widetilde{\Qcal}^{\ast,(n-1)}_1$ where $\widetilde{\Qcal}^{\ast,(n-2)}$ and $\widetilde{\Qcal}^{\ast,(n-1)}$ contains the users satisfied their relaxed-QoS requirements up to iteration~$n-2$ and the new ones at iteration~$n-1$, respectively. Notice that $p_k^{\ast, (n-1)} = p_k^{\ast, (n-2)}$ if $k \in \widetilde{\Qcal}^{\ast,(n-2)} $ and $\widetilde{\Qcal}^{\ast,(n-2)} = \varnothing$ as $n=1$.  From the optimized power solution $\{ p_{k}^{\ast, (n-1)} \}$ to problem \eqref{prob: RZF1}, we can truncate the transmit power of new satisfied user~$k$ to as
	%\vspace{-0.2cm}
	\begin{align} \label{eq:powerkn}
%		\begin{split}
			p_k^{\ast,(n-1)} = (2^{\frac{(\xi_k + \omega_k)}{B}} -1)&\frac{\sum_{\ell\in\Kcal \setminus \{k\}}p_\ell^{\ast,(n-1)} |\bh_k^H\bw_\ell^{\mathrm{rzf}}|^2 +\sigma^2}{|\mathbf{h}_k \mathbf{w}_k^{\mathrm{rzf}}|^2},\non\\
			 &\qquad  \qquad  \quad  k \in \widetilde{\Qcal}_1^{\ast,(n-1)},
%		\end{split}
		%\vspace{-0.2cm}
	\end{align}
	and the dedicated power $\widetilde{P}_{\max}^{(n)} = P_{\max} - \sum_{k \in \widetilde{\Qcal}^{\ast,(n-1)} } 	p_k^{\ast,(n-1)}$ are utilized to improve the remaining scheduled  users by solving the following optimization problem
	%\vspace{-0.2cm}
	\begin{subequations} \label{prob: RZF1}
		\begin{alignat}{2}
			\underset{\{p_{k'}^{(n)} \geq 0, k' \in \Kcal\backslash\widetilde{\Qcal}^{\ast,(n-1)} \}}{\mathrm{maximize}}  & \sum\nolimits_{k\in\Kcal\backslash\widetilde{\Qcal}^{\ast,(n-1)}} \widetilde{R}_k(p_{k'}^{(n)}) \\
			\mbox{subject to} \qquad& \sum\nolimits_{k \in\Kcal\backslash\widetilde{\Qcal}^{\ast,(n-1)}}p_k^{(n)} \leq \widetilde{P}_{\max}^{(n)}. \label{eq:Power1}
			%\vspace{-0.2cm}
		\end{alignat}
%		%\vspace{-0.2cm}
	\end{subequations} 
	By denoting $\{ p_k^{\ast,(n)} \}, \forall k \in \Kcal \setminus \widetilde{\Qcal}^{\ast,(n-1)},$ the solution to  problem~\eqref{prob: RZF1}, which is computed in a semi-closed form as
	%\vspace{-0.2cm}
	\begin{equation}
		p_k^{\ast,(n)} = \max\left( 0, \frac{1}{\mu^{\ast,(n)} \ln2} - \frac{\sigma^2}{|\bh_k^H \bw_k^{\mathrm{rzf}}|^2}  \right), %\forall k \in \Kcal \setminus \widetilde{\Qcal}^{\ast,(n-1)},
		%\vspace{-0.2cm}
	\end{equation}
	where $\mu^{\ast,(n)}$ is the optimal Lagrange multiplier associated with the constraint~\eqref{eq:Power1}, then the algorithm enables to boost data throughput for the unsatisfied users.  Algorithm~\ref{alg: MMR_RZF_MRT} terminates as the cardinality of the satisfied-user set retains, i.e., $ |\bar{\Qcal}^{\ast, (n)} | =  |\bar{\Qcal}^{\ast, (n-1)} |$.

	\begin{figure*}[t]
		\begin{minipage}{0.32\textwidth}
			\centering
			\includegraphics[width=1\textwidth]{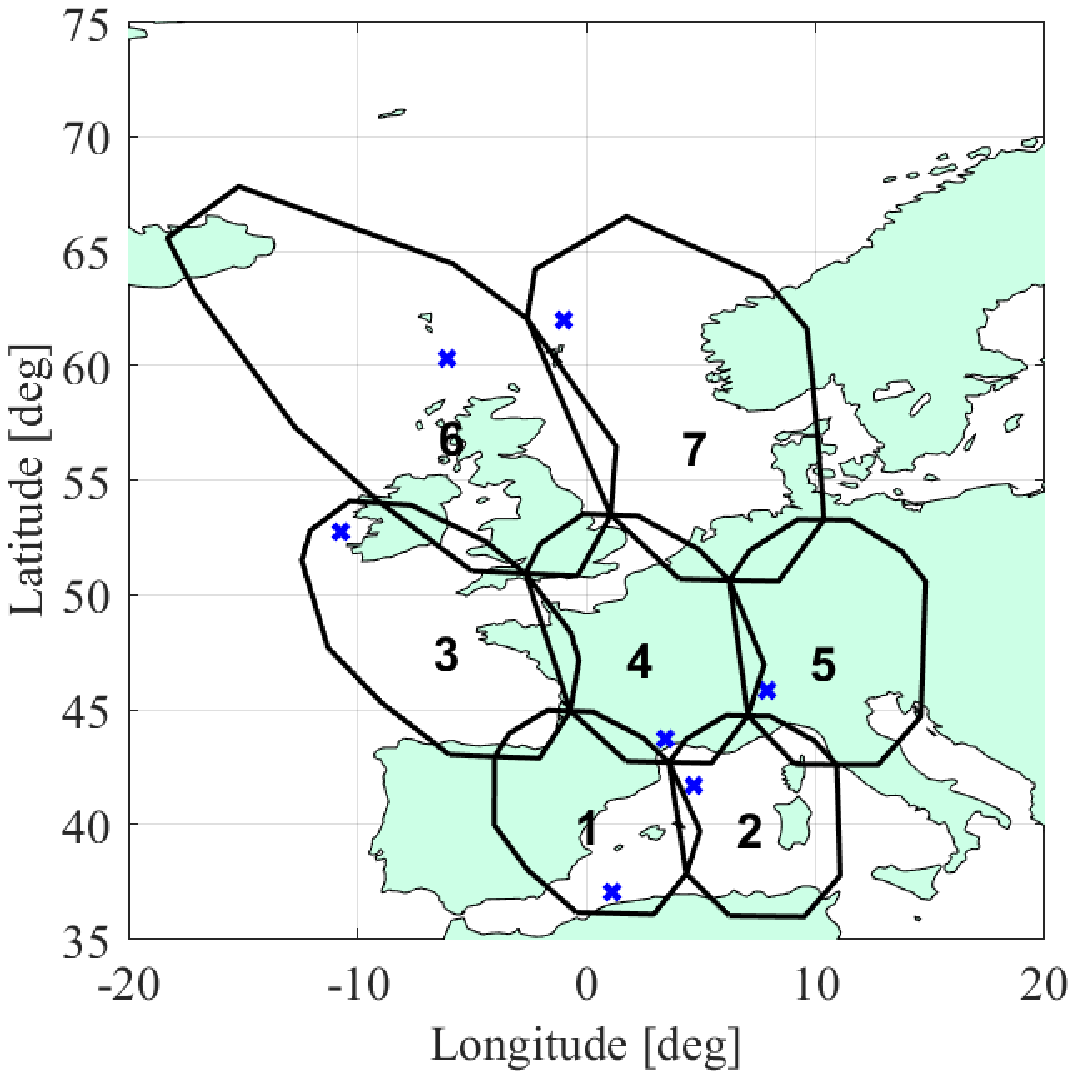} \\
			\centering $(a)$
			%\vspace{-0.1cm}
		\end{minipage}
		\begin{minipage}{0.32\textwidth}
			\centering
			\includegraphics[width=1\textwidth]{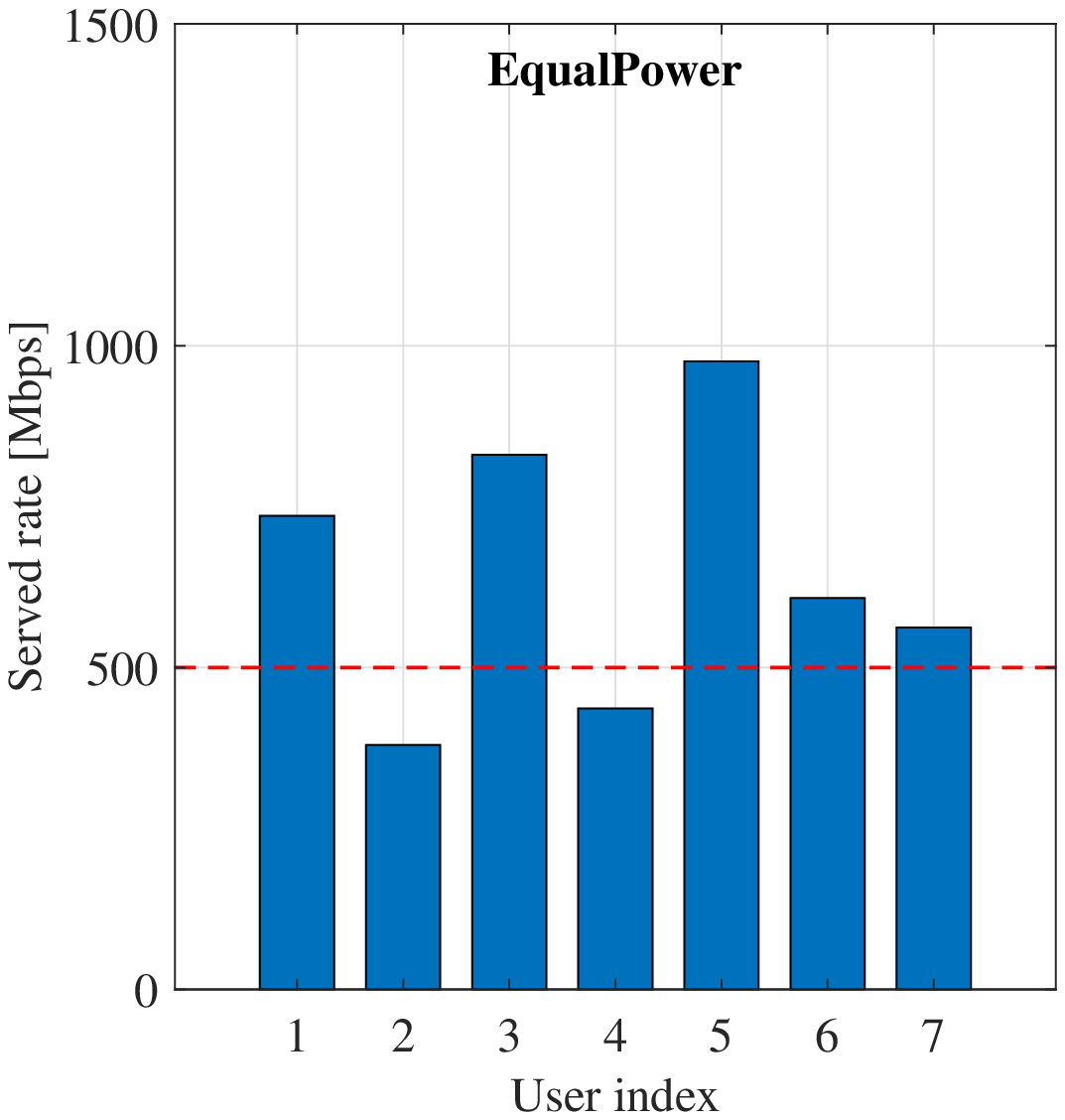}\\
			\centering  $(c)$
			%\vspace{-0.1cm}
		\end{minipage}
		\begin{minipage}{0.32\textwidth}
			\centering
			\includegraphics[width=1\textwidth]{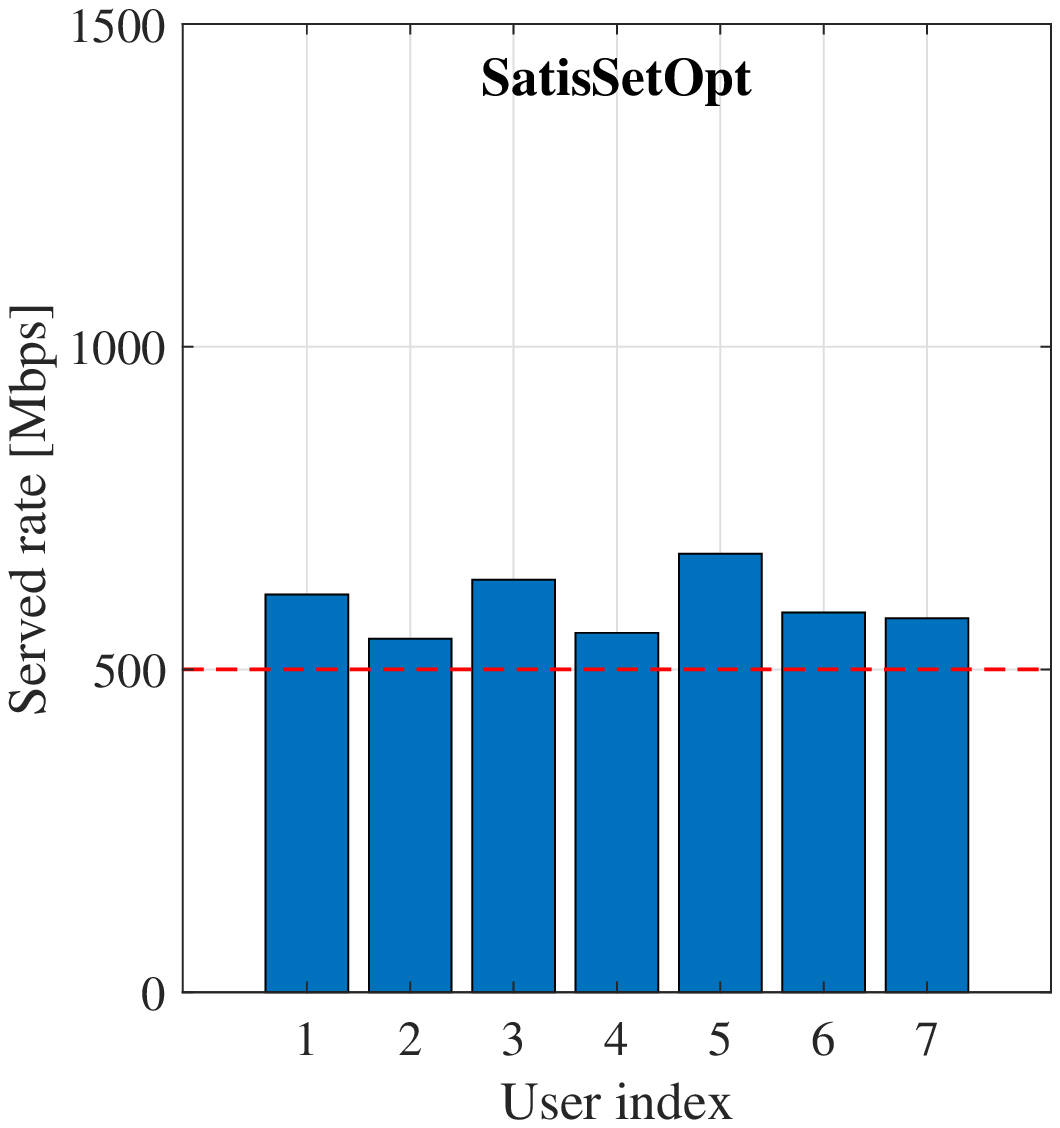} \\
			\centering $(e)$
			%\vspace{-0.1cm}
		\end{minipage}
		\begin{minipage}{0.33\textwidth}
			%		\centering
			\includegraphics[width=1\textwidth]{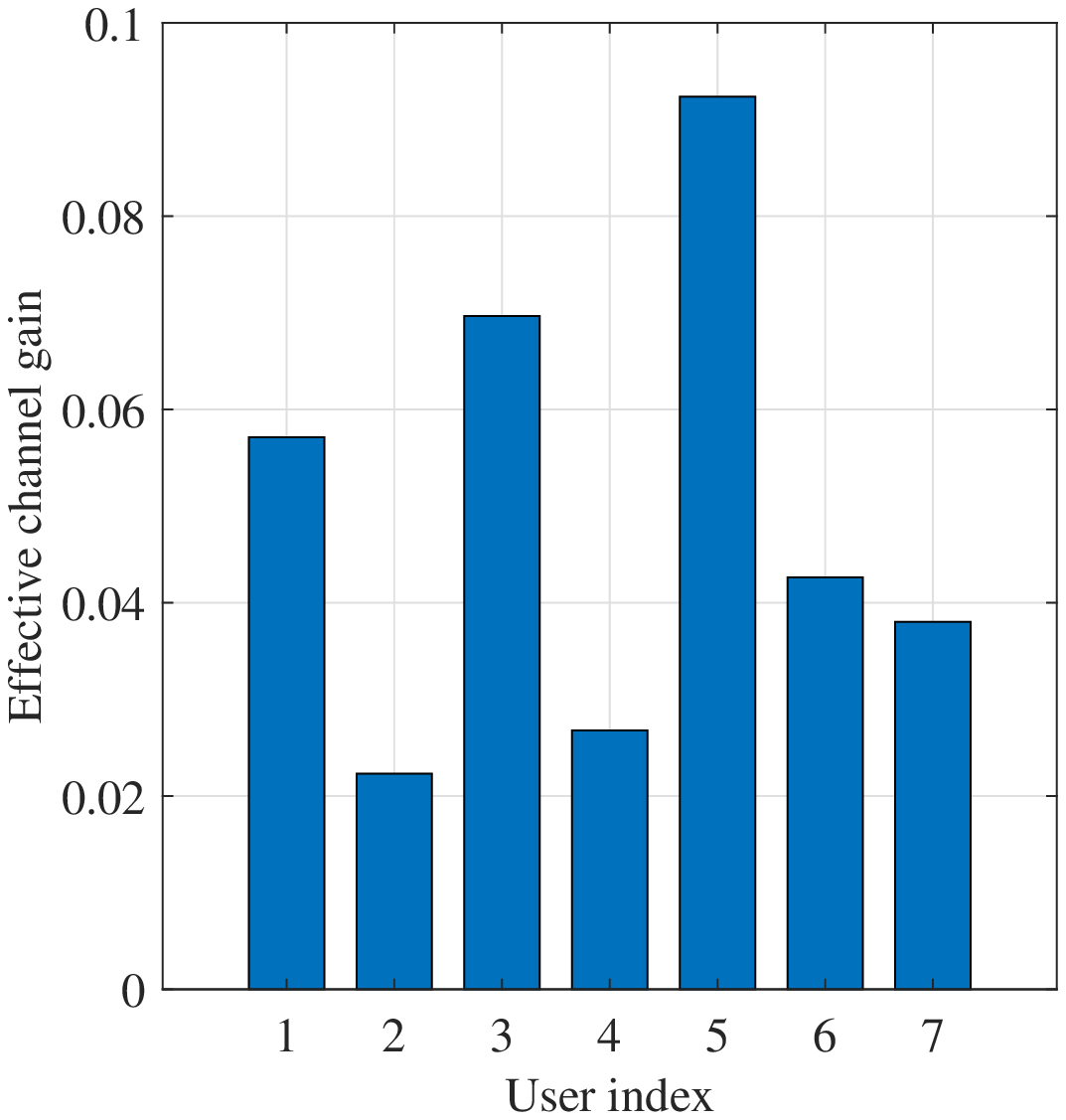} \\
			\centering $(b)$
			%\vspace{-0.1cm}
		\end{minipage}
		\begin{minipage}{0.33\textwidth}
			\centering
			\includegraphics[width=1\textwidth]{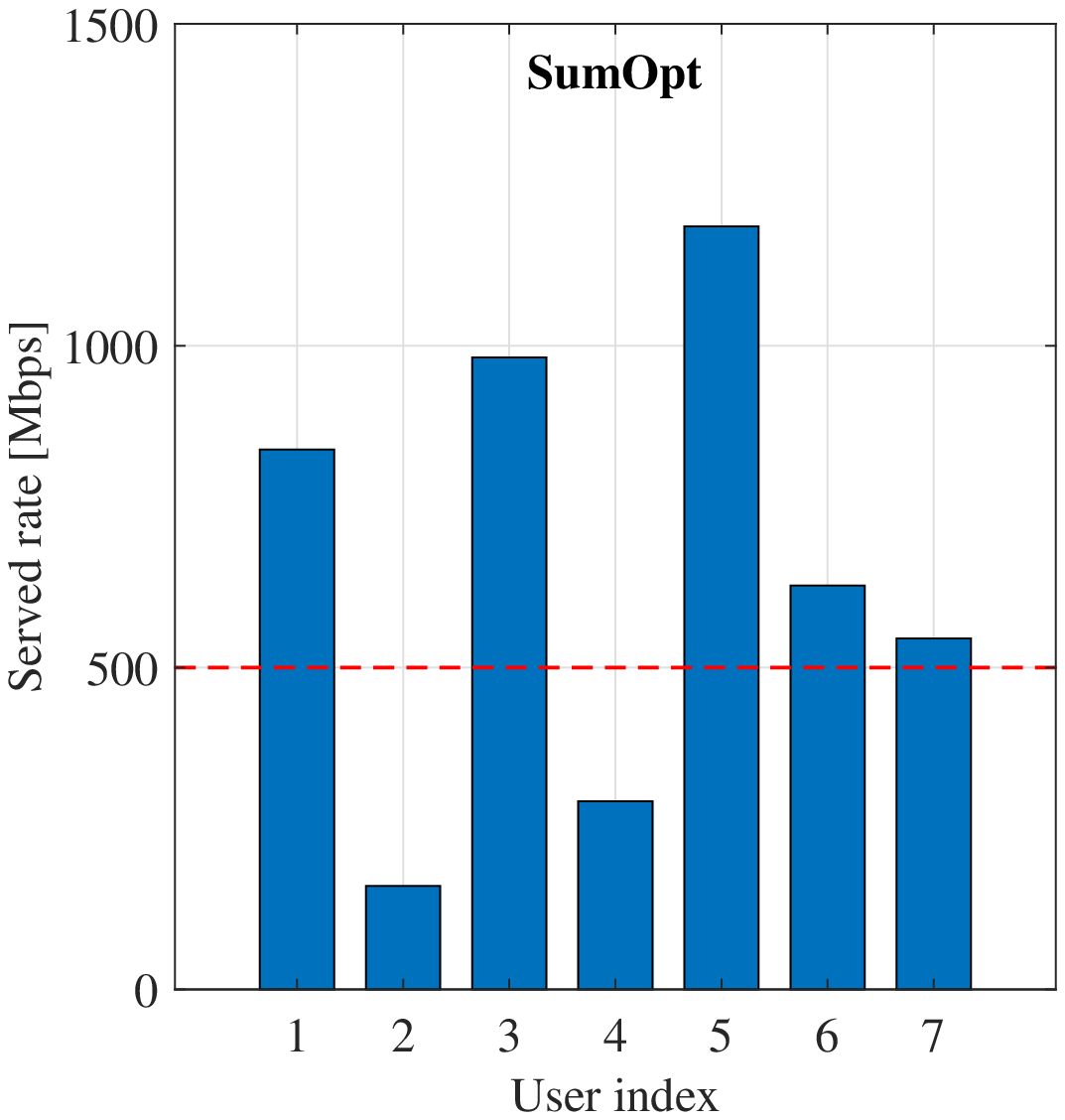}\\
			\centering $(d)$
			%\vspace{-0.1cm}
		\end{minipage}
		\begin{minipage}{0.328\textwidth}
			\centering
			\includegraphics[width=1\textwidth]{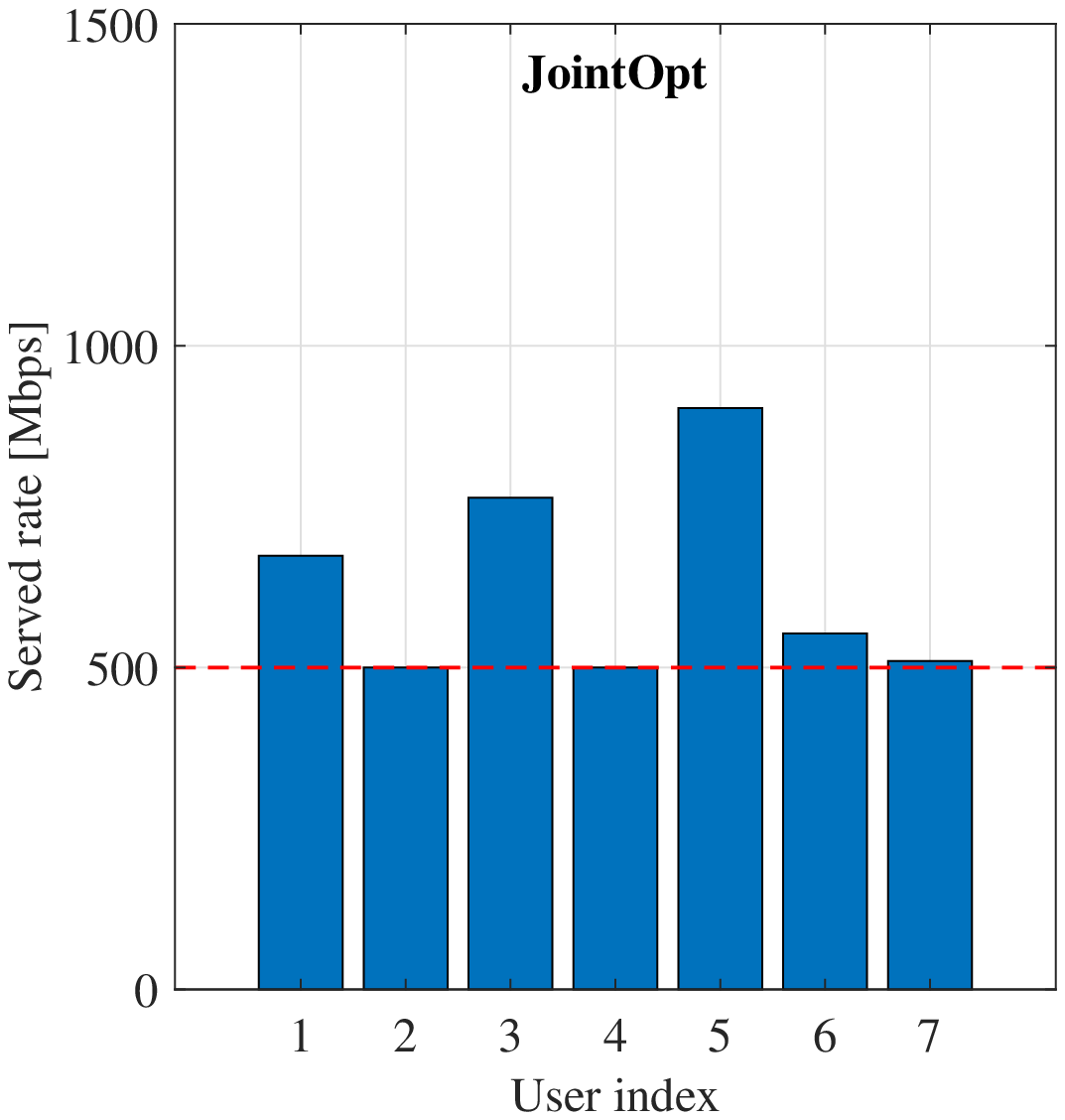} \\
			\centering $(f)$
			%\vspace{-0.1cm}
		\end{minipage}
		%		%\vspace{-0.5cm}
		\caption{A snapshot of the served rate per user [Mbps] with the ZF precoding technique for the different benchmarks and the QoS requirement per user $500$~[Mbps]: $(a)$ the users' locations; $(b)$ the effective channel gains; $(c)$ the equal power allocation (EqualPower); $(d)$ the sum rate maximization (SumOpt); $(e)$ the satisfied-user set maximization (SatisSetOpt); and $(f)$ the joint sum rate and satisfied-user set maximization (JointOpt).}
		\label{Fig:Snapshot}
		%\vspace{-30pt}
	\end{figure*}
	%\vspace{-0.4cm}
	\section{Numerical Results} \label{Sec:Results}
	%\vspace{-0.2cm}
	We consider a GEO satellite system consisting of $N=7$ beams that serve at most $K=7$ scheduled  users in each coherence time interval.\footnote{\textcolor{black}{In practice the entire system is split in terms of geographical coverage or carriers due to the limited feeder link bandwidth and each part is handled by a different gateway. This practical constraint makes our numerical results reasonable in terms of a single gateway managing a cluster of beams.}} {\hili Specifically, in the simulation section, we investigate a satellite system with a total of $35000$ users evenly distributed across beams, i.e., with approximately $5000$ users laying on each beam coverage region. At each time slot, a random user per beam selected for consideration in the power allocation problem (unicast user scheduling). For sake of the simplicity, there is no user mobility. The parameters associated with the satellite and the beam radiation patterns are provided by ESA in the context of \cite{ESA}. In detail, the radiation patterns are based on a Defocused Phased Array-Fed Reflector (PAFR), with reflector size of $2.2$m and an array diameter of roughly $1.2$m. The antenna array before the reflector is a circular array with the space of $2\times$ carrier wavelength  and $511$ elements.} The satellite location is at $13^\circ$~E, and the system operates at Ka band, for which the carrier frequency is $20$~[GHz] \cite{CGD}. The system bandwidth is $500$~[MHz] and the satellite height is $35,786$~[km]. The maximum transmit power is $P_{\max}  = 23.37$~{[dBW]} corresponding to the average beamforming gain $44.4$~[dBi] and the effective isotropic radiated power (EIRP) $-27$ [dBW/Hz]. The receive antenna diameter is $0.6$~m and the noise power per user is $-118.3$~[dB]. For the data-driven approach, we construct a fully-connected neural network comprising  hidden layers with $128$ and $64$ neurons, respectively. The rectified linear unit (ReLU) is used as the activation function. The $25000$ realizations of different user locations are captured for the training phase to learn the continuous mappings in Section~\ref{Sec:DataDriven}. We also use $10000$ realizations for the testing phase to demonstrate the effectiveness of our proposed data-driven approach. Simulation results are implemented by using MATLAB on a personal Dell  Latitude 5510 laptop with CPU Intel Core(TM) i7-10610U @ 1.8-2.3~[GHz], and 16~[GB] RAM.
	\begin{figure*}[t]
		\begin{minipage}{0.325\textwidth}
			\centering
			\includegraphics[width= 1.1 \textwidth]{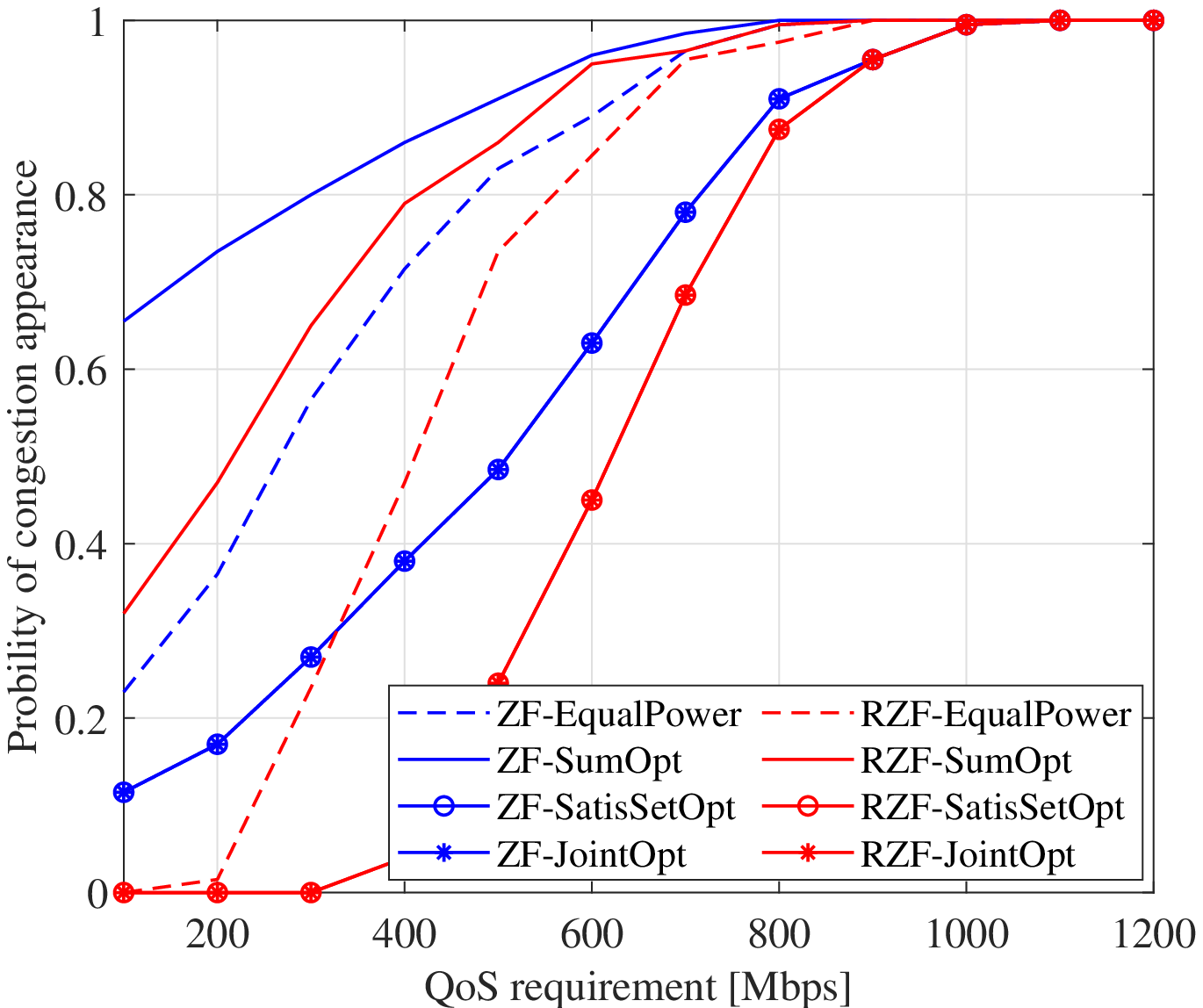}
			%\vspace{-1cm}
			\caption{The probability of congestion appearance versus the QoS requirement.}
			\label{Fig:ProbInfea}
		\end{minipage}
		\hfill
		\begin{minipage}{0.325\textwidth}
			\centering
			\includegraphics[width= 1.1 \textwidth]{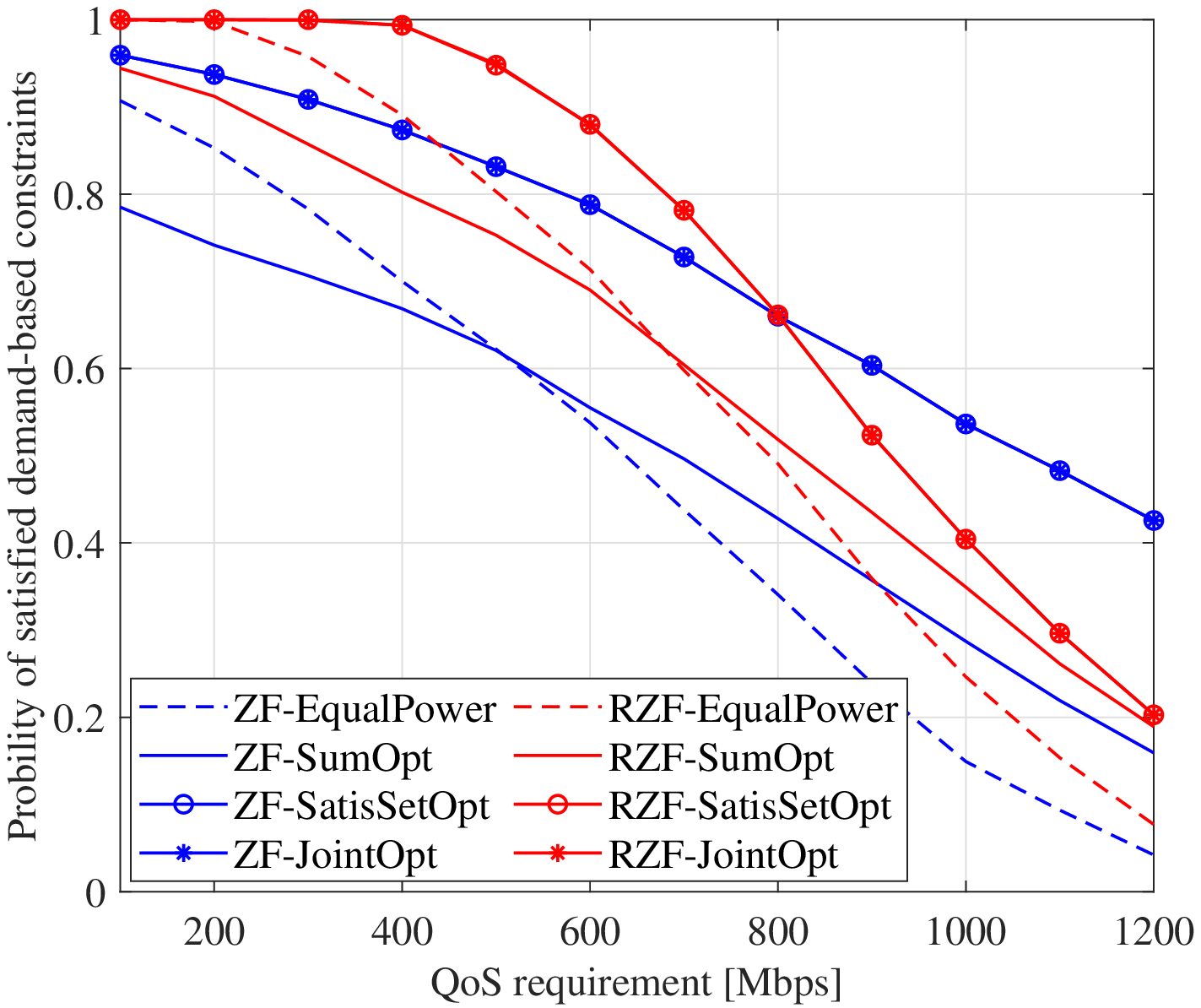}
			%\vspace{-1cm}
			\caption{The probability of satisfied users versus the QoS requirement.}
			\label{Fig:ProbSatis}
		\end{minipage}
		\hfill
		\begin{minipage}{0.325\textwidth}
			\centering
			\includegraphics[width= 1.1 \textwidth]{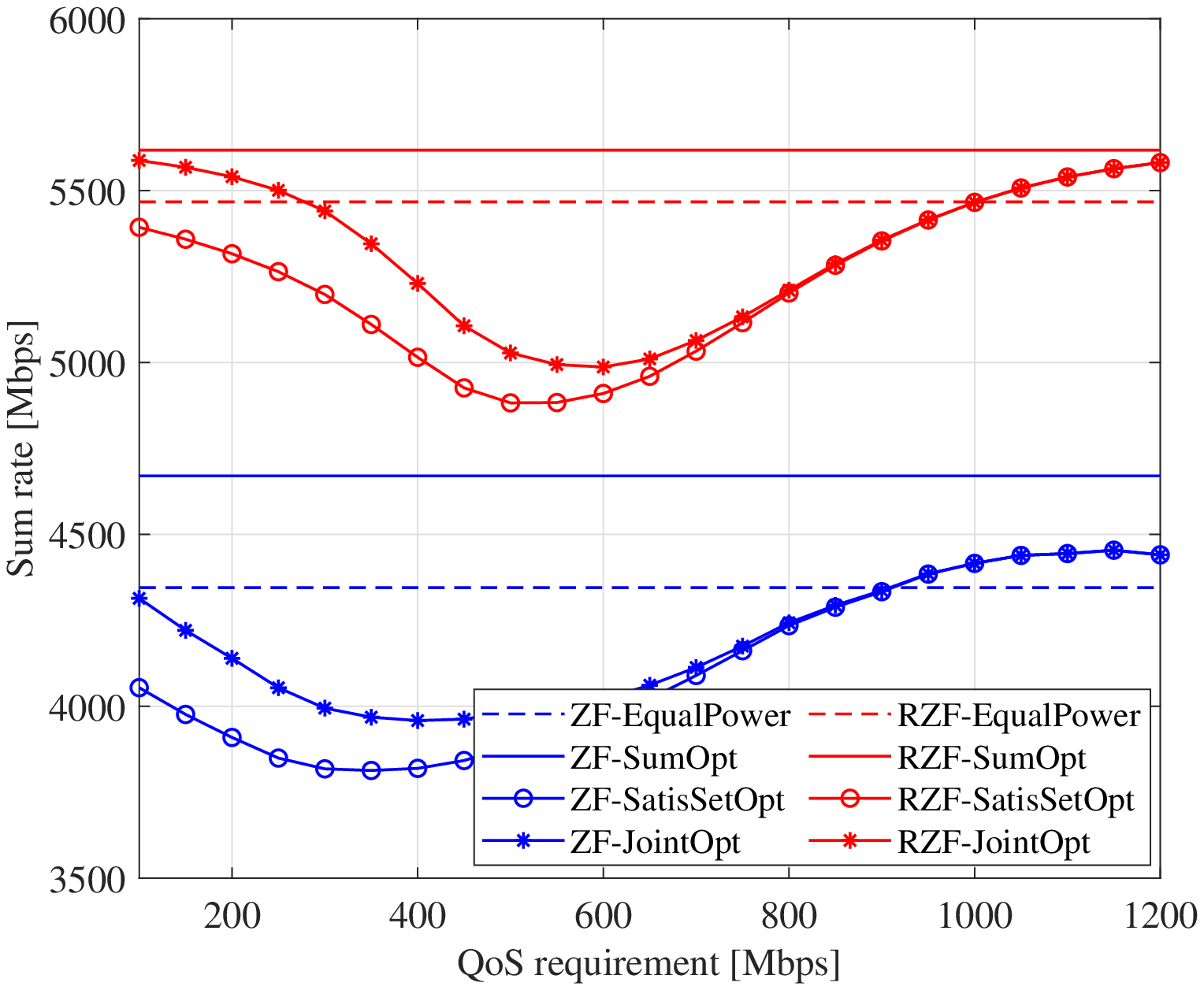}
			%\vspace{-1cm}
			\caption{The sum rate versus the QoS requirement.}
			\label{Fig:Sum}
		\end{minipage}
		\hfill
		\begin{minipage}{0.325\textwidth}
			\centering
			\includegraphics[width= 1.1 \textwidth]{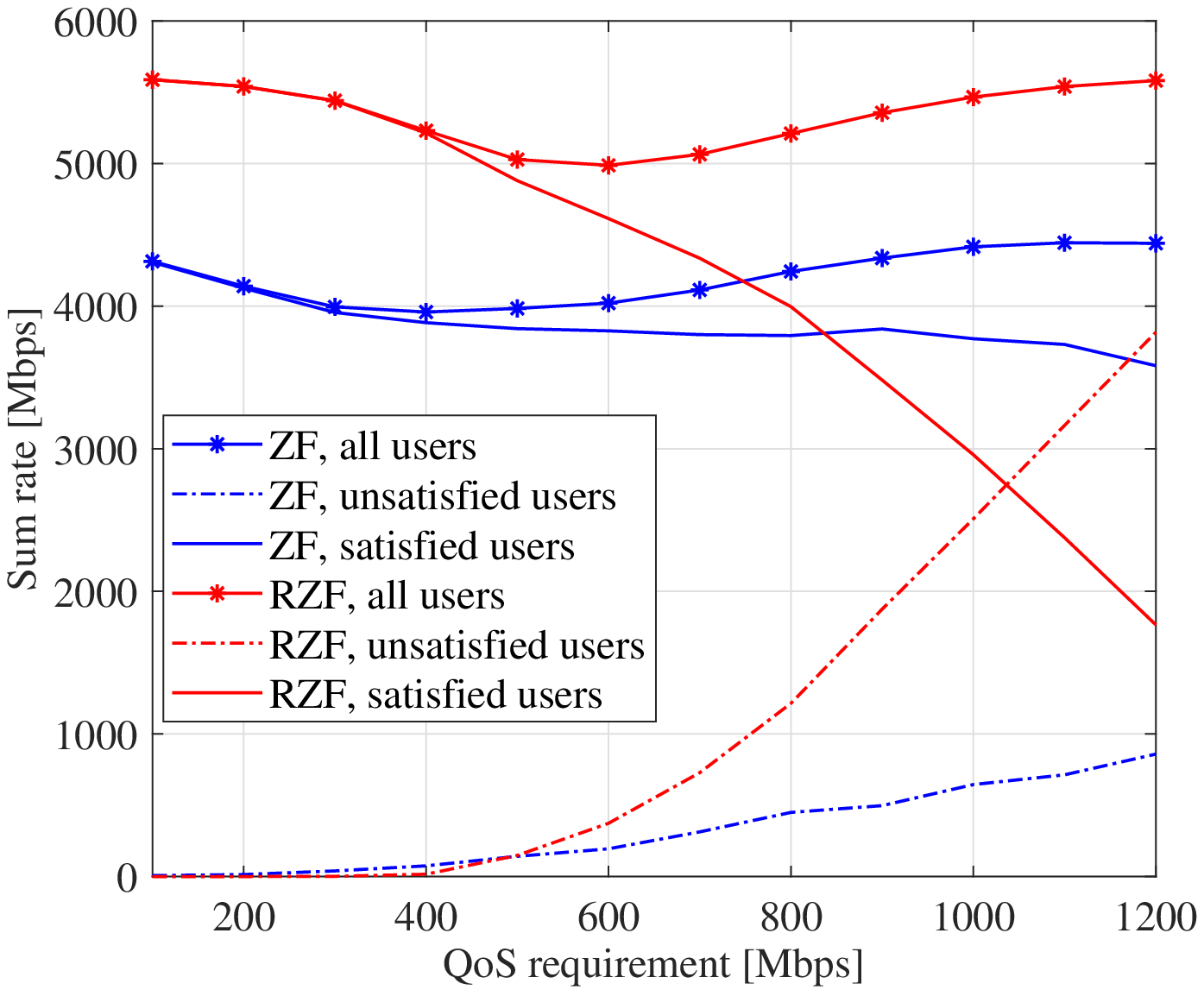}
			%\vspace{-1cm}
			\caption{{\hili The sum rate versus the QoS requirement obtained by JointOpt.}}
			\label{Fig:SumSeperateRate}
		\end{minipage}
		\hfill
		\begin{minipage}{0.325\textwidth}
			\centering
			\includegraphics[width= 1.1 \textwidth]{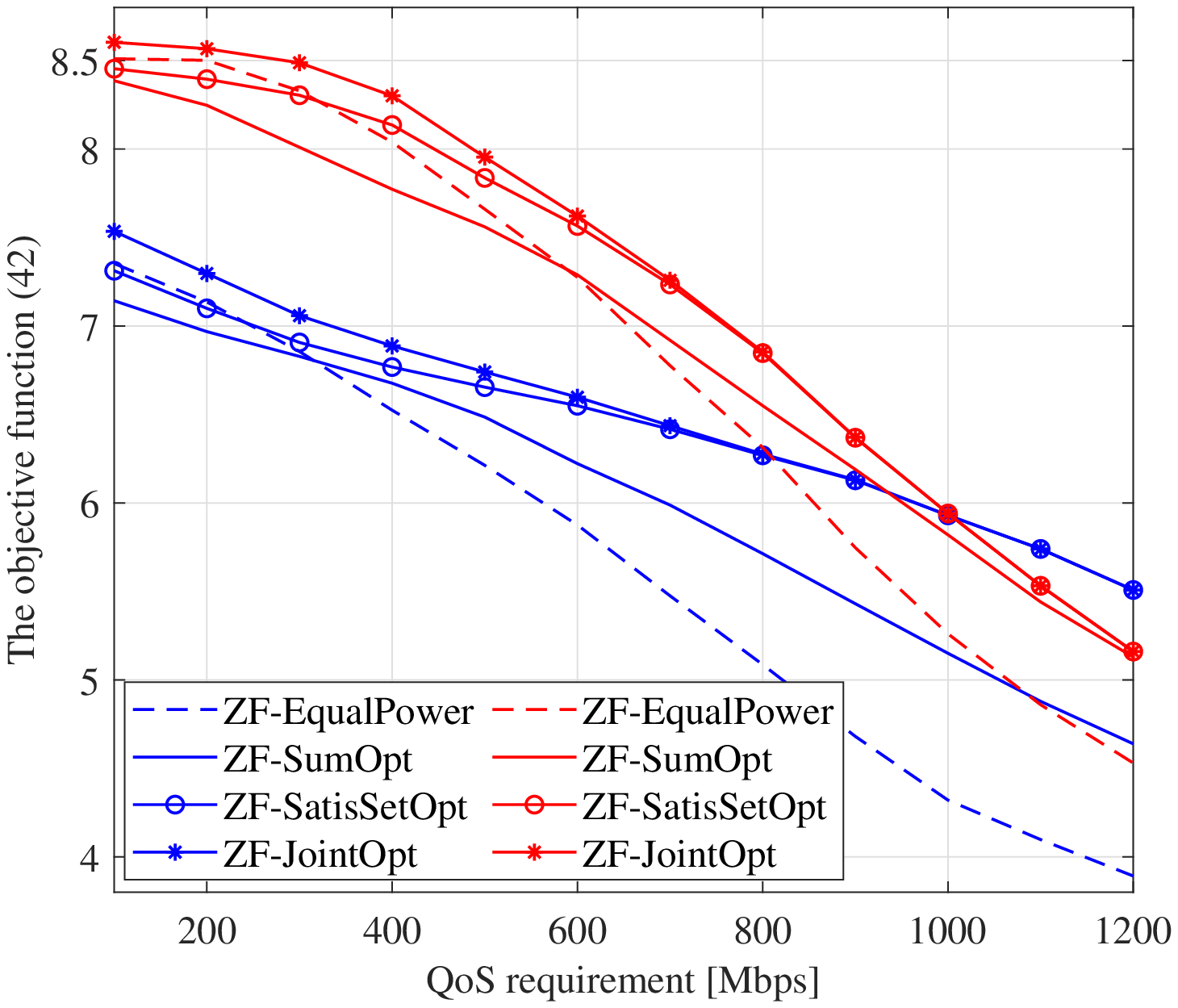}
			%\vspace{-1cm}
			\caption{{\hili The objective function defined in \eqref{specific_obj} versus the QoS requirement.}}
			\label{Fig:NormObj}
		\end{minipage}
		\hfill
		\begin{minipage}{0.325\textwidth}
			\centering
			\includegraphics[width= 1.1 \textwidth]{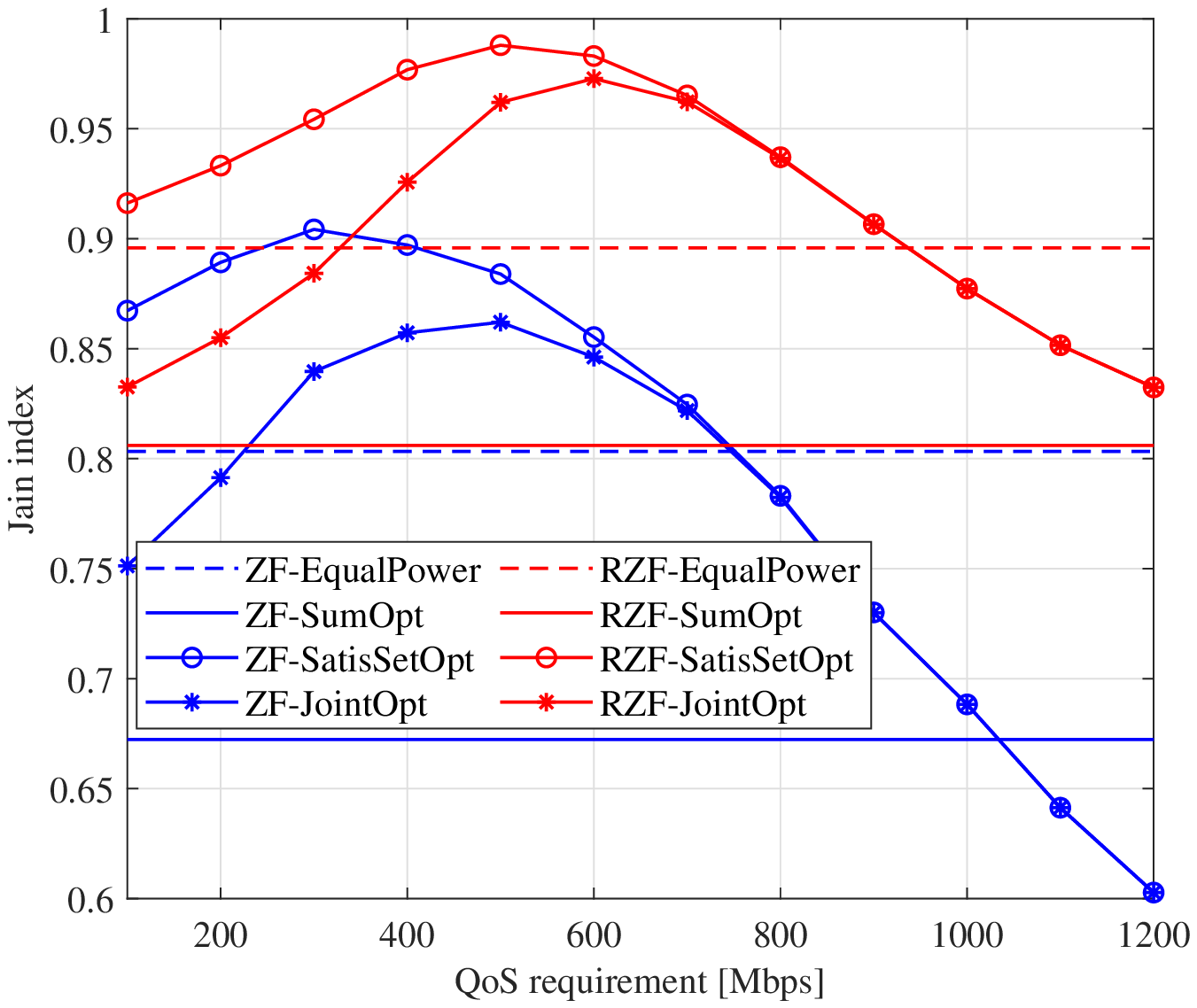}
			%\vspace{-1cm}
			\caption{The Jain's index versus the QoS requirement.}
			\label{Fig:JainIndex}
		\end{minipage}	
		%	\caption{The two example of the served rates based on RZF: $(a)$ all users satisfied their demands and $(b)$ some users cannot meet their requirements.}
		%	\label{fig:}
		%	%\vspace{-0.5cm}
		%\vspace{-25pt}
	\end{figure*}
	By exploiting the ZF and RZF precoding, the following benchmarks are involved for comparison:
	\begin{itemize}
		\item[$i)$] \textit{Joint sum rate and satisfied-user set maximization} is presented by Algorithm~\ref{alg:glob_alg} for a general framework, and by Algorithm~\ref{alg: prob_MMR_ZF} and \ref{alg: MMR_RZF_MRT} for the ZF and RZF precoding technique, respectively. This benchmark is denoted as ``JointOpt'' in the figures.
		\item[$ii)$] \textit{Satisfied-user set maximization} is a relaxation of JointOpt that only focuses on maintaining users' demand, especially users with bad channel conditions. If all the $K$ scheduled  users are served with their demands, the remaining power budget is equally assigned to every user. This benchmark is denoted as ``SatisSetOpt" in the figures.
		\item[$iii)$] \textit{Sum rate maximization} has been previously demonstrated in \cite{lu2019robust}, which only maximizes the total data throughput for which users with extreme channel conditions may be out of service to dedicate the power budget to other users. This benchmark is denoted as ``SumOpt" in the figures.
		\item[$iv)$] \textit{Equal power allocation} serves as a baseline to demonstrate the benefits of power allocation and satisfied-user set optimization \cite{trinh2021user,Krivochiza:21:access}. The transmit power level $14.92$~dB is assigned to each user without a guarantee on users' demand. This benchmark is denoted as ``EqualPower" in the figures.
	\end{itemize}
	
	In Fig.~\ref{Fig:Snapshot}, we plot the served rate [Mbps] for every user relying on \eqref{eq:ratezf} by a given realization of user locations (see Fig.~\ref{Fig:Snapshot}(a)). Fig~\ref{Fig:Snapshot}(b) shows the effective channel gains, with users~$2$ and $4$ as the worst who are located near the boundary of the overlapping beams. For a fixed power level, EqualPower cannot guarantee the QoS requirements and those users get lower data throughput than their requests, which is $500$~[Mbps]. If the system deploys the sum-rate optimization to maximize the total data throughput of the entire network, users~$2$ and $4$ even get $1.5\times$ to $2\times$ lower data throughput than that of EqualPower. Both SatisSetOpt and JointOpt offer satisfactory data throughput to all the users. Nonetheless, JointOpt gives  $200$ [Mbps] higher the sum rate than SatisSetOpt, corresponding to the $4.8\%$ improvement. \textcolor{black}{In the following, we report the average system performance over $200$ different realizations of users' locations.} % From Fig.~\ref{Fig:ProbInfea} to Fig.~\ref{Fig:JainIndex},
	
    In Fig.~\ref{Fig:ProbInfea}, we evaluate the probability of congestion appearance, which is defined for time instances when at least one scheduled user does not satisfy its QoS requirement. If the QoS requirement increases, our proposed algorithms provide the lowest probability of congestion appearance for both the ZF and RZF precoding techniques, especially at a low QoS regime. By maximizing the total system sum rate only, SumOpt always causes the highest congestion since scheduled users with lower channel gains are allocated less power since there is no QoS guarantee. %For instance, a GEO system utilizing the ZF precoding technique with $\xi_k = 400$ [Mbps], SumOpt results in the $86$\% of users' locations congested. Meanwhile, the result of JointOpt and SatisSetOpt is equal to each other and is $3.2\times$ lower than SumOpt. Concurrently, the congestion just appears about $4$\% when the satellite system exploiting the RZF precoding technique and our proposed algorithms, which is $79.8\times$ and $11.8\times$ better than SumOpt and EqualPower, respectively. 
    In Fig.~\ref{Fig:ProbSatis}, the probability of satisfying demand-based constraints is defined as a ratio between the number of satisfied users and the total users in the networks, i.e., $\mathbb{E}\{ |\Qcal| \}/K$.  When the QoS requirement per user increases, the satisfaction reduces since the network faces difficulties in maintaining the demands for many users with a limited power budget. If each user requires a QoS requirement level less than $400$ [Mpbs], SumOpt offers the lowest probability of satisfying  demand-based constraints. \textcolor{black}{After the effort to maximize the number of satisfied users, our proposed approaches allow some users served by a data throughput less than requested to maximize the sum rate. Another possible option is to schedule these unsatisfied users later in the next time slots. The joint congestion control and sum-rate maximization over multiple time slots are left for future work.}
     %For example, with $\xi_k = 300$ [Mbps], the satisfactory probability offered SumOpt is $99.9$\% and $90.9$\% with the ZF and RZF precoding techniques, respectively. Our proposed algorithms, i.e., JointOpt and SatisSetOpt, bring significant improvements of the satisfied users, that is $1.2\times$ and $1.3\times$  better than SumOpt regarding the ZF and RZF precoding technique.
     
    Besides, Fig.~\ref{Fig:Sum} demonstrates the scarification of the sum rate to improve the number of satisfied users. Both EqualPower and SumOpt allocate the transmit powers to the users without any guarantee of the individual QoSs, thus they should provide the constant sum rate of $5618$~[Mbps] and $5467$~[Mbps] on average by exploiting the RZF precoding technique. Meanwhile, the system with the ZF precoding technique is $4346$~[Mbps] and $4670$~[Mbps]. By using the RZF precoding technique, SatisSetOpt needs to lower the sum rate $420$~[Mbps] compared with SumOpt to enhance the QoS, while the reduction is only about $177$~[Mbps] if the network deploys JointOpt. 
    {\hili In Fig.~\ref{Fig:SumSeperateRate}, we explain the features of the sum rate [Mbps] when the demand-based constraints are involved by utilizing JointOpt with the different sets, including the set of all scheduled  users $\Kcal$, the satisfied-user set $\Qcal$, and the unsatisfied-user set $\Kcal \setminus \Qcal$. The sum rate of all the users is synthesized from the sum rate of satisfied- and unsatisfied-user sets as a consequence of problem~\eqref{probGlobal}. %For a system with the ZF precoding, the sum rate of satisfied users dramatically reduces from $3793.5$ [Mbps] to $3582.2$ [Mbps] as the  QoS requirement increases from $800$ [Mbps] to $1200$ [Mbps]. Meanwhile,  the sum rate of unsatisfied users rapidly increases from $449.6$ [Mbps] to $858.4$ [Mbps]. Similar trends are observed for the RZF precoding technique. 
    }
	\begin{figure*}[t]
		\begin{minipage}{0.325\textwidth}
			%		\centering
			\includegraphics[width=1.1\textwidth]{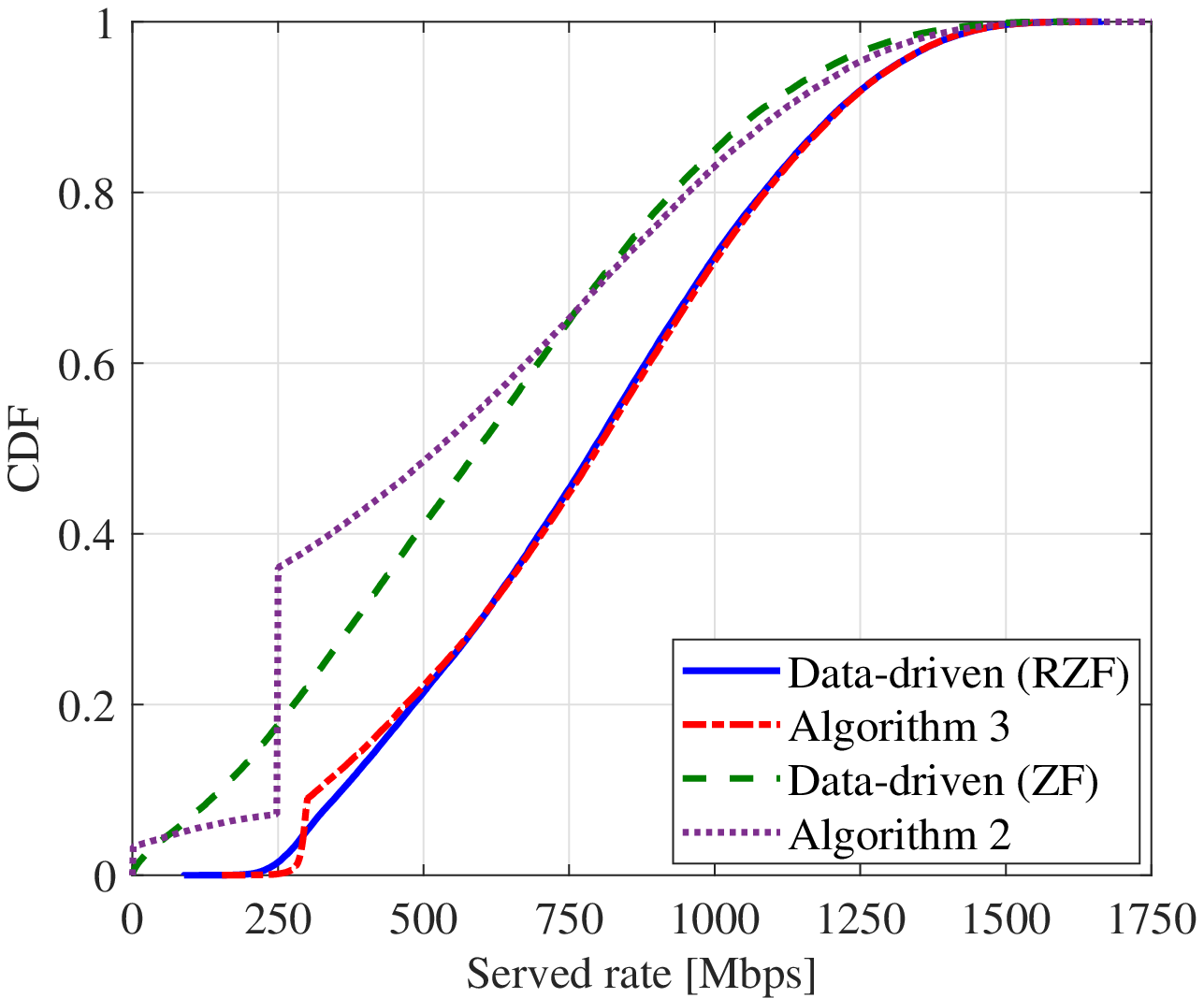} \\
			\centering $(a)$
		\end{minipage}
		\begin{minipage}{0.325\textwidth}
			%		\centering
			\includegraphics[width=1.1\textwidth]{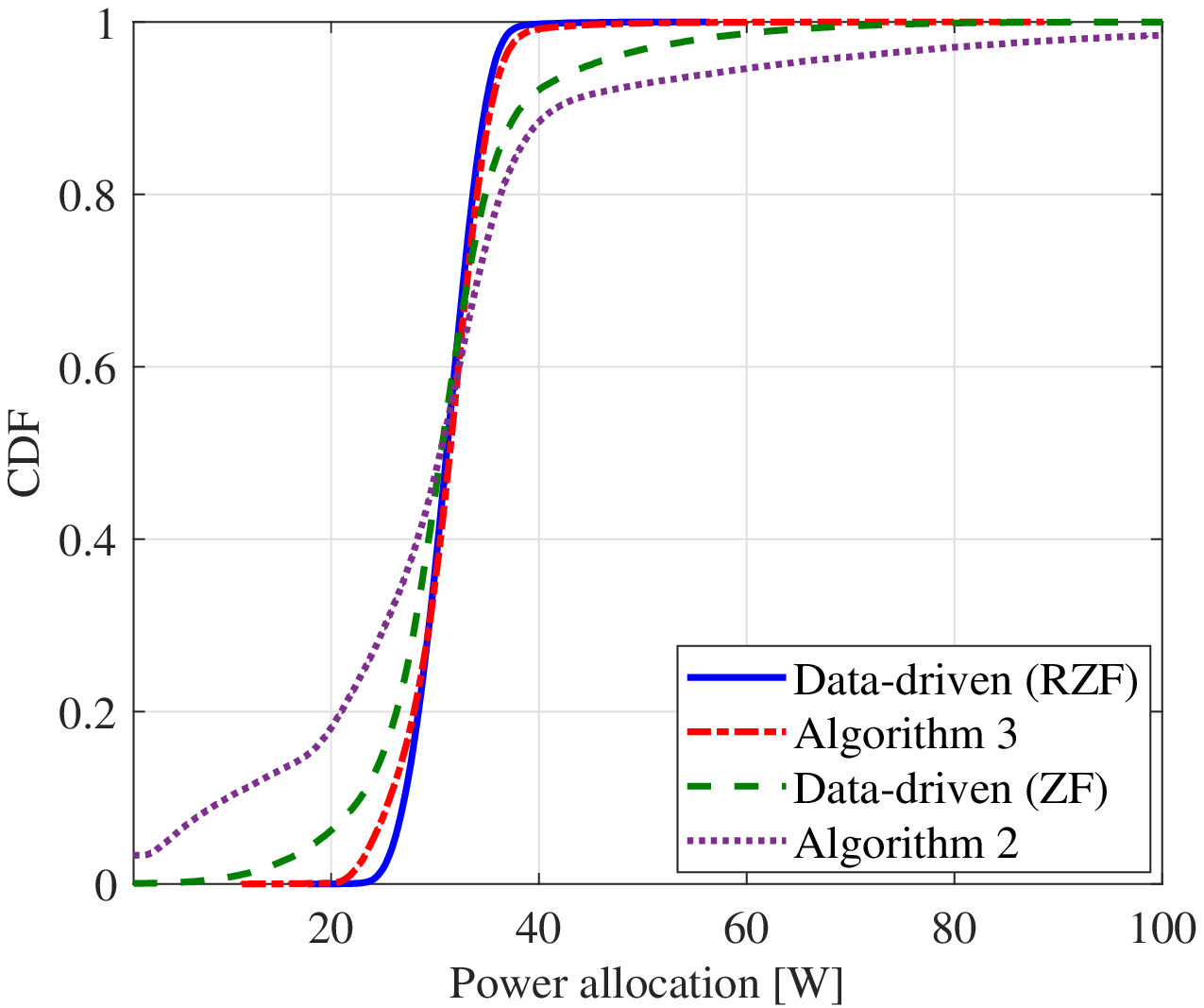}\\
			\centering $(b)$
		\end{minipage}
		\begin{minipage}{0.325\textwidth}
			%		\centering
			\includegraphics[width=1.1\textwidth]{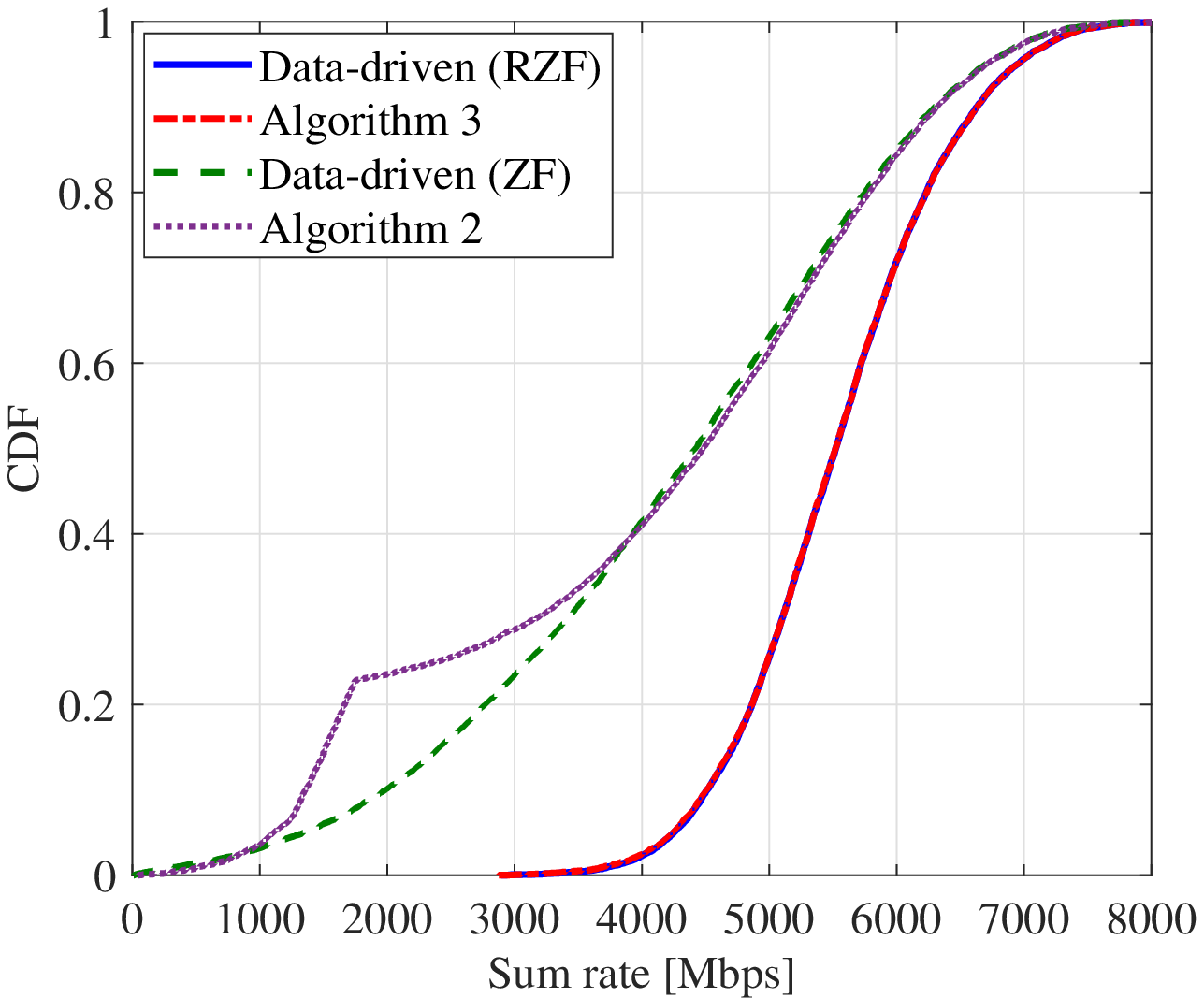} \\
			\centering $(c)$
		\end{minipage}
		%\vspace{-0.25cm}
		\caption{The cumulative distribution function (CDF) of the different metrics provided by the model-based and data-driven approaches with the individual QoS requirement $250$~[Mbps]: $(a)$ the served rate per user; $(b)$ the power allocation to each user; $(c)$ the sum rate. }
		\label{Fig:MLDR}
		%\vspace{-15pt}
	\end{figure*}
	
%	\begin{table*}[t]
%		\caption{The performance and run time (milliseconds) comparison of the model-based and data-driven approaches}
%		\label{runtime}
%		%\vspace{-15pt}
%		\centering
%		{\setlength{\tabcolsep}{0.14em}
%			%		\setlength{\extrarowheight}{0.1em}
%			\begin{tabular}{|c|c|c|c|c|c|c|c|c|}
%				\hline
%				& \multicolumn{2}{c|}{Algorithm~\ref{alg: prob_MMR_ZF}} &\multicolumn{2}{c|}{ Algorithm~\ref{alg: MMR_RZF_MRT}} &\multicolumn{2}{c|}{ Data-driven (ZF)} &\multicolumn{2}{c|}{ Data-driven (RZF)}\\
%				\hline
%				\makecell{ QoS requirement $\mbox{[Mbps]}$ }& $\xi_k = 250$ & $\xi_k = 300$ &$\xi_k = 250$ &$\xi_k = 300$&$\xi_k = 250$ &$\xi_k = 300$&$\xi_k = 250$ & $\xi_k = 300$  \\
%				\hline
%				Time [ms]  & $17.38$ & $19.97 $ & $19.26$ & $20.15 $  & $2.0$ & $2.2 $ & $1.3$ & $2.1 $ \\
%				\hline		
%				Sum rate [Mbps]    & $4054$ & $4048$ & $5542$ & $5483$ & $4260$ & $4248$ & $5545$ & $5491$ \\
%				\hline				
%				\makecell{Percentage of \\ satisfactions (\%)}& $92.14$ & $91.27$ & $99.86$ & $99.83$ & $82.38$ & $79.82$ & $98.47$ & $97.64$ \\
%				\hline 
%		\end{tabular}}
%		%\vspace{-25pt}
%	\end{table*}
	\begin{table}[t]
		\caption{The performance and run time (milliseconds) comparison of the model-based and data-driven approaches}
			\fontsize{8}{9.5}\selectfont
		\label{runtime}
		\centering
		\begin{tabular}{|l|c|c|c|c|}
			\hline
			& \begin{tabular}[c]{@{}l@{}}QoS \\ require.\\ $\mbox{[Mbps]}$\end{tabular} & \begin{tabular}[c]{@{}l@{}}Time \\ {[}ms{]}\end{tabular} & \begin{tabular}[c]{@{}l@{}}Sum \\ rate \\ {[}Mbps{]}\end{tabular} & \begin{tabular}[c]{@{}l@{}}Percentage of\\ satisfactions\\ {[}\%{]}\end{tabular} \\ 
			\hline
			\multirow{2}{*}{Model-based (ZF)}& $\xi_k = 250$ & $17.38$  & $4054$ & $92.14$ \\ 
			\cline{2-5} 
			& $\xi_k = 300$   & $19.97$ & $4048$ & $91.27$  \\ 
			\hline
			\multirow{2}{*}{Model-based (RZF)} & $\xi_k = 250$ & $19.26$& $5542$ & $99.86$ \\ 
			\cline{2-5} 
			& $\xi_k = 300$  & $20.15$ & $5483$ & $99.83$  \\ 
			\hline
			\multirow{2}{*}{Data-driven (ZF)}  & $\xi_k = 250$ & $2.0$ & $4260$& $82.38$ \\ \cline{2-5} 
			& $\xi_k = 300$ & $2.2$ & $4248$ & $79.82$ \\ 
			\hline
			\multirow{2}{*}{Data-driven (RZF)} & $\xi_k = 250$ & $1.3$ & $5545$ & $98.47$ \\ \cline{2-5} 
			& $\xi_k = 300$ & $2.1$ & $5491$& $97.64$ \\ 
			\hline
		\end{tabular}
		%\vspace*{-15pt}
	\end{table}
	{\hili For evaluating the balance between the sum rate and the satisfied-user set, we now define a specific case of objective function as
		%\vspace{-0.2cm}
		\begin{align}\label{specific_obj}
			\Lambda \triangleq \Omega \bigg(\frac{|\mathcal{Q}|}{K}  + \frac{\sum\nolimits_{k\in\mathcal{K}}R_{k} (\{p_{k'}^\ast\} )}{ \sum\nolimits_{k\in\mathcal{K}}R_{k}^\text{SumOpt} (\{p_{k'}^\ast \} )}\bigg),
			%\vspace{-0.2cm}
		\end{align}
%		\begin{equation}
%			\Lambda \triangleq \frac{1}{2}\left(\frac{|\Qcal|}{K} + \frac{\sum\nolimits_{k\in\Kcal}R_{k} (\{p_{k'}^\ast\} )}{ \sum\nolimits_{k\in\Kcal}R_{k}^\text{SumOpt} (\{p_{k'}^\ast \} )}\right),
%			%\vspace{-0.2cm}
%		\end{equation}
		where $\Omega= \frac{K\sum_{k\in\mathcal{K}}R^\text{SumOpt}_k(\{p^*_{k'}\})}{K+\sum_{k\in\mathcal{K}}R^\text{SumOpt}_k(\{p^*_{k'}\})}$ stands for the normalized factor and $R_{k}^\text{SumOpt}(\{p_{k'}^\ast\} )$ is  the channel capacity of $\UEk$ obtained by solving problem~\eqref{RelatedWorkSumRate}. We compare the performance of all the benchmarks versus the different QoS requirements as in Fig.~\ref{Fig:NormObj}. JointOpt gives the highest performance as the individual QoS requirement $\{\xi_k\}$ varies.  Algorithms~\ref{alg: prob_MMR_ZF} and \ref{alg: MMR_RZF_MRT} can handle the conflict utility met rices well. The other benchmarks, i.e., EqualPower and SumOpt, give significantly lower performance due to ignoring the users' demands. Moreover, the higher QoS requirements expand the gap between EqualPower and SumOpt and our proposed algorithms. If $\UEk$ requests $\xi_k = 200$~[Mbps] and the system uses the ZF precoding technique, JointOpt and EqualPower is almost overlapped. However, the gap expands  $10\times$ when the individual QoS requirement is $1200$~[Mbps]. 
		}
	
%For the network with a hybrid number of users where or not they can be either served by their QoS requirements, 
The Jain's fairness index is a good metric to measure how the offered data throughput matches the demands at the user levels \cite{jain1984quantitative}. By computing the satisfaction demand of each user, i.e., denoted by $o_{k}$ as a ratio between the offered data throughput and the QoS requirement of user~$k, \forall k$, then the Jain's index is $J = {(\sum_{k \in \Kcal} o_k)^2}/({K\sum_{k \in \Kcal} o_k^2})$, 
		% 		\begin{equation}
		% 			J = \frac{(\sum_{k \in \Kcal} o_k)^2}{K\sum_{k \in \Kcal} o_k^2},
		% 		\end{equation}
		which varies from $1/K$ to $1$. Fig.~\ref{Fig:JainIndex} plots the Jain's index for all the benchmarks as a function of the QoS requirement. To maximize the sum rate, the power should be dedicated to users with good channel conditions. This unfair policy leads SumOpt to a very low Jain's index. Consequently, an equal power allocation strategy offers a better fairness level with the Jain's index $1.2\times$ and $1.1\times$ higher than SumOpt by utilizing the ZF and RZF precoding techniques, respectively. The two conflict objective functions, i.e., the satisfied-user set and the sum rate, results in the second-best Jain's index with up to $1.23\times$  better than SumOpt with the RZF precoding. %The proposed framework, SatisSetOpt, in which the satisfied-user set is first maximized, followed by an equal power allocation, provides the best Jain's index among the benchmarks. 
	
	Figure~\ref{Fig:MLDR} shows the CDF of some metrics for both the model-based and data-driven approaches. The continuous mappings in \eqref{eq:well}--\eqref{eq:M3} may  not be isomorphisms since the codomains are non-smooth functions, especially for the achievable rates in \eqref{eq:M2}  (see Figs.~\ref{Fig:MLDR}(a) and (c)). The fact manifests difficulties in training and predicting the joint power allocation and satisfied-user set optimization. However, the neural network learns pretty well for some regimes with smooth CDFs. %Those results show that learning the features of a system with the ZF precoding technique, which exploits the orthogonality among the channels, is more challenging than those of the RZF precoding technique. 
	Fig~\ref{Fig:MLDR}(b) shows that the power allocation difference between the data-driven and model-based approaches are $30.16\%$ and $12.35\%$ on average with the ZF and RZF precoding techniques, respectively. Furthermore, we show in detail the performance and run time of those two approaches in Table~\ref{runtime}. Although there is a slightly increasing the run time when the individual QoS requirement increases, all the proposed approaches yield the results in milliseconds (ms). Specifically, %thanks to the benefits of machine learning and properly defined continuous mappings, 
	the data-driven approach reduces run times up to about { $14\times$} compared with the model-based approach. %Furthermore, the rate mismatch is up to about $5.08\%$ and $0.14\%$ between the two approaches by using the ZF and RZF precoding techniques, respectively.
	
		\begin{figure*}[t]
		\begin{minipage}{0.5\textwidth}
			%		\centering
			\includegraphics[width=0.8\textwidth]{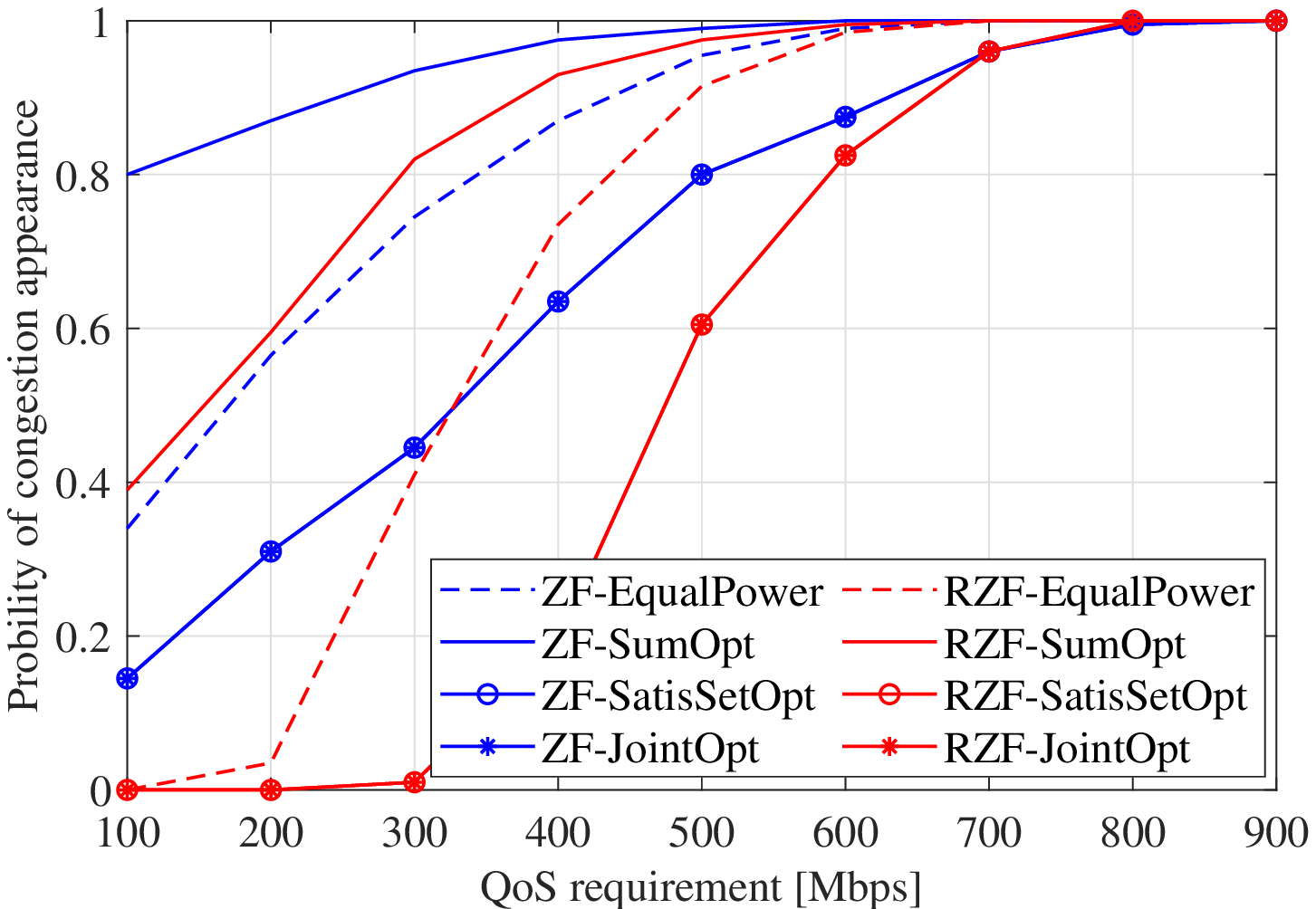}\\
			\centering $(a)$
			%\vspace{-0.2cm}
		\end{minipage}
		\begin{minipage}{0.5\textwidth}
			%		\centering
			\includegraphics[width=0.8\textwidth]{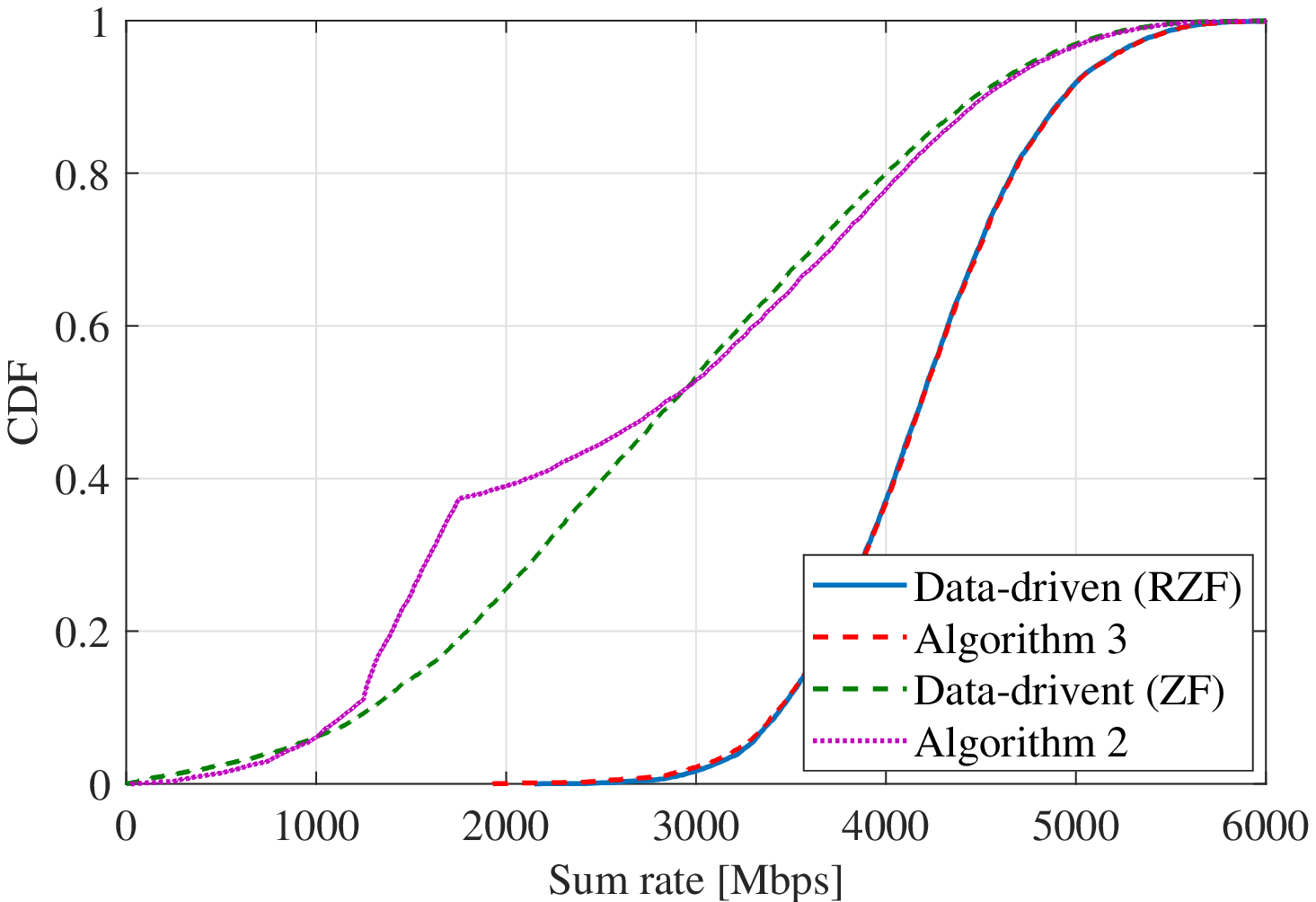} \\
			\centering $(b)$
			%\vspace{-0.2cm}
		\end{minipage}
		%\vspace{-0.25cm}
		\caption{The performance evaluation under the propagation environments including the rain and cloud attenuation: $(a)$ the probability of congestion appearance versus the QoS requirement.; $(b)$ the CDF of the  sum rate [Mbps] provided by the model-based and data-driven approaches with the QoS requirement per user 250 [Mbps]. }
		\label{Fig:rain_cloud}
		%\vspace{-28pt}
	\end{figure*}
	{\color{black} In Fig. \ref{Fig:rain_cloud}, we show the impact of the atmosphere loss including rain and cloud attenuation on the system performance. In particular, rain fading model can be modeled by a log-normal distribution, whose parameters such as mean and variance have been selected according to the European climate \cite{6184256}. Salonen-Uppala model \cite{8718480, el19910687} is used for modeling cloud attenuation, which depends on several features, i.e., the elevation angle toward the satellite, user's location and the carrier frequency.  Particularly, the channel model from the satellite to $\mathtt{UE}_k$ is formulated as
%			%\vspace{-0.2cm}
%			\begin{equation}
			$	\tilde{\bh}_k = \bh_k\sqrt{r_k}/\sqrt{c_k}, k\in\mathcal{K},$ 
%				%\vspace{-0.2cm}
%			\end{equation}		
			where $r_k$ and $c_k$ is the rain fading and cloud attenuation at $\mathtt{UE}_k$, respectively. Herein, $r_k$ is modeled as a lognormal random variable with mean $-2.6$~[dB] and variance $1.63$~[dB] \cite{6184256}. $c_k$ is mathematically defined as
			%\vspace{-0.2cm}
			\begin{equation}
				c_k = \frac{0.819fW_\text{red}}{\varepsilon''(1+\zeta^2)}\frac{1}{\sin(E_k)}, k\in\mathcal{K},
			%\vspace{-0.2cm}
			\end{equation}
			where $f$ [GHz] is the carrier frequency, $W_\text{red}=0.6$ [kg/m$^2$] is the statistics for the integrated reduced liquid water content, $E_k$ denotes the elevation angle between $\mathtt{UE}_K$ and the satellite. We define $\zeta = (2+\varepsilon')/\varepsilon''$ with $\varepsilon'$
			and $\varepsilon''$ present the real and imaginary parts of the permittivity of water, which is calculated as \cite{el19910687}
			%\vspace{-0.2cm}
			\begin{align}
					\varepsilon' &= \varepsilon_2 + \frac{\varepsilon_0 - \varepsilon_1}{1+(\frac{f}{f_p})^2} + \frac{\varepsilon_1-\varepsilon_2}{1 + (\frac{f}{f_s})^2},\\
					\varepsilon'' & = \frac{f(\varepsilon_0 - \varepsilon_1)}{f_p\big(1+(\frac{f}{f_p})^2\big)} + \frac{f(\varepsilon_1-\varepsilon_2)}{f_s\big(1 + (\frac{f}{f_s})^2\big)},
					%\vspace{-0.2cm}
			\end{align}
			with $\varepsilon_0=77.66 + 103.3(\vartheta - 1)$, $\varepsilon_1=5.48, \varepsilon_2 = 3.51$. $f_p = 20.09 - 142(\vartheta-1) + 294(\vartheta-1)^2$ [GHz] and $f_s = 590-1500(\vartheta-1)$ [GHz] are the principal and secondary relaxation frequencies, respectively. Finally, $\vartheta = 300/T$ with $T=273.15$ is temperature measured in Kevin. 
			In Fig.~\ref{Fig:rain_cloud}(a), it is numerically observed that even though the congestion probability increases in all algorithms because of the consideration of the atmosphere loss, our proposed algorithms still outperforms other benchmarks. {\color{black} Furthermore, Fig.~\ref{Fig:rain_cloud}(b) manifests that the data-driven approaches work well with the updated channel models.}}

	%\vspace{-0.55cm}	
	\section{Conclusions}\label{sec:conclusion}
	%\vspace{-0.2cm}
	%	%\vspace{-0.1cm}
	This paper has investigated the congestion issue in the demand-based optimization for multi-beam multi-user satellite communications. Two for one, under the methodology of multi-objective optimization, we jointly maximized the sum rate and satisfied-user set with all the channel conditions when many users share the same time and frequency resource plane. Conditioned on maintaining the QoS requirements as the priority, we have designed the heuristic algorithms that can effectively solve the optimization problem and operate in both feasible and infeasible domains under the limited power budget and the individual QoS requirements. By exploiting the water filling method and the linear precoding technique, numerical results confirmed that the number of satisfied users is significantly increased by utilizing our framework compared with the state-of-the-art benchmarks. Furthermore, the run time by deploying a neural network reduces to be far away to $10$~ms enabling real-time power allocation and satisfied-user control in satellite systems where the solution needs to be updated even at the millisecond time sale because of variety in the user scheduling decisions or individual user demands.
	%\vspace{-15pt}	
	\appendix
	%\vspace{-0.5cm}
	\subsection{Proof of Theorem~\ref{Cardinality}} \label{Appendix:Cardinality}
	%\vspace{-0.2cm}
	From Assumption~\ref{AssumpPrior}, the system first prioritizes on maximizing the number of satisfied users with the minimum transmit power consumption. This priority will lead to the maximum amount of the remaining power budget for the objective function $f_0(\{ p_{k'} \})$. 
	By assuming that the solution to power control is available and $\Qcal = \Kcal$, the total transmit power minimization problem is formulated as follows
	% 	\begin{equation} \label{probGlobalApp1}
	% 		\begin{aligned}
	% 			&\underset{ \{ p_{k'} \in \mathbb{R}_{+} \}}{\mathrm{minimize}}& & \sum\nolimits_{k \in \mathcal{K}} p_k \\
	% 			&\mbox{subject to} && R_k( \{ p_{k'} \}) \geq \xi_k,  \forall k\in\Kcal, \\
	% 			&&& \sum\nolimits_{k \in\Kcal} p_k \leq P_{\max}.
	% 		\end{aligned}
	% 	\end{equation}
		%\vspace{-0.2cm}
%	\begin{equation} %\label{probGlobalApp1}
%		% 		\begin{aligned}
%		\underset{ \{ p_{k'} \in \mathbb{R}_{+} \}}{\mathrm{min}} \sum\nolimits_{k \in \mathcal{K}} p_k , \quad
%		\mbox{s.t.} \   R_k( \{ p_{k'} \}) \geq \xi_k,  \forall k\in\Kcal, \And
%		\sum\nolimits_{k \in\Kcal} p_k \leq P_{\max}.
%		% 		\end{aligned}
%		%\vspace{-0.2cm}
%	\end{equation}
	\begin{subequations} \label{probGlobalApp1}
	\begin{alignat}{2}
		&	\underset{ \{ p_{k'} \in \mathbb{R}_{+} \}}{\mathrm{minimize} }&& \sum\nolimits_{k \in \mathcal{K}} p_k ,\\
		&\mbox{subject to}\ & & R_k( \{ p_{k'} \}) \geq \xi_k,  \forall k\in\Kcal,\\
		&&& \sum\nolimits_{k \in\Kcal} p_k \leq P_{\max}.
	\end{alignat}
	%\vspace{-0.2cm}
\end{subequations}
	We notice that problem~\eqref{probGlobalApp1} has a non-empty feasible set, and it is indeed a convex problem. By denoting $\alpha_k = 2^{\xi_k/B} -1, \forall k,$  \eqref{probGlobalApp1} is converted from the demand-based constraints to the SINR requirements as
	% 	\begin{equation} \label{probGlobalApp2}
	% 		\begin{aligned}
	% 			&\underset{ \{ p_{k'} \in \mathbb{R}_{+} \}}{\mathrm{minimize}}&\quad & \sum\nolimits_{k \in \mathcal{K}} p_k \\
	% 			&\mbox{subject to} &\quad& \gamma_k( \{ p_{k'} \}) = \alpha_k,  \forall k\in\Kcal, \\
	% 			&&& \sum\nolimits_{k \in\Kcal} p_k \leq P_{\max}.
	% 		\end{aligned}
	% 	\end{equation}
	%\vspace{-0.2cm}
%	\begin{equation} %\label{probGlobalApp2}
%		\underset{ \{ p_{k'} \in \mathbb{R}_{+} \}}{\mathrm{min}}\quad  \sum\nolimits_{k \in \mathcal{K}} p_k, \quad
%		\mbox{s.t.} \  \gamma_k( \{ p_{k'} \}) = \alpha_k,  \forall k\in\Kcal, \And  
%		\sum\nolimits_{k \in\Kcal} p_k \leq P_{\max}.
%		%\vspace{-0.2cm}
%	\end{equation}
	\begin{subequations} \label{probGlobalApp2}
	\begin{alignat}{2}
		&\underset{ \{ p_{k'} \in \mathbb{R}_{+} \}}{\mathrm{minimize} }&& \quad  \sum\nolimits_{k \in \mathcal{K}} p_k,\\
		&\mbox{subject to}\ & & \gamma_k( \{ p_{k'} \}) = \alpha_k,  \forall k\in\Kcal, \\
		&&& \sum\nolimits_{k \in\Kcal} p_k \leq P_{\max}.
	\end{alignat}
	%\vspace{-0.2cm}
\end{subequations}
	The equality constraints in \eqref{probGlobalApp2} is obtained by the fact that problems \eqref{probGlobalApp1} and \eqref{probGlobalApp2} share the same global optimum. By exploiting the SINR expression in \eqref{eq:SINR} for $\UEk$ into the corresponding SINR constraint in \eqref{probGlobalApp2}, we now recast this SINR constraint into an equivalent form as
	%\vspace{-0.2cm}
	\begin{equation} \label{eq:pk}
		\begin{split}
			p_k |\mathbf{h}_K^H \mathbf{w}_k|^2 &= \alpha_k \sigma^2  + \alpha_k	\sum\nolimits_{\ell \in \mathcal{K} \setminus \{k\} } p_{\ell } | \mathbf{h}_k^H \mathbf{w}_{\ell} |^2 \\
			\stackrel{(a)}{\Leftrightarrow} p_k &= \frac{\alpha_k \sigma^2}{(\alpha_k +1)|\mathbf{h}_K^H \mathbf{w}_k|^2}\\
			&\quad + \frac{\alpha_k}{(\alpha_k +1)|\mathbf{h}_k^H \mathbf{w}_k|^2} \sum\nolimits_{\ell \in \mathcal{K}} p_{\ell } | \mathbf{h}_k^H \mathbf{w}_{\ell} |^2,
		\end{split}
		%\vspace{-0.2cm}
	\end{equation}
	where $(a)$ is obtained by adding the extra term $\alpha_k p_{k } | \mathbf{h}_k^H \mathbf{w}_{k} |^2$ into both sides of the first equality in \eqref{eq:pk}, then doing some algebraic manipulation. Repeating the same steps for the SINR constraints of all the $K-1$ scheduled  users and then stacking them in the matrix form, we obtain the linear equation as
	%\vspace{-0.2cm}
	\begin{equation} \label{eq:LinearApp}
		(\mathbf{I}_K - \mathbf{R} \mathbf{Q}) \mathbf{p} = \pmb{\nu},
		%\vspace{-0.2cm}
	\end{equation}
	where $\mathbf{R}, \mathbf{Q},$ and $\pmb{\nu}$ are  given in the theorem. In \eqref{eq:LinearApp}, $\mathbf{p} = [p_1, \ldots, p_K] \in \mathbf{R}_{+}^K$.  We observe that $\mathbf{R} \mathbf{Q}$ has nonnegative elements. By applying the Perron-Frobenius theorem \cite{pillai2005perron}, the spectral radius of matrix $\mathbf{R} \mathbf{Q}$ should satisfy $\rho(\mathbf{R} \mathbf{Q}) = \max\{ |\lambda_1|, \ldots, |\lambda_K|\} < 1$. 
	% 	\begin{equation}
	% 		\rho(\mathbf{R} \mathbf{Q}) = \max\{ |\lambda_1|, \ldots, |\lambda_K|\} < 1. 
	% 	\end{equation}
	After that, the unique solution to \eqref{eq:LinearApp} exists since $(\mathbf{R} \mathbf{Q})^m \rightarrow \mathbf{0}$ as $m \rightarrow \infty$, which implies that $(\mathbf{I}_K - \mathbf{R}\mathbf{Q})^{-1} = \sum_{m=0}^{\infty} (\mathbf{R} \mathbf{Q})^m$ converges, and each element is nonnegative. Consequently, the first condition as shown in the theorem.  The minimum power solution that the satellite spends on serving all the $K$ scheduled  users with the QoS requirements is % $\mathbf{p}^{\ast} = (\mathbf{I}_K - \mathbf{R}\mathbf{Q})^{-1} \pmb{\nu}$. 
		%\vspace{-0.2cm}
		\begin{equation} \label{eq:PowerSol}
			\mathbf{p}^{\ast} = (\mathbf{I}_K - \mathbf{R}\mathbf{Q})^{-1} \pmb{\nu}.
		%\vspace{-0.2cm}
		\end{equation}
	Combining the power solution and the limited power budget constraint in \eqref{probGlobalApp1}, we obtain the second condition as shown in the theorem.
	%\vspace{-0.55cm}
	\subsection{Proof of Theorem~\ref{theorem:Coverge}} \label{appendix:Coverge}
	%\vspace{-0.2cm}
	Let us define $\Qcal^{(n)}$ a feasible satisfied-user set to problem~\eqref{alg_SRM:phase23} that contains all the scheduled  users with at least their QoS requirements at iteration~$n$, which is defined as follows
	%\vspace{-0.2cm}
	\begin{equation}
		\begin{split}
			\Qcal^{(n)} = \Big\{ k | R_k ( \{ p_{k'}^{(n)} \}) &= \xi_k, \forall k \in \Qcal^{\ast,(n-1)}, \\
			R_k ( \{ p_{k'}^{(n)} \}) &\geq \xi_k, k \in \Kcal \setminus \Qcal^{\ast,(n-1)} \Big\}.
		\end{split}
		%\vspace{-0.25cm}
	\end{equation}
	We further introduce $\widetilde{\Qcal}^{(n)}$ being the feasible region that contains all the possibilities $\Qcal^{(n)}$, then we obtain the following properties
	% 	\begin{equation}
	$\Qcal^{\ast, (n-1)} \in \widetilde{\Qcal}^{(n)} \mbox{ and }  |\Qcal^{\ast, (n-1)}| \leq  |\Qcal^{\ast, (n)}|$,
	% 	\end{equation}
	where the first property is attained by the fact that $\Qcal^{\ast, (n-1)}$ is involved in the demand-based constraint at iteration~$n$. The second property is because problem~\eqref{alg_SRM:phase23} should give a solution to the satisfied-user set not worse than the previous one. This establishes the monotonically increasing property in \eqref{eq:QSeries}. We only consider a finite set of scheduled  users, i.e., $|\Qcal^{(n)}| < K, \forall n$, thus \eqref{eq:QSeries} should be bounded from above. If the convergence holds at iteration~$n$, then the optimal satisfied-user set must be also a solution to iteration~$n+1$. Otherwise, it results in $|\Qcal^{\ast,(n+1)}| \geq |\Qcal^{\ast,(n)}|$. Algorithm~\ref{alg:glob_alg} ensures the cardinality of the satisfied-user set $\Qcal^{\ast}$ non-decreasing along with iterations and converges to a fixed point. 
	
	We prove the monotonic decreasing function of the sum rate in \eqref{eq:RSeries} by induction. Indeed, the first inequality holds, i.e., $\sum_{k \in \Kcal} R_k ( \{ p_{k'}^{\ast,(0)} \} ) \geq \sum_{k \in \Kcal} R_k ( \{ p_{k'}^{\ast,(1)} \} )$ since the feasible domain of problem~\eqref{RelatedWorkSumRate} corresponding to the weight values $\mu_1 = 1$ and $\mu_2 =0$ that provides a better sum rate solution than that of problem~\eqref{alg_SRM:phase23}. Assume that the inequality holds up to iteration~$n$, i.e., $\sum_{k \in \Kcal} R_k ( \{ p_{k'}^{\ast,(n-1)} \} ) \geq \sum_{k \in \Kcal} R_k ( \{ p_{k'}^{\ast,(n)} \} )$, and the proof should confirm that it also holds at iteration~$n+1$:
	%\vspace{-0.2cm}
	\begin{equation} \label{eq:Rn1ineq}
		\sum\nolimits_{k \in \Kcal} R_k \big( \{ p_{k'}^{\ast,(n)} \} \big) \geq \sum\nolimits_{k \in \Kcal} R_k \big( \{ p_{k'}^{\ast,(n+1)} \} \big).
		%\vspace{-0.2cm}
	\end{equation}
	We reformulate the left-hand side of \eqref{eq:Rn1ineq} by decomposing the scheduled -user set $\Kcal$ into the satisfied-user set and the unsatisfied-user set as follows 
	%\vspace{-0.2cm}
	\begin{equation}  \label{eq:LHSn}
		\begin{aligned}
			&\sum\nolimits_{k \in \Kcal} R_k \big( \{ p_{k'}^{\ast,(n)} \} \big) \\
			&= \sum\nolimits_{k \in  \Qcal^{\ast,(n)}} R_k \big( \{ p_{k'}^{\ast,(n)} \} \big) + \sum\nolimits_{k \in \Kcal\setminus \Qcal^{\ast,(n)}} R_k \big( \{ p_{k'}^{\ast,(n)} \} \big)   \\
			&=\sum\nolimits_{k \in \Qcal^{\ast,(n-1)}} \xi_k +	\sum\nolimits_{k \in  \bar{\Qcal}^{\ast,(n-1)}} R_k \big( \{ p_{k'}^{\ast,(n)} \} \big) \\
			&\quad + \sum\nolimits_{k \in \Kcal\setminus \Qcal^{\ast,(n)}} R_k \big( \{ p_{k'}^{\ast,(n)} \} \big),
		\end{aligned}
		%\vspace{-0.2cm}
	\end{equation}	
	%
	%	\begin{align} \label{eq:LHSn}
	%		&\sum\nolimits_{k \in \Kcal} R_k \big( \{ p_{k'}^{\ast,(n)} \} \big) = \sum\nolimits_{k \in  \Qcal^{\ast,(n)}} R_k \big( \{ p_{k'}^{\ast,(n)} \} \big) + \sum\nolimits_{k \in \Kcal\setminus \Qcal^{\ast,(n)}} R_k \big( \{ p_{k'}^{\ast,(n)} \} \big) \notag  \\
	%		&=\sum\nolimits_{k \in \Qcal^{\ast,(n-1)}} \xi_k +	\sum\nolimits_{k \in  \bar{\Qcal}^{\ast,(n-1)}} R_k \big( \{ p_{k'}^{\ast,(n)} \} \big) + \sum\nolimits_{k \in \Kcal\setminus \Qcal^{\ast,(n)}} R_k \big( \{ p_{k'}^{\ast,(n)} \} \big),
	%	\end{align}
	with noting that $\Qcal^{\ast,(n)} = \Qcal^{\ast,(n-1)} \cup \bar{\Qcal}^{\ast,(n-1)}$, where $\bar{\Qcal}^{\ast,(n-1)}$ is the satisfied-user set at iteration $n-1$ consisting of users with better data throughput than requested. Since the first part of \eqref{eq:LHSn} provides users with rates equal to their demands, we can formulate an optimization problem to maximize the left-hand side of \eqref{eq:Rn1ineq} with  the feasible domain $\mathcal{D}^{(n)}$ defined as follows
	%\vspace{-0.2cm}
%	\begin{equation} \label{eq:Dn}
		\begin{align} \label{eq:Dn}
			\mathcal{D}^{(n)} = \Big\{ p_{k}^{(n)},\forall k \in \Kcal \big| R_{k} &(\{ p_{k'}^{(n)}\}) = \xi_k, \forall k \in \Qcal^{\ast,(n-1)}, \non\\
			&\sum\nolimits_{k\in \Kcal} p_{k}^{(n)} \leq P_{\max} \Big\}.
		\end{align}
		%\vspace{-0.2cm}
%	\end{equation}
	Next, we recast the right-hand side of \eqref{eq:Rn1ineq} to an equivalent form as
	%\vspace{-0.2cm}
%	\begin{equation} 
%		\begin{split}
%			& \sum\nolimits_{k \in \Kcal} R_k \big( \{ p_{k'}^{\ast,(n+1)} \} \big) =  \sum\nolimits_{k \in  \Qcal^{\ast,(n+1)}} R_k \big( \{ p_{k'}^{\ast,(n+1)} \} \big) + \sum\nolimits_{k \in \Kcal\setminus \Qcal^{\ast,(n+1)}} R_k \big( \{ p_{k'}^{\ast,(n+1)} \} \big) 
%			%\vspace{-0.2cm}
%	\end{split}
%	\end{equation}
%		\begin{equation*}  \label{eq:CostRHSn1}
%		\begin{split}
%			& = \sum\nolimits_{k \in \Qcal^{\ast,(n)}} \xi_k +  \sum\nolimits_{k \in \bar{\Qcal}^{\ast,(n)}} R_k \big( \{ p_{k'}^{\ast,(n+1)} \} \big)   +
%			\sum\nolimits_{k \in \Kcal\setminus \Qcal^{\ast,(n+1)}} R_k \big( \{ p_{k'}^{\ast,(n+1)} \} \big)\\
%			& = \sum\nolimits_{k \in \Qcal^{\ast,(n-1)}} \xi_k + \sum\nolimits_{k \in \bar{\Qcal}^{\ast,(n-1)}} \xi_k  +   \sum\nolimits_{k \in \bar{\Qcal}^{\ast,(n)}} R_k \big( \{ p_{k'}^{\ast,(n+1)} \} \big)  +  
%			\sum\nolimits_{k \in \Kcal\setminus \Qcal^{\ast,(n+1)}} R_k \big( \{ p_{k'}^{\ast,(n+1)} \} \big). 
%		\end{split}
%		%\vspace{-0.2cm}
%	\end{equation*}
	\begin{equation} \label{eq:CostRHSn1}
	\begin{split}
		 \sum\limits_{k \in \Kcal} &R_k \big( \{ p_{k'}^{\ast,(n+1)} \} \big) \\
		&=  \sum\nolimits_{k \in  \Qcal^{\ast,(n+1)}} R_k \big( \{ p_{k'}^{\ast,(n+1)} \} \big) \\
		&\quad + \sum\nolimits_{k \in \Kcal\setminus \Qcal^{\ast,(n+1)}} R_k \big( \{ p_{k'}^{\ast,(n+1)} \} \big) \\
		& = \sum\nolimits_{k \in \Qcal^{\ast,(n)}} \xi_k +  \sum\nolimits_{k \in \bar{\Qcal}^{\ast,(n)}} R_k \big( \{ p_{k'}^{\ast,(n+1)} \} \big)   \\
		&\quad +
		\sum\nolimits_{k \in \Kcal\setminus \Qcal^{\ast,(n+1)}} R_k \big( \{ p_{k'}^{\ast,(n+1)} \} \big)\\
		& = \sum\nolimits_{k \in \Qcal^{\ast,(n-1)}} \xi_k + \sum\nolimits_{k \in \bar{\Qcal}^{\ast,(n-1)}} \xi_k  \\
		& \quad +   \sum\nolimits_{k \in \bar{\Qcal}^{\ast,(n)}} R_k \big( \{ p_{k'}^{\ast,(n+1)} \} \big)  \\
		&\quad +  
		\sum\nolimits_{k \in \Kcal\setminus \Qcal^{\ast,(n+1)}} R_k \big( \{ p_{k'}^{\ast,(n+1)} \} \big). 
	\end{split}
	%\vspace{-0.2cm}
\end{equation}
	%	\begin{align}\label{eq:CostRHSn1}
	%		& \sum\nolimits_{k \in \Kcal} R_k \big( \{ p_{k'}^{\ast,(n+1)} \} \big) \\
	%		&=  \sum\nolimits_{k \in  \Qcal^{\ast,(n+1)}} R_k \big( \{ p_{k'}^{\ast,(n+1)} \} \big) + \sum\nolimits_{k \in \Kcal\setminus \Qcal^{\ast,(n+1)}} R_k \big( \{ p_{k'}^{\ast,(n+1)} \} \big) \\
	%		& = \sum\nolimits_{k \in \Qcal^{\ast,(n)}} \xi_k +  \sum\nolimits_{k \in \bar{\Qcal}^{\ast,(n)}} R_k \big( \{ p_{k'}^{\ast,(n+1)} \} \big)  +  
	%		\sum\nolimits_{k \in \Kcal\setminus \Qcal^{\ast,(n+1)}} R_k \big( \{ p_{k'}^{\ast,(n+1)} \} \big) \\
	%		& = \sum\nolimits_{k \in \Qcal^{\ast,(n-1)}} \xi_k + \sum\nolimits_{k \in \bar{\Qcal}^{\ast,(n-1)}} \xi_k +   \sum\nolimits_{k \in \bar{\Qcal}^{\ast,(n)}} R_k \big( \{ p_{k'}^{\ast,(n+1)} \} \big)  +  
	%		\sum\nolimits_{k \in \Kcal\setminus \Qcal^{\ast,(n+1)}} R_k \big( \{ p_{k'}^{\ast,(n+1)} \} \big). 
	%	\end{align}
	where $\Qcal^{\ast,(n+1)} = \Qcal^{\ast,(n)} \cup \bar{\Qcal}^{\ast,(n)}$, and \eqref{eq:CostRHSn1} is obtained since \eqref{eq:RSeries} holds until iteration~$n$ by induction. By observing \eqref{eq:CostRHSn1},  an optimization problem is formulated to maximize the right-hand side of \eqref{eq:Rn1ineq} with  the feasible domain $\mathcal{D}^{(n+1)}$ defined as follows
	%\vspace{-0.2cm}
%	\begin{equation} \label{eq:Dn1}
		\begin{align}
			\mathcal{D}^{(n+1)} &= \Big\{ p_{k}^{(n+1)},\forall k \in \Kcal \big| R_{k} (\{ p_{k'}^{(n+1)}\}) = \xi_k, \label{eq:Dn1}\\
			& \forall k \in \Qcal^{\ast,(n-1)} \cup \bar{\Qcal}^{\ast,(n-1)} , \sum\nolimits_{k\in \Kcal} p_{k}^{(n+1)} \leq P_{\max} \Big\}.\non
		\end{align}
		%\vspace{-0.2cm}
%	\end{equation}
	Combining \eqref{eq:Dn} and \eqref{eq:Dn1}, it holds that $\mathcal{D}^{(n+1)} \subseteq \mathcal{D}^{(n)}$ since $\varnothing$ is a subset of $\bar{\Qcal}^{\ast,(n-1)}$. Hence, \eqref{eq:Rn1ineq} holds and we conclude the proof.
	%\vspace{-0.55cm}
	\subsection{Proof of Lemma~\ref{lemmaExistNeural}} \label{appendixExistNeural}
	%\vspace{-0.2cm}
	From the given optimized power coefficients $\{ p_{k}^\ast \}$ to the $K$ scheduled  users, the satisfied-user set $\Qcal^{\ast}$ is defined as
	% 	\begin{equation} \label{eq:Qset}
	$\Qcal^{\ast} = \big\{ k \big|  R_k ( \{ p_{k'}^{\ast}\}) \geq \xi_k, k \in \Kcal \big\}$,
	% 	\end{equation}
	where $R_k ( \{ p_{k'}^{\ast}\})$ is given as in \eqref{eq:rate} with $p_{k'} = p_{k'}^\ast, \forall k \in \Kcal$. The result indicates that the discrete set $\Qcal$ is explicitly characterized by  the propagation channels  and the power coefficients, which are continuous variables. This result is obtained by noting that the precoding vectors are defined by the instantaneous channel state information. Let us define $ \tau_k = \| \mathbf{h}_k \|$ and the law of conservation of energy points out that $0 \leq \tau_k \leq \sqrt{N}$, which is bounded from above. We observe that
	% 	\begin{equation} \label{eq:BoundNorm}
	$0 \leq | \mathbf{h}_k^H \mathbf{w}_{\ell} |^2 \stackrel{(a)}{\leq} \| \mathbf{h}_k^H  \|^2  \| \mathbf{w}_{\ell} \|^2 \stackrel{(b)}{=} \tau_k^2$,
	% 	\end{equation}
	where $(a)$ is obtained by the Cauchy-Schwarz inequality and $(b)$ is due to each precoding vector having the unit norm. %From \eqref{eq:BoundNorm}, 
	From this, the channel capacity of $\UEk$ is a continuous function and its feasible set is compact, which fulfill all the conditions of the universal approximation theorem \cite{goodfellow2016deep, hornik1989multilayer}. Consequently, we can construct a neural network with a finite number of neurons to learn the sum-rate optimization problem respect to both the power coefficients and satisfied user set. 
	%\subsection{Proof of Lemma~\ref{lemma:WF}} \label{appendix:WF}

	%\vspace{-0.3cm}
	%\setstretch{0.89}
	\bibliographystyle{IEEEtran}
	%	\balance
	\bibliography{Journal1}
	\begin{IEEEbiographynophoto} 
		%[{\includegraphics[width=1.0in,height=1.25in,clip,keepaspectratio]{VanPhucBui1.jpg}}]
		{Van-Phuc Bui} received the B.Sc. degree (Hons.) in Electronics and Telecommunications from Ho Chi Minh city University of Technology (HCMUT), Vietnam, in 2018, and the M.Sc. degree in Electronics and Telecommunications from Soongsil University, Seoul, South Korea, in 2020. He was  a research assistant at University of Luxembourg. He is currently pursuing the Ph.D. degree with the Aalborg University, Aalborg, Denmark. His research interests are in convex optimization techniques and machine learning applications for wireless communications, and satellite communications.
	\end{IEEEbiographynophoto}

		\begin{IEEEbiographynophoto} 
		%[{\includegraphics[width=1.0in,height=1.25in,clip,keepaspectratio]{TrinhVanChien_Photo.jpg}}]
		{Trinh Van Chien} (S'16-M'20) received the B.S. degree in Electronics and Telecommunications from Hanoi University of Science and Technology (HUST), Vietnam, in 2012. He then received the M.S. degree in Electrical and Computer Enginneering from Sungkyunkwan University (SKKU), Korea, in 2014 and the Ph.D. degree in Communication Systems from Link\"oping University (LiU), Sweden, in 2020. He was  a research associate at University of Luxembourg. He is now with the School of Information and Communication Technology (SoICT), Hanoi University of Science and Technology (HUST), Vietnam. His interest lies in convex optimization problems and machine learning applications for wireless communications and image \& video processing. He was an IEEE wireless communications letters exemplary reviewer for 2016, 2017, and 2021. He also received the award of scientific excellence in the first year of the 5Gwireless project funded by European Union Horizon's 2020.
	\end{IEEEbiographynophoto}

	\begin{IEEEbiographynophoto} 
	%[{\includegraphics[width=1.0in,height=1.25in,clip,keepaspectratio]{photo_eva.JPG}}]
	{Eva Lagunas} received the M.Sc. and Ph.D. degrees in telecommunications engineering from the Polytechnic University of Catalonia (UPC), Barcelona, Spain, in 2010 and 2014, respectively. She was Research Assistant within the Department of Signal Theory and Communications, UPC, from 2009 to 2013. During the summer of 2009 she was a guest research assistant within the Department of Information Engineering, Pisa, Italy. From November 2011 to May 2012 she held a visiting research appointment at the Center for Advanced Communications (CAC), Villanova University, PA, USA. In 2014, she joined the Interdisciplinary Centre for Security, Reliability and Trust (SnT), University of Luxembourg, where she currently holds a Research Scientist position. Her research interests include terrestrial and satellite system optimization, spectrum sharing, resource management and machine learning.
\end{IEEEbiographynophoto}

	\begin{IEEEbiographynophoto} 
	%[{\includegraphics[width=1.0in,height=1.25in,clip,keepaspectratio]{photo_joel.jpg}}]
	{Jo\"el Grotz} (Senior Member IEEE) graduated in Electrical Engineering from the University of Karlsruhe and the Grenoble Institute of Technology in 1999 and completed his Ph.D. in Telecommunications jointly at the University of Luxembourg and KTH in Stockholm in 2008. He has worked for SES in Betzdorf, Luxembourg on the development of satellite broadband communication system design for GEO and MEO high throughput satellite systems on different ground segment and space segment topics and system optimization aspects. He has previously worked in the Technical Labs at ST Engineering iDirect (former Newtec Cy at Sint-Niklaas in Belgium) on topics of system design and signal processing in satellite modems. He is currently working as Senior Manager at SES on the development of a dynamic resource management system for novel flexible satellite systems, including SES-17 and O3b mPOWER as well as future satellite systems under planning. 
\end{IEEEbiographynophoto}

\begin{IEEEbiographynophoto}
	%[{\includegraphics[width=1in,height=1.25in,clip,keepaspectratio]{SymeonChatzinotas_photo.jpg}}]
	{Symeon Chatzinotas} (S'06, M'09, SM'13) is Full Professor and Head of the SIGCOM Research Group at SnT, University of Luxembourg. He is coordinating the research activities on communications and networking across a group of 70 researchers, acting as a PI for more than 40 projects and main representative for 3GPP, ETSI, DVB. He is currently serving in the editorial board of the IEEE Transactions on Communications, IEEE Open Journal of Vehicular Technology and the International Journal of Satellite Communications and Networking.
	
	In the past, he has been a Visiting Professor at the University of Parma, Italy and was involved in numerous R\&D projects for NCSR Demokritos, CERTH Hellas and CCSR, University of Surrey.
	
	He was the co-recipient of the 2014 IEEE Distinguished Contributions to Satellite Communications Award and Best Paper Awards at WCNC, 5GWF, EURASIP JWCN, CROWNCOM, ICSSC. He has (co-)authored more than 600 technical papers in refereed international journals, conferences and scientific books. 
\end{IEEEbiographynophoto}

\begin{IEEEbiographynophoto}
	%[{\includegraphics[width=1in,height=1.25in,clip,keepaspectratio]{photo_Bjorn.jpg}}]
	{Bj\"orn Ottersten} (S'87–M'89–SM'99–F'04) received the M.S. degree in electrical engineering and applied
	physics from Linköping University, Linköping, Sweden, in 1986, and the Ph.D. degree in electrical
	engineering from Stanford University, Stanford, CA, USA, in 1990. He has held research positions with
	the Department of Electrical Engineering, Linköping University, the Information Systems Laboratory,
	Stanford University, the Katholieke Universiteit Leuven, Leuven, Belgium, and the University of
	Luxembourg, Luxembourg. From 1996 to 1997, he was the Director of Research with ArrayComm, Inc., a
	start-up in San Jose, CA, USA, based on his patented technology. In 1991, he was appointed Professor of
	signal processing with the Royal Institute of Technology (KTH), Stockholm, Sweden. Dr. Ottersten has
	been Head of the Department for Signals, Sensors, and Systems, KTH, and Dean of the School of
	Electrical Engineering, KTH. He is currently the Director for the Interdisciplinary Centre for Security,
	Reliability and Trust, University of Luxembourg. He is a recipient of the IEEE Signal Processing Society
	Technical Achievement Award, the EURASIP Group Technical Achievement Award, and the European
	Research Council advanced research grant twice. He has co-authored journal papers that received the
	IEEE Signal Processing Society Best Paper Award in 1993, 2001, 2006, 2013, and 2019, and 9 IEEE
	conference papers best paper awards. He has been a board member of IEEE Signal Processing Society,
	the Swedish Research Council and currently serves of the boards of EURASIP and the Swedish
	Foundation for Strategic Research. Dr. Ottersten has served as Editor in Chief of EURASIP Signal
	Processing, and acted on the editorial boards of IEEE Transactions on Signal Processing, IEEE Signal
	Processing Magazine, IEEE Open Journal for Signal Processing, EURASIP Journal of Advances in Signal
	Processing and Foundations and Trends in Signal Processing. He is a fellow of EURASIP.
\end{IEEEbiographynophoto}
\end{document}